%% file: ForQuantum-final-arxiv.tex
\documentclass[a4paper,onecolumn,11pt,accepted=2025-06-16]{quantumarticle}
\pdfoutput=1
\usepackage[utf8]{inputenc}
\usepackage[english]{babel}
\usepackage[T1]{fontenc}
\usepackage{lipsum}
\usepackage{soul}

\usepackage{CJKutf8}
\newcommand{\Label}[1]{\label{#1}}
\usepackage{mymacros}
\begin{document}

\begin{CJK*}{UTF8}{gbsn}
\title{On the composable security of weak coin flipping}

\author{Jiawei Wu (吴家为)}
\affiliation{Centre for Quantum Technologies, National University of Singapore, Singapore}
\orcid{0000-0001-7340-7846}
\email{constchar0212@gmail.com}

\author{Yanglin Hu (胡杨林)}
\affiliation{Centre for Quantum Technologies, National University of Singapore, Singapore}
\orcid{0009-0000-1105-589X}
\email{yanglin.hu@u.nus.edu}

\author{Akshay Bansal}
\affiliation{Department of Computer Science, Virginia Polytechnic Institute and State University, Blacksburg, VA, USA}
\affiliation{Virginia Tech Center for Quantum Information Science and Engineering, Blacksburg, VA, USA}
\email{akshaybansal14@gmail.com}

\author{Marco Tomamichel}
\affiliation{Centre for Quantum Technologies, National University of Singapore, Singapore}
\affiliation{Department of Electrical and Computer Engineering, National University of Singapore, Singapore}
\orcid{0000-0001-5410-3329}
\homepage{https://marcotom.info/}
\email{marco.tomamichel@nus.edu.sg}
\maketitle

\end{CJK*}

\begin{abstract}
    Weak coin flipping is a cryptographic primitive in which two mutually distrustful parties generate a shared random bit to agree on a winner via remote communication.
    While a stand-alone secure weak coin flipping protocol can be constructed from noiseless quantum communication channels, its composability remains unexplored. In this work, we demonstrate that no weak coin flipping protocol can be abstracted as a simple black-box resource with composable security. Despite this, we also establish the overall stand-alone security of quantum weak coin flipping protocols under composition in sequential order.
\end{abstract}

\section{Introduction}

\subsection*{Motivation and background}

Cryptographic systems or protocols are expected to guarantee security even when composed with multiple other cryptographic systems.
To illustrate this, consider a cryptographic protocol designed to perform secure communication between two parties while leaking at most the length of the communicated message.  
One possible way to construct such a scheme is to first use a system that generates and distributes secret keys amongst the users and then use another system that pads the key with the intended message to be sent to the receiver (the so-called one-time pad). While it is relatively straightforward to see that a buggy implementation that compromises the security of key distribution or one-time pad would result in overall insecure communication, 
it is at times difficult to formally argue that the final composition results in a secure communication channel when both the constituent systems (key distribution and one-time pad) are assumed secure in isolation. 
Indeed, it might seem natural to use small accessible information of the distributed key as the security criterion for key distribution. However, the fact that key distribution using this criterion is secure in isolation does not imply that the final communication protocol is secure~\cite{konig2007small}, due to the data locking effect~\cite{divincenzo2004locking}.
We thus say that composability fails for this security criterion.

Several frameworks have been proposed that are intended to analyze compositions of cryptographic systems, starting with Canetti's seminal Universally Composable (UC) security  framework~\cite{canetti2000security,canetti2001universally}. Since then, several extensions and variants have emerged (see, e.g.,~\cite{canetti2007universally,unruh2010universally,canetti2011universally,canetti2015simpler}). We focus here on the model of Abstract Cryptography (AC) by Maurer and Renner~\cite{maurer2011abstract,maurer2012constructive}. 
In contrast to the previous frameworks, the latter is a top-down axiomatic approach and is intended to provide better insights via abstraction of a given complex cryptographic system and is thus most suitable for our purposes. The security in the AC framework is built on simulation-based security. 
Here, for a given system to be deemed secure, it suffices to ensure that all adversarial attacks on the system can be simulated on an ideal system. 
Therefore, to prove the security of any multi-party protocol in the AC framework, one needs to show the existence of such a simulator simulating all possible adversary attacks.

The cryptographic setting we consider in this work is a two-party primitive called weak coin flipping (WCF). 
Informally speaking, in WCF two mutually distrustful parties, say Alice and Bob, need to agree on a binary outcome even though the two parties have opposite preferences. The protocol is deemed secure if it produces a fair coin toss in case both parties are honest and a dishonest party can at most bias the honest party's outcome towards their preference. 
The WCF task is fundamentally different from the other commonly studied tasks such as strong coin flipping (SCF), bit commitment (BC), oblivious transfer, or more generally, secure function evaluation. While the latter protocols are shown to be insecure with an unbounded adversary~\cite{lo1997security, lo1997quantum, mayers1997unconditionally, lo1998why, ambainis2004multiparty}, it is possible to construct weak coin flipping protocols with quantum communication that exhibit almost perfect security with arbitrarily small bias~\cite{mochon2007quantum}. 
The explicit protocols which are shown to have an arbitrarily small bias are developed in~\cite{arora2019quantum, arora2021analytic, arora2022solutions}. 
The security notion employed in these results, known as information-theoretic security, allows an adversary to enjoy unlimited resources (i.e.,\ the adversary has no restrictions on its computational power or use of memory, etc.). 
It is well-studied for a range of different cryptographic tasks and, in most cases, is stronger than other commonly studied security definitions such as computational~\cite{bernstein2008post,damgaard2009quantum,mahadev2023classical}, bounded storage~\cite{damgaard2008crypto,wehner2008composable,unruh2011concurrent} or even game-based~\cite{boneh2020graduate,nowak2007framework}, where an adversary can perform attacks in only some prescribed ways.

Previously, it was shown that WCF enjoys information-theoretic stand-alone security, i.e., it is secure as a single protocol instance in isolation~\cite{mochon2007quantum,arora2019quantum,arora2021analytic,arora2022solutions}. 
While WCF itself is of limited interest, it is a powerful tool in constructing optimal protocols for several other primitives such as bit commitment~\cite{chailloux2011optimal} and strong coin flipping~\cite{chailloux2009optimal}.
However, such constructions would require more than stand-alone security, and it is still unclear if WCF protocols have composable security when combined with other cryptographic protocols.

\subsection*{Contributions}

Our main contributions encompass two negative theorems and one positive theorem. 
In \Cref{thm:wcf-not}, we demonstrate the impossibility of constructing a typical  WCF resource (see \Cref{def:wcf-res}) analogous to the widely-used SCF resource~\cite{demay2013unfair}  from noiseless communication channels in the AC framework. 
Generally speaking, a resource is the abstraction of certain cryptographic functionality (see Section~\ref{sec:acf}).
This WCF resource generates a random bit for both parties, allowing the cheating party to learn the bit in advance and bias it towards their preferred choice.  
It was shown that the widely-used SCF resource cannot be achieved in the AC framework~\cite{vilasini2019composable}, but the result for the similar WCF resource remains unknown until this work.
We state \Cref{thm:wcf-not} here informally. 

\begin{informaltheorem}[Informal]
    The typical WCF resource in \Cref{def:wcf-res} with composable security cannot be constructed from noiseless communication channels in the AC framework. 
\end{informaltheorem}

To prove \Cref{thm:wcf-not}, we show that WCF protocols using noiseless communication channels are vulnerable to the man-in-the-middle attack when composed in parallel. 
As a result, the typical WCF resource cannot be constructed by any WCF protocol from noiseless communication channels in the AC framework. 

While \Cref{thm:wcf-not} shows the impossibility of the typical WCF resource in the AC framework, we can still define an alternative WCF resource at the cost of increasing the interaction complexity.
Nevertheless, such an alternative WCF resource that can be constructed from noiseless communication channels must be extremely complicated in terms of its interaction complexity, as is shown in \Cref{thm:general}. Replacing an actual protocol with an ideal resource intends to simplify the analysis; therefore, it is meaningless to study such an alternative WCF resource. Here is the informal statement of \Cref{thm:general}.
\begin{informaltheorem}[Informal]
    Any WCF resource with composable security constructed from noiseless communication channels in the AC framework must have high interaction complexity. 
\end{informaltheorem}

The proof of \Cref{thm:general} consists of two steps.  
First, we show that an alternative WCF resource with composable security constructed from noiseless communication channels implies a WCF protocol with the same interaction complexity as the alternative WCF resource. 
Then, together with the exponential lower bound on the interaction complexity of WCF protocols~\cite{miller2020impossibility}, we conclude that a WCF resource with a low interaction complexity cannot exist. 

Next, we extend the classical version of the  composition theorem for sequentially ordered protocols~\cite[Theorem 7.4.3] {goldreich2004foundations} to the quantum case.
We find a sufficient condition called global security where WCF and other modules can be composed into a larger system with stand-alone security. 
A composition between WCF and other components is globally secure if there is no correlation between the outcomes of WCF and other components. 
We thus show in \Cref{thm:cond-gs} that global security holds when WCF and all other components are composed in a strict causal order. 
That is, there is no other component running concurrently with WCF. 
The constructions for optimal SCF and optimal BC exactly satisfy global security.

\begin{informaltheorem}[Informal]
    A primitive constructed from WCF with stand-alone security under strictly ordered sequential composition maintains its stand-alone security. 
\end{informaltheorem}

The idea to prove \Cref{thm:cond-gs} is by contradiction.
Assume that there is some correlation between outcomes of WCF with stand-alone security and other components under ordered sequential composition. 
Using the information carried from previous components to WCF, we can construct an attack breaking the stand-alone security of WCF. 
This contradicts our assumption and thus proves \Cref{thm:cond-gs}.

\subsection*{Related work}

\textit{Causality in the composable framework.}
Some early works~\cite{micali1992secure,beaver1992foundations} developed the paradigm that only allows for sequential-order composition, where protocols are executed one after another.
Canetti~\cite{canetti2000security} relaxed sequential composition to modular composition in the UC framework, where the former is a special case of the latter. 
The AC framework~\cite{maurer2011abstract,maurer2012constructive} takes a top-down approach. Only the connection of abstract systems is considered at the top level, and the execution order is left at the lower level specification.
Under the AC framework, the causal box model~\cite{portmann2017causal} provides the causal order specification that allows for a quantum superposition of different message orders.

\textit{Concerning concurrent composition.}
Literature \cite{wehner2008composable, unruh2011concurrent} addresses the difficulty of concurrent parallel composition in the UC framework and the bounded-storage model.
The composability of BC was discussed in~\cite{kaniewski2013secure}. 
They define the weak binding security of BC in a stand-alone rather than composable manner and propose a counterexample where the string commitment constructed from multiple instances of BC is insecure.
In a previous model of general cryptographic protocols \cite[Section 7.3]{goldreich2004foundations}, concurrent calls to an oracle (resource) are also forbidden. This restriction primarily stems from the formulation of syntax in that model, rather than an explicit security concern.
Recent work by Batra et al.~\cite{batra2024robust} provides device-independent protocols for oblivious transfer and bit commitment proven secure under sequential composition, where all queries to the ideal functionality must be performed sequentially, but their concurrent composable security remains unknown.

\textit{Coin flipping within the AC framework.}
To study the composability of SCF, Demay and Maurer~\cite{demay2013unfair} propose an SCF resource and prove that it can be constructed from a BC resource.
However, this resource is shown to be impossible to construct from noiseless communication channels, even in a relativistic setting~\cite{vilasini2019composable}.


\section{Preliminaries}
\subsection{The abstract cryptography framework}
\label{sec:acf}
The building blocks of the framework are resources, converters (protocols), and distinguishers. Resources are interactive systems that abstract cryptographic functionalities, and converters transform one resource into another. Distinguishers distinguish resources from each other.
The framework can describe multiparty functionalities, but in this work, we only focus on two-party ones. 

\begin{itemize}
    
\item A \emph{resource}, denoted by $R$, has 
one
interface accessible to each party. Each party can send inputs, receive outputs, perform operations and give instructions via the corresponding interface. In two-party cryptography, it is necessary to consider three different behaviors $R$, $R_A$ and $R_B$ of a resource, in which $R$, $R_A$ and $R_B$ correspond to the case where both are honest, where only Alice is dishonest, and where only Bob is dishonest, respectively. We use $\ca{R} = (R,R_A,R_B)$ to compactly denote a resource.

\item A \emph{converter}, denoted by $\alpha$, is a cryptographic system with an inner interface that connects to a resource and an outer interface that connects to a player.
The set of all converters is denoted by $\Sigma$. 
Each party can instruct the converter on the outer interface to execute specific tasks on the inner interface. 
A converter $\alpha$ can convert one resource $R$ into another resource $S$ by connecting its inner interface to one interface of the resource $R$. 
The interfaces of the resource $S$ consist of the unconnected interface of the resource $R$ and the outer interface of the converter $\alpha$. 
We denote the connection between the converter $\alpha$ and the left-side (right-side) interface of the resource $R$ by $\alpha R$ ($R \alpha$).

\item A \emph{protocol} $\pi=(\pi_A,\pi_B)$ consists of Alice's converter $\pi_A$ and Bob's converter $\pi_B$. 
By connecting each converter to the corresponding interface of a resource $R$, the protocol $\pi$ can transform $R$ into the required resource $S$, provided that $S= \pi_A  R \pi_B$.
For example, suppose $R$ is a binary symmetric channel resource and $(\pi_{A}$,$\pi_{B})$ is an encoder-decoder pair corresponding to a proper error-correcting code.  
Then one can construct a noiseless channel resource as $S = \pi_{A} R \pi_{B}$. 

\item A \emph{distinguisher}, denoted by $\ca{D}$, is an agent who can access both interfaces of a resource. Given one of two resources $R$ and $S$ uniformly at random, a distinguisher distinguishes whether the resource is $R$ or $S$ by interacting with it on both interfaces. The distinguishing advantage of the above task is defined as  
\begin{align}
    d(R,S) = \sup_{\mathcal{D}}|\Pr[\mathcal{D}(R)=0] - \Pr[\mathcal{D}(S)=0]|.
\end{align}
The distinguishing advantage is a pseudo metric between two resources, satisfying the symmetry and triangle inequality properties. 
If $d(R,S) \leq \delta$, we also write $R \approx_\delta S$. 
\end{itemize}

With these building blocks, we define composable security for two-party functionalities. 
\begin{definition}[Composable security] \label{def:coms}
    A protocol $\pi = (\pi_A,\pi_B)$ constructs resource $\ca{S} = (S, S_A,S_B)$ out of resource $\ca{R} = (R,R_A,R_B)$ with distinguishing advantage $\delta$ if the following conditions are satisfied: 
    \begin{enumerate}
        \item Correctness
            \begin{align}
                S \approx_\delta \pi_A R \pi_B; 
            \end{align}
        \item Security against dishonest Bob
            \begin{align}
               \exists \sigma_B \in \Sigma, S_{B} \sigma_B \approx_\delta \pi_A R_{B}; \label{eq:dB}
            \end{align}
        \item Security against dishonest Alice
            \begin{align}
               \exists \sigma_A \in \Sigma, \sigma_{A} S_{A}  \approx_\delta R_{A} \pi_{B}. 
            \end{align}
    \end{enumerate}
    This construction is denoted by $\ca{R} \overset{\pi,\delta}{\longrightarrow}\ca{S}$. 
    Alternatively, we say that $\ca{S}$ is an abstraction of $\pi$ based on $\ca{R}$, or that $\ca{S}$ abstracts $\pi$ based on $\ca{R}$. 
\end{definition}

In \Cref{def:coms}, $\sigma_{A}$ and $\sigma_{B}$ are called simulators. These simulators enforce a critical security property: any cheating strategy against the real resource $\cR$ can be translated via the simulators into a corresponding cheating strategy against the ideal resource $\cS$. 
For instance, if dishonest Bob deviates from $\pi_{B}$ and follows $\pi'_{B}$ instead, then, we have $S_B \sigma_B \pi'_{B}\approx_\delta\pi_A R_B \pi_B'$ by \Cref{eq:dB}, which means  Bob’s strategy $\pi_B'$ in the real world (interacting with $\pi_A R_B$) is simulated by the $\sigma_B \pi'_{B}$ in the ideal world (interacting with $S_B$).
For a detailed description of simulators and the related security notions in cryptography, we refer to the tutorial~\cite{lindell2017simulate}.

When $\ca{R}$ is a noiseless communication channel, the conditions can be simplified as 
\begin{align}
    &S  \approx_\delta \pi_A\pi_B\label{eq:correctness}, \\
     &\exists \sigma_B \in \Sigma, S_B \sigma_B  \approx_\delta \pi_A\label{eq:cheatingBob}, \\
     &\exists \sigma_A \in \Sigma, \sigma_A S_A \approx_\delta \pi_B\label{eq:cheatingAlice} . 
\end{align}

\subsection{The quantum comb model}
\label{sec:qcm}
A quantum comb~\cite{gutoski2007toward,chiribella2009theoretical} is a general quantum information processing system with internal memory that processes incoming messages sequentially.
An example of a quantum comb is shown in \Cref{fig:comb}.
At the lower level specification of the AC framework where the execution order is taken into account, building blocks are commonly modeled as quantum combs.
A protocol which can be modeled within the quantum comb model (QCM) is called a QCM protocol.
All existing WCF protocols can be classified as QCM protocols.

\begin{figure}[htbp!]
    \centering
    \input{pic/comb}
    \caption{An example of quantum comb. The lines represent quantum registers, and $\cE_i$ is a quantum operation. Each quantum operation interacts with the environment through an input-output pair, resembling the teeth of a comb.}
    \label{fig:comb}
\end{figure}

The interaction complexity of a QCM protocol $\pi$, denoted by $\mathsf{Intcom}(\pi)$, is defined as the communication rounds it requires.
As an example, for a protocol $\pi$ modeled by \Cref{fig:comb}, $\mathsf{Intcom}(\pi) =  n$.
For resources $S_A$ and $S_B$, their interaction complexity refers to the communication rounds at the adversary's side. 
For example, the interaction complexity of WCF resource $S_A$ in \Cref{fig:ideal-wcf}(b) is the number of rounds on the left side, which is $1$.

\subsection{The causal box model}

The causal box model~\cite{portmann2017causal} is a more general specification of building blocks in the AC framework.
The causal box model describes the causal structure of the messages exchanged by these building blocks. 
The quantum comb model can be treated as a special case of the causal box model.

In the causal box model, the abstract blocks are modeled as causal boxes.
A causal box is an information processing system with input and output wires.
Each message on a wire is described by a quantum state on space $\ca{H}\otimes l^2(\cT)$, where $\cT$ is a partially ordered set (poset) that defines the message ordering, and $l^2(\cT)$ is the set of all functions $\cT \to \mathbb{C}$ with bounded $L^2$-norm. 
Note that $\cT$ is not necessarily interpreted as time, as time is totally ordered. 
In general, a wire may carry multiple messages. 
Therefore, the wire space is defined as the Fock space of the message space,
\begin{align}
    \cF_{\ca{H}}^{\cT} \coloneqq \bigoplus_{n=0}^{\infty} \vee^n \left( \ca{H} \otimes l^2(\cT) \right),
\end{align}
where $\vee^n \ca{H}$ denotes the symmetric subspace of $\ca{H}^{\otimes n}$.
Here, $\cF_{\cH}^{\cT}$ is also a Hilbert space.
For convenience, we use the shorthand notation $\cF_{A}^{\cT}$ for $\cF_{\ca{H}_A}^{\cT}$.
With this definition, wires can be merged or split, i.e., for $\ca{H}_A = \ca{H}_{A_1} \oplus \ca{H}_{A_2}$,
\begin{align} \label{eq:split}
    \cF_{A}^{\cT} = \cF_{A_1}^{\cT} \otimes \cF_{A_2}^{\cT}.
\end{align}
Formally, a causal box is defined as a completely positive and trace-preserving (CPTP) map that transforms the input wires to the output wires.
This map should further respect the causality condition, i.e., the output message at $t' \in \cT$ only depends on the input messages at $t\in \cT$ where $t\preceq t'$. A precise definition can be found in \Cref{app:causal-box}. 

The composition of causal boxes consists of parallel composition and loops. 
The parallel composition of causal boxes $\Phi$ and $\Psi$, denoted by $\Phi \| \Psi$, is defined as the map $\Phi \otimes \Psi$.
By \Cref{eq:split}, the input and output spaces of $\Phi \otimes \Psi$ remain valid wires. 
A loop connects segments with the same dimension of the input and output wires. 
An interface labeled by $x$ for a causal box $\Phi$ is denoted by $\mathsf{int}_{x}(\Phi)$, which consists of a set of wires, possibly including both input and output wires.
For example, the inner and outer interfaces of a converter $\alpha$ are denoted by $\mathsf{int}_{in}(\alpha)$ and $\mathsf{int}_{out}(\alpha)$, respectively.
Two interfaces $\mathsf{int}_a(\Phi)$ and $\mathsf{int}_b(\Psi)$ are compatible if there exists a bijection $P: \mathsf{int}_a(\Phi) \leftrightarrow \mathsf{int}_b(\Psi)$ such that each input wire of $\mathsf{int}_a(\Phi)$ corresponds to an output wire of $\mathsf{int}_b(\Psi)$ with the same dimension and vice versa.

\subsection{Weak coin flipping}
Here, we formulate WCF protocols with causal boxes. 

\begin{definition}[Weak coin flipping protocol]
    A weak coin flipping protocol using noiseless communication channels is a pair of converters $\pi= (\pi_A,\pi_B)$ with compatible inner interfaces, where $\pi_A$ ($\pi_B$) has one output wire on the outer interface and only outputs a single classical bit at $t_a$ ($t_b$). 
    The poset $\cT$ of the protocol is bounded both from above and below, i.e., $\exists t_0, t_a, t_b \in \cT$, such that $\forall t\in \cT$, it holds that $t_0 \preceq t$, $t \preceq t_a$ and $t \preceq t_b$.
\end{definition}

\begin{figure}[htbp!]
    \centering
    \input{pic/CFConverters}
    \caption{A general model for a coin flipping protocol using noiseless communication channels. The message spaces $X,Y$ and $X',Y'$ are of the same dimension, respectively.}
    \label{fig:CFConverters}
\end{figure}

In a WCF protocol, we say Alice wins if $c_B = 0$, while Bob wins if $c_A = 1$. In this case, an honest Bob only needs to prevent a cheating Alice from biasing $c_B$ towards $0$ and an honest Alice only needs to prevent a cheating Bob from biasing $c_A$ towards $1$.

\begin{definition}[Stand-alone security of WCF] \label{def:wcf-stand-alone}
    Let $z\in [0,\frac{1}{2}]$ and $\epsilon\leq z$. A coin flipping protocol $\pi=(\pi_A, \pi_B)$ is a $z$-unbalanced $\epsilon$-biased WCF protocol (with stand-alone security), denoted by $\mathrm{WCF}(z,\epsilon)$, if the following conditions are satisfied: 
\begin{enumerate}[start=1, label={(\bfseries S\arabic*)}]
    \item The composition $\pi_A \pi_B$ by connecting the inner interface of $(\pi_A,\pi_B)$ outputs $(c_A, c_B)$ with a probability of
    \begin{align}
        \Pr[c_A = c_B = 0] = z, \quad \Pr[c_A = c_B = 1] = 1- z. 
    \end{align}
    \item For any converter $\alpha$ compatible with the inner interface of $\pi_B$, the composition $\alpha \pi_B$ outputs $c_B=0$ with a probability of
    \begin{align}
        \Pr[c_B = 0] \leq z + \epsilon.
    \end{align}
    \item For any converter $\beta$ compatible with inner interface of $\pi_A$, the composition $\pi_A \beta$ outputs $c_A=1$ with a probability of
    \begin{align}
        \Pr[c_A = 1] \leq 1 - z + \epsilon.
    \end{align}
\end{enumerate}

We call a protocol balanced if $z= \frac{1}{2}$ and unbiased if $\epsilon=0$. 
\end{definition}

\section{Weak coin flipping is not universally composably secure} \label{sec:neg}

This section shows that WCF cannot be universally composably secure in the AC framework.
The result contains two parts. 
First, we consider a specific WCF resource derived from a simple but widely adapted SCF resource~\cite{demay2013unfair} and show that this specific WCF resource cannot be constructed from a noiseless communication channel.
Nevertheless, this result does not exclude the existence of other more complicated resources that can be constructed with, say, the protocol in~\cite{arora2021analytic}.
Then our second result shows that, for any WCF protocol with enough stand-alone security, one cannot abstract it into a useful black-box resource, where usefulness means that the resource captures the full functionality of WCF while maintains a low interaction complexity.

\subsection{Impossibility for a specific resource}

We consider a reasonable and simple WCF resource in~\Cref{fig:ideal-wcf} which is largely motivated from the strong coin flipping resource described in~\cite[Section III]{demay2013unfair}\footnote{
Although the resource had origins in Blum’s protocol~\cite{manuel1981coin} that constructs an \emph{unfair} coin flip (where a party is allowed to abort on observing the outcome), it serves as an abstraction for a large family of SCF protocols. 
When using a similar protocol here, it is equivalent to a \emph{biased} coin flip if the honest party declares themselves the winner (locally) whenever the dishonest party aborts~\cite{demay2013unfair}. Therefore, it is unnecessary to consider the case where players abort.}. 
In the WCF task, a dishonest party wants to force only one of the two outcomes.
To reflect this idea, we model the possible malicious operation as follows: (1) a cheating Alice (Bob) intercepts the random bit $c'$ ($c''$) which ought to be the honest output in advance and (2) based on $c'$ ($c''$), Alice (Bob) sends $b'$ and $p'$ ($b''$ and $p''$) to the resource to either resign by setting $p'=1$ ($p''=1$) or replace $c'$ with $b$ with probability $\frac{\epsilon}{1-z}$ ($\frac{\epsilon}{z}$) by setting $p'=0$ ($p''=0$). 
This establishes~\Cref{def:wcf-res} as a tuple of ideal WCF resources that allows for a maximum bias of $\epsilon$ for different scenarios.

\begin{figure}[h]
    \centering
    \input{pic/WCFres2}
    \caption{(a), (b) and (c) respectively show the components $S,S_A,S_B$ of the WCF resource $\ca{S}$. Every message is assigned an element of the order set to indicate the causal order.}
    \label{fig:ideal-wcf}
\end{figure}

\begin{definition}[$z$-unbalanced $\epsilon$-biased WCF resource] \label{def:wcf-res}
    Let $z\in[0,\frac{1}{2}]$ and $\epsilon\in[0, z]$. A $z$-unbalanced $\epsilon$-biased WCF resource is characterized by a tuple $\ca{S} = (S,S_A,S_B)$, as is shown in \Cref{fig:ideal-wcf}. 
    \begin{enumerate}
    \item $S$ describes the case when both Alice and Bob are honest. 
    In this case, both parties obtains a random bit $c$ according to the probabilities
    \begin{align}\label{eqn:ideal-prob}
        \Pr[c = 0] = z,\quad \Pr[c = 1] = 1 - z. 
    \end{align}
    \item $S_{A}$ describes the case when Alice is dishonest and Bob is honest. 
    $S_A$ outputs a random bit $c'$ subject to the probability distribution $(z,1-z)$ 
    at $t_0'$ and receives two bits $b',p'$ at $t_1'$ at the left interface. 
    Then it outputs a bit $c_B$ at $t_b$ at the right interface. 
    If $p'=1$, then $c_B=1$, else $c_B=b'$ with probability $\frac{\epsilon}{1-z}$ and $c_B = c'$ with probability $1- \frac{\epsilon}{1-z}$. The causality condition is $t_0' \prec t_1' \prec t_b$.
    \item $S_B$ describes the case when Alice is honest and Bob is dishonest.
    $S_B$ outputs a random bit $c''$ subject to the probability distribution $(z,1-z)$ at $t_0''$ and receives two bits $b'',p''$ controlled by Bob at $t_1''$ at the right interface. Then it outputs a bit $c_A$ at $t_a$ at the left interface. If $p''=1$, then $c_A=0$, else $c_A=b''$ with probability $\frac{\epsilon}{z}$ and $c_A = c''$ with probability $1-\frac{\epsilon}{z}$.
    The causality condition is $t_0'' \prec t_1'' \prec t_a$.
\end{enumerate}
\end{definition}
To elaborate on the functionality when one party is malicious, consider an example where cheating Alice tries to bias the output $c_B$ of $S_A$ towards $0$.
For this purpose, she needs to set $p'=0$ regardless of the value of $c'$. Then the probability of her successfully forcing $c_B=0$ is
\begin{equation*}
    \Pr[c_B = 0] = \Pr[c'=0]\cdot \Pr[c_B=0|c'=0] + \Pr[c'=1]\cdot \Pr[c_B=0|c'=1].
\end{equation*}
As $p' = 0$, the optimal strategy for Alice is to set $b' = 0$ since $c_B$ is either $c'$ or $b'$.
If $c'=0$, which happens with probability $\Pr[c'=0] = z$, then
$c_B=0$ for certain, i.e., $\Pr[c_B=0|c'=0] = 1$. 
If $c'=1$, which happens with probability $\Pr[c'=1] = 1-z$, Alice's conditional winning probability equals the probability that $c_B$ takes the value of $b'$, i.e., $\Pr[c_B=0|c'=1] = \frac{\epsilon}{1-z}$. 
We thus have
\begin{equation}
    \Pr[c_B = 0] = \Pr[c'=0]  \cdot 1 + \Pr[c'=1 ]\cdot \frac{\epsilon}{1-z} = z  + \epsilon.
\end{equation}
Similar arguments also hold for cheating Bob.

The following theorem indicates that the resource in \Cref{def:wcf-res} cannot be constructed from noiseless communication channel.

\begin{theorem} \label{thm:wcf-not}
    For $z \in [0,\frac{1}{2}]$ and $\epsilon \in [0,z]$, there does not exist a protocol $\pi=(\pi_A,\pi_B)$ which can construct an  $z$-unbalanced $\epsilon$-biased WCF resource $\mathcal{S} = (S,S_A,S_B)$ with a distinguishing advantage $\delta < \frac{1}{3} \left( z(1-z) - \epsilon + \epsilon \min\{ 2\epsilon, z\} \right)$. 
\end{theorem}
\begin{proof}
    Suppose that there exists a protocol $\pi=(\pi_A,\pi_B)$ that constructs an ideal $z$-unbalanced $\epsilon$-biased WCF resource $\mathcal{S} = (S,S_A,S_B)$ with a distinguishing advantage $\delta$ (without loss of generality, we assume $z\in [0, \frac{1}{2}]$). Combining the security conditions for Alice~\eqref{eq:cheatingBob} and Bob~\eqref{eq:cheatingAlice} using triangle inequality, we obtain 
    \begin{align}
        \exists \sigma \in \Sigma, S_B \sigma S_A \approx_{2\delta} \pi_A \pi_B, 
    \end{align}
    where $\sigma = \sigma_B\sigma_A $. 
    Once again, combining the previous equation with the correctness condition~\eqref{eq:correctness} using triangle inequality, we get
    \begin{align} 
        \exists \sigma \in \Sigma, S_B \sigma S_A \approx_{3\delta} S,  
    \end{align}
    or equivalently,
    \begin{align}
        \exists \sigma \in \Sigma, d(S_B \sigma S_A , S) \le 3 \delta. \label{eq:AC}
    \end{align}
    This relation is illustrated in \Cref{fig:com-wcf}.
    \begin{figure}[h]
        \centering
        \input{pic/WCFcom2}
        \caption{Illustration for Inequality \eqref{eq:AC}.}
        \label{fig:com-wcf}
    \end{figure}
    
    We next establish a constant lower bound on such $\delta$ (constant in terms of $z$) by deriving a lower bound on the distinguishing advantage $d(S_B \sigma S_A , S)$ for any $\sigma$, i.e., $\min_{\sigma} d(S_B \sigma S_A , S)$. 
    As the two systems $S_B \sigma S_A$ and $S$ in~\Cref{fig:com-wcf} do not need any external input, $d(S_B \sigma S_A , S)$ is simply the total variation distance between their output distributions, denoted by $P_{c_A c_B}$ and $P_{cc}$ respectively, i.e.,
    \begin{align}
        d(S_B \sigma S_A, S) = \frac{1}{2} \|P_{c_A c_B} - P_{cc} \|_1,
    \end{align}
    with $P_{cc}(0,0) = z, P_{cc}(1,1) = 1-z$.

    The best choice of causal order for $\sigma$ is $t_0' \prec t_1', t_0'\prec t_1'', t_0''\prec t_1', t_0''\prec t_1''$ since it can maximize the use of information of $c',c''$. Then, the minimum distinguishing advantage can be bounded as follows:
    \begin{equation} \label{eq:disadv}
    \begin{aligned}
        \min_{\sigma} d(S_B \sigma S_A, S) & = 
        \min_{\sigma} \frac{1}{2} \left( \Pr[c_A = 0, c_B =1] + 
        \Pr[c_A = 1, c_B =0] \right.\\
        & \quad\quad  +  \left.  \big\lvert
        \Pr[c_A = 0, c_B =0] - z \big\rvert + 
        \big\lvert \Pr[c_A = 1, c_B =1] - 1+z \big\rvert \right) \\
        & \ge \min_{\sigma} \left(\Pr[c_A = 0, c_B =1] + 
        \Pr[c_A = 1, c_B =0] \right)\\
        & = \min_{\sigma} \Pr[c_A \ne c_B]\\
        & \geq \min_{\sigma} \Pr[c' = 1, c'' = 0] \cdot \Pr[c_A \ne c_B | c' = 1, c'' = 0] \enspace \\ 
        & \quad  +  \min_{\sigma} \Pr[c' = 0, c'' = 1] \cdot \Pr[c_A \ne c_B | c' = 0, c'' = 1] \enspace\\
        &  \quad +  \min_{\sigma} \Pr[c' = c''] \cdot \Pr[c_A \ne c_B | c' = c''].\\
    \end{aligned}
    \end{equation}
    In the previous equation, the optimal distribution on $(b',p',b'',p'')$ determines the optimal $\sigma$ for each received value $c',c''$. The optimal distributions for different choices of $c'$ and $c''$ are discussed next. 
    \begin{enumerate}
        \item\label{item:equalInputs} If $c'=c''=0$ or $c'=c''= 1$, then the box $\sigma$ outputs $p'=0, b'=c', p''=0, b''=c''$ to keep $c_A=c''=c'=c_B$, implying that $\min_{\sigma} \Pr[c_A \ne c_B | c' = c''] = 0$. 
        \item\label{item:zeroOneInputs} If $c'=0, c''=1$, then the box $\sigma$ outputs $p'=1, p''=0$ to force the outcome $c_A=c_B=1$, implying that $\min_{\sigma} \Pr[c_A \ne c_B | c' = 0, c'' = 1] = 0$.  
        \item\label{item:oneZeroInputs} If $c'=1, c''=0$, then we can assume the box $\sigma$ outputs $p'=p''=0$ because setting $p'=1$ ($p''=1$) has the same effect as setting $p'=0, b'=c'$ ($p''=0, b''=c''$).
        Consider the event $c_A\ne c_B$ conditioned on different cases of $(b',b'')$:
        \begin{equation}\label{eq:cAcBOverStrategies}
            \begin{aligned}
            &\Pr[c_A \ne c_B | c'=1, c''=0, b'=0,b''=0] = 1 - \frac{\epsilon}{1-z}; \\
            &\Pr[c_A \ne c_B | c'=1, c''=0, b'=1,b''=0] = 1;\\
            &\Pr[c_A \ne c_B | c'=1, c''=0, b'=0,b''=1] = 1 - \frac{\epsilon(1-2\epsilon)}{z(1-z)};\\
            &\Pr[c_A \ne c_B | c'=1, c''=0, b'=1,b''=1] = 1 - \frac{\epsilon}{z}.
            \end{aligned}
        \end{equation}
        All possible cheating strategies are in the convex hull of the above four cases, then we have  
        \begin{align}\label{eq:cAcBOverInputs}
             &\min_{\sigma} \Pr[c_A \ne c_B | c'=1,c''=0]\\  = & \min_{P_{b',b''}} \sum_{i',i''\in \{0,1\}}  P_{b',b''}(i',i'') \cdot \Pr[c_A \ne c_B | c'=1, c''=0, b'=i',b''=i''] \\
             =& \min_{i',i''\in \{0,1\}} \Pr[c_A \ne c_B | c'=1, c''=0, b'=i',b''=i''] \\
            =& \begin{cases}
            1 - \frac{\epsilon (1 - 2\epsilon)}{z(1-z)}, & \text{if } 0 \le \epsilon < z/2\\
            1 - \frac{\epsilon }{z},  & \text{if } z/2 < \epsilon \le z
            \end{cases}. \label{eq:minpr}
        \end{align}
        where the conditional probabilities $\Pr[c_A \ne c_B | c'=1, c''=0, b'=i',b''=i'']$ are given in~\Cref{eq:cAcBOverStrategies}, and \Cref{eq:minpr} holds because $0\le \epsilon \le z \le 1/2$.
    \end{enumerate}

    By the above discussion, \Cref{eq:disadv} continues as 
    \begin{align}
        \min_{\sigma} d(S_B \sigma S_A, S) & \geq 
        \min_{\sigma} \Pr[c'=1,c''=0] \cdot \Pr[c_A \ne c_B | c'=1,c''=0] \label{eq:AAS} \\
        &=z(1-z) - \epsilon + \epsilon \min\{ 2\epsilon, z\}, \label{eq:lowrBoundDistinguisher}
    \end{align}
    where \Cref{eq:AAS} follows from~\Cref{eq:disadv} and the implications of \Cref{item:equalInputs} and \Cref{item:zeroOneInputs}. \Cref{eq:lowrBoundDistinguisher} simply follows from \Cref{eq:minpr}.
    Finally, combining \Cref{eq:AC} and \Cref{eq:lowrBoundDistinguisher}, we have
    \begin{align}
        \delta \geq \frac{1}{3} \left( z(1-z) - \epsilon + \epsilon \min\{ 2\epsilon, z\} \right),
    \end{align} 
    which implies that $\pi$ cannot construct $\ca{S}$ with a vanishing distinguishing advantage.
\end{proof}

Setting $\epsilon = 0, z=\frac{1}{2}$, we conclude that there does not exist a protocol $\pi$ that constructs an ideal balanced unbiased WCF resource $\ca{S}$ with a distinguishing advantage $\delta < \frac{1}{12}$. 

\subsection{Impossibility for a general resource}

Composable security is defined by requiring that an actual protocol emulates an ideal resource within a certain accuracy. As established in \Cref{thm:wcf-not}, no actual protocol can approximate the ideal WCF resource defined in \Cref{def:wcf-res}. However, this impossibility result pertains only to this specific ideal WCF resource. By redefining the ideal WCF resource, potentially introducing additional structure, a protocol that securely constructs this modified ideal resource may still exist at the cost of a more complicated and less useful ideal WCF resource.

To address this possibility, it is crucial to determine what properties a WCF resource must possess to be useful in general.
For a WCF resource $\ca{S} = (S,S_A,S_B)$, we propose the following conditions:
\begin{itemize}
    \item \textbf{Fully expressing}. A  resource $\ca{S}$ for $\mathrm{WCF}(z,\epsilon)$ is fully expressing if 
    \begin{enumerate}
        \item $S$ outputs the same value $c$ at each side subject to distribution: $\Pr[c=0] = z, \Pr[c=1] = 1-z$.
        \item For any converter $\alpha$, the composition $\alpha S_A$ outputs bit $0$ at the right interface with probability at most $z+\epsilon$.
        \item For any converter $\beta$, the composition $S_B \beta$ outputs bit $1$ at the left interface with probability at most $1-z+\epsilon$.
    \end{enumerate}    
    \item \textbf{Simple}. 
    Resource simplicity refers to abstraction of non-essential details of a protocol, thereby reducing its interaction complexity. This simplification facilitates the analysis of the resource in a more general context.
    Here we take the interaction complexity $\mathsf{Intcom}(S_A)$, $\mathsf{Intcom}(S_B)$ (defined in \Cref{sec:qcm}) to quantify it.
    Although the idea of a simple resource is rather subjective, we state a resource to be simple if its net interaction complexity ($\mathsf{Intcom}(S_A) + \mathsf{Intcom}(S_B)$) is of $\mathcal{O}(1)$.
\end{itemize}

Next, we will show that a fully expressing and simple WCF resource cannot be constructed from noiseless communication channels.
Specifically, \Cref{thm:general} states that, if a fully expressing resource for $\mathrm{WCF}(z,\epsilon)$ with small $\epsilon$ can be constructed from noiseless communication channels, then its interaction complexity cannot be small.

\begin{theorem} \label{thm:general}
    Let $\pi$ be a QCM $\mathrm{WCF}(\frac{1}{2},\epsilon)$ protocol, and $\pi$ constructs a fully expressing resource $\ca{S}=(S,S_A,S_B)$.
    Then the interaction complexity of $\ca{S}$ satisfies
     \begin{align}
         \min \{\mathsf{Intcom}(S_A), \mathsf{Intcom}(S_B)\} \ge \exp\left( \Omega(1/\sqrt{\epsilon})\right).
     \end{align}
\end{theorem}
To prove \Cref{thm:general}, we need a proposition rewritten from~\cite[Theorem 8.2]{miller2020impossibility}.
\begin{proposition} \label{prop:wcf-comp}
    If $\pi$ is a QCM $\mathrm{WCF}(\frac{1}{2},\epsilon)$ protocol, then
    \begin{align}
        \mathsf{Intcom}(\pi) \ge \exp\left(\Omega(1/\sqrt{\epsilon})\right).
    \end{align}
\end{proposition}

\begin{proof}[Proof of \Cref{thm:general}]
    Without loss of generality, we can assume $\mathsf{Intcom}(S_B) \le \mathsf{Intcom}(S_A)$. 
    Since $\mathcal{C} \overset{\pi}{\longrightarrow} \mathcal{S}$, there exist converters $\sigma_A,\sigma_B$ such that $\pi_A\pi_B \approx S$, $\pi_A \approx S_B\sigma_B$, and $\pi_B \approx \sigma_A S_A$.
    Construct a new protocol $\pi' = (\pi'_A,\pi'_B)$ with $\pi_A' =  S_B, \pi_B' = \sigma_B \sigma_A S_A$.
    Because resource $\ca{S}$ is fully expressing, we conclude that $\pi'$ is also a $\mathrm{WCF}(\frac{1}{2},\epsilon)$ protocol.
    By \Cref{prop:wcf-comp}, 
    \begin{align}
     \mathsf{Intcom}(S_B) = \mathsf{Intcom}(\pi') \ge \exp\left(\Omega(1/\sqrt{\epsilon})\right).
     \end{align}
\end{proof}

Some parts of the argument in the above proof are of more general interest. 
That is, if there exists a resource $\mathcal{S}$ that abstracts a two-party protocol $\pi$, it implies the existence of another protocol $\pi'$ with the same interaction complexity as the resource $\mathcal{S}$. 
Furthermore, if the resource reduces the interaction complexity of $\pi$,
then $\pi'$ is a better protocol than $\pi$ in terms of interaction complexity.
An example is delegated quantum computation (DQC), where the client aims to delegate computation to the server while maintaining privacy. 
The $S_A$ (against dishonest client) of DQC resource~\cite{dunjko2014composable} abstracting the protocol in~\cite{broadbent2009universal, fitzsimons2017private} has no interaction on either side.
This implies the existence of a simple DQC protocol with no interaction, wherein the client performs all the computation herself, ensuring both correctness and blindness. 
\section{Security of weak coin flipping under composition}
\Cref{sec:neg} highlights the difficulties when formulating the security of WCF in the AC framework.
Fortunately, despite these difficulties, there is a prospect of attaining positive outcomes regarding the composability of such protocols.
This section starts with the limitation of the stand-alone security definition (\Cref{def:wcf-stand-alone}) in the context of composition.
Specifically, for the construction of unbalanced WCF from a balanced one, we propose a condition (Inequality~\eqref{eq:cond-ubwcf}) which, while implicit in~\cite{chailloux2009optimal}, is essential in the security proof of the unbalanced WCF protocol depicted there.
Subsequently, we introduce global security as an extension of the stand-alone definition to encompass possible interdependency between the outputs of the WCF protocol and other components in a larger system.
Furthermore, we establish the condition under which global security is maintained. 

\subsection{Stand-alone security is not sufficient}
For a cryptographic task, some security definitions seem fine from a stand-alone viewpoint, but may present some loopholes when composed with other protocols. 
For example, the accessible information criteria for quantum key distribution falls short of providing composable security due to the data-locking effect~\cite{divincenzo2004locking,konig2007small}.
A similar problem occurs in WCF. 
In particular, a WCF protocol $\pi=(\pi_1, \pi_2)$ (where $\pi_1$ and $\pi_2$ are used instead of $\pi_A$ and $\pi_B$ to avoid confusion in subsequent protocol description) may find its stand-alone security insufficient for upper-level tasks.

To illustrate, consider an attacking strategy for an unbalanced WCF protocol $\mu = (\mu_A,\mu_B)$ constructed from two instances of balanced WCF protocol $\pi=(\pi_1,\pi_2)$. 
The protocol $\mu$ aims to model the scenario where Alice and Bob are concurrently running two instances of WCF with an undetermined execution order.
Assume the protocol $\pi$ follows the quantum comb model~\cite{chiribella2009theoretical,gutoski2007toward} with $n+1$ (assume that $n$ is odd without loss of generality) rounds of communication, where Alice and Bob send messages alternately. 
Then, $\pi_1$ and $\pi_2$ can be characterized by the ordered set $\cT_{\pi}$ shown in \Cref{fig:insecure-b}, where the subset $\{t_i\}_{i=0}^n$ is totally ordered.
The comb model further requires that 
\begin{enumerate}[start=1, label={(\bfseries R\arabic*)}]
    \item Messages are sent by $\pi_1$ and accepted by $\pi_2$ only at $t_i$ where $i$ is even.
    \item Messages are sent by $\pi_2$ and accepted by $\pi_1$ only at $t_i$ where $i$ is odd.
\end{enumerate}
The above conditions allow us to consider the wires in the two directions separately--one within the even subset of $\cT_{\pi}$ and the other within the odd subset of $\cT_{\pi}$.

The protocol $\mu$ is constructed by connecting two independent instances of $\pi$ with the converter $\eta = (\eta_A, \eta_B)$, as shown in \Cref{fig:insecure-a}.
The converter $\eta_A$ accepts a classical bit $c_A$ at position $t_n$ from $\pi_1$ and a classical bit $c_A'$ at position $t_b'$ from $\pi_2'$ and outputs $c_A''$ at position $t_A$,
where $c_A'' = c_A \land c_A'.$

\begin{figure}[ht]
    \centering
     \subcaptionbox[protocol mu]{\label{fig:insecure-a}}[\textwidth][c]{
    \input{pic/insecure}}
    \subcaptionbox[order of mu]{\label{fig:insecure-b}}[\textwidth][c]{
    \input{pic/insecure-order}}
    \caption{
    (a) The protocol $\mu = (\mu_A,\mu_B)$ is constructed from protocol $\pi$, $\pi'$ and $\eta$. $\eta = (\eta_A,\eta_B)$ is a pair of converters and the two instances $\pi'=(\pi'_1,\pi'_2)$ and $\pi=(\pi_1,\pi_2)$ are the same except the ordered sets $\cT_{\pi}',\cT_{\pi}$ are independent. 
    (b) This figure shows the Hasse diagram of the posets $\cT_{\pi},\cT_{\pi}',\cT_{\eta_A},\cT_{\eta_B},\cT_{\mu_A},\cT_{\mu_B}$. 
    The Hasse diagram depicts the order relation with directed edge and omits the transitive and reflexive connections for simplicity. 
    The poset $\cT_{\mu_A}$ is obtained by union of $\cT_{\pi},\cT_{\pi}',\cT_{\eta_A}$, which is defined in \Cref{def:upo}. The poset of $\cT_{\mu_b}$ is obtained similarly.} 
    \label{fig:insecure}
\end{figure}

We demonstrate that the stand-alone security of $\mu$ can be compromised even if that of $\pi$ is maintained. 
As per~\cite[Lemma 4]{chailloux2009optimal}, if $\pi$ is a $\mathrm{WCF}(\frac{1}{2},\epsilon)$ protocol and $\epsilon < 1/6$, $\mu$ is expected to be a $\mathrm{WCF}(\frac{3}{4},\frac{3}{2}\epsilon)$ protocol. 
However, it fails to be a $\mathrm{WCF}(\frac{3}{4},\frac{3}{2}\epsilon)$ protocol, as there exists a strategy for dishonest Bob to win the game with probability $\frac{1}{2}$.
Before elaborating on this strategy, we introduce the delay box, denoted by $\Theta_f$.
This box essentially transforms a message at position $t$ to a future position $t'$ according to some delay function $f$. 
For a rigorous definition of delay box, see \Cref{app:insecure}.

For simplicity, denote $\theta_{f_1}, \theta_{f_2}$ as $\theta_1, \theta_2$. Consider the parallel composition of delay boxes $\theta = \theta_1\| \theta_2$ depicted in \Cref{fig:delay}.
The delay functions are defined as
\begin{align}
    &f_1: \{t_i \in \cT_{\pi} \mid \text{$i$ is even} \} \to \cT_{\theta_1},  t_i \mapsto t_i', \label{eq:f1}\\
    &f_2: \{t_i' \in \cT_{\pi}' \mid \text{$i$ is odd} \} \to \cT_{\theta_2},  t_i' \mapsto t_i. \label{eq:f2}
\end{align}

\begin{figure}[h]
    \centering
     \subcaptionbox[protocol mu]{\label{fig:delay-a}}[\textwidth][c]{
    \input{pic/delaybox.tex}}
    \subcaptionbox[order of mu]{\label{fig:delay-b}}[\textwidth][c]{
    \input{pic/delaybox-order.tex}}
    \caption{(a) This figure shows the connection of $\mu_A$ and delay box $\theta$. The box $\theta = \theta_1 \| \theta_2$ is a parallel composition of two delay boxes. (b) This figure shows the Hasse diagram of the ordered sets $\cT_{\theta_1}$, $\cT_{\theta_2}$ and $\cT_{(\pi_1\|\pi_2')\theta}$.} 
    \label{fig:delay}
\end{figure}

\Cref{prop:insecure} shows that \Cref{fig:delay-a} emulates a WCF protocol between Alice and herself. See \Cref{prf:insecure} for the proof.
\begin{proposition} \label{prop:insecure}
The box $(\pi_1\| \pi_2') \theta$ has the same output distribution as $\pi_1 \pi_2$.
\end{proposition}

Since $\pi=(\pi_1,\pi_2)$ is a balanced WCF protocol, the distribution of the final output $c_A''$ is $P_{c_A''} = (1/2 , 1/2)$, which contradicts our expectation that $\mu$ is a $\mathrm{WCF}(\frac{3}{4},\frac{3}{2} \epsilon)$ protocol.

This example suggests that a stronger criterion is necessary to guarantee the safe composition of WCF protocols.
To achieve a $\mathrm{WCF}(\frac{3}{4},\frac{3}{2}\epsilon)$ protocol, it is essential that for any box connected to the inner interface of $\pi_1 \| \pi_2'$, the output satisfies
\begin{align}
    \forall x \in \{0,1\}, \Pr[c_A'=1|c_A=x] \le \frac{1}{2} + \epsilon.
\end{align}
More generally, to construct an unbalanced WCF protocol $\mathrm{WCF}(z,\epsilon')$ for arbitrary $z\in[0,\frac{1}{2}]$ with $n$ instances of protocol $\pi$, the following condition is needed:
\begin{align} \label{eq:cond-ubwcf}
\forall i \in[n], \forall x^{i-1}_1 \in \{0,1\}^{i-1}, \Pr[c_{A,i}=1|c_{A,1}=x_1,\ldots, c_{A,i-1}=x_{i-1}] \le \frac{1}{2} + \epsilon,
\end{align}
where $x_1^{i} = x_1x_2 \dots x_i$ and the result of the $i$-th instance is denoted by $c_{A,i}$. 
Notably, inequality~\eqref{eq:cond-ubwcf} is an implicit but crucial condition for the security of unbalanced WCF protocol introduced in~\cite{chailloux2009optimal}.

\subsection{Global security}
What is the sufficient condition for Inequality~\eqref{eq:cond-ubwcf} to hold? 
More broadly speaking, what is the safe way to compose WCF protocols? 
Despite negative results discussed in \Cref{sec:neg} on the possibility of abstracting WCF as a resource, there are still certain secure ways of composition, e.g., those satisfying Inequality~\eqref{eq:cond-ubwcf}. 
To explore this, we generalize Inequality~\eqref{eq:cond-ubwcf}, introduce the concept of global secrecy and specify the sufficient condition for it.

\begin{definition}[Global security of WCF in a larger system] \label{def:global-sec}
Let $\pi$ be a $\mathrm{WCF}(z,\epsilon)$ protocol and $\eta$ be a protocol that only outputs classical messages $K_A,K_A',K_B,K_B' \in \mathbb{N}$ at its outer interface, where $K_A,K_B$ are output before $C_A,C_B$ and $K_A',K_B'$ are output after $C_A,C_B$.
For a new protocol $\xi = (\xi_A, \xi_B)$ with $\xi_A = \pi_A \| \eta_A, \xi_B = \pi_B \| \eta_B$ (denoted by $\xi = \pi \| \eta$), the WCF protocol $\pi$ is said to be globally secure in $\xi$ if:
\begin{enumerate}
    \item For any box $\alpha$ compatible with the inner interface of $\xi_A$: $\forall k \in \mathbb{N}, \Pr[ C_A=1 | K_A =k ] \le z+\epsilon$,
    \item For any box $\beta$ compatible with the inner interface of $\xi_B$: $\forall k \in \mathbb{N}, \Pr[ C_B=0 | K_B =k ] \le 1- z +\epsilon$.
\end{enumerate}
\end{definition}

The global security defined above is a property of the whole composed system $\pi \| \eta$ rather than the protocol $\eta$ alone.
For some protocols $\eta$ and $\eta'$, a protocol $\pi$ may be globally secure in $\pi \| \eta$ but not in $\pi \| \eta'$.

Notably, global security in this definition is much weaker than the universal composability in the AC framework.
It is not within any composable framework and only applies to the scenario where the system composed with WCF in parallel takes no input from the outer interface.
This assumption is even stricter than the non-interactive assumption in the bounded storage model~\cite{wehner2008composable}.
Therefore, global security of all components in a larger system only ensures the stand-alone security of the larger system. 

The following theorem formulates the intuition that the WCF protocol is not affected by other boxes as long as we forbid any other interaction when the WCF protocol is active. 
\begin{theorem}
 \label{thm:cond-gs}
    Let $\pi$ be a $\mathrm{WCF}(z,\epsilon)$ protocol and $\eta$ be any protocol with only classical output at the outer interface. $\xi = \pi \| \eta$ is their parallel composition with the ordered set $\cT$. Then $\pi$ is globally secure in the protocol if there exists a partition $\cT_1,\cT_2,\cT_3$ of $\cT$ such that
    \begin{enumerate}[start=1, label={(\bfseries C\arabic*)}]
    \item $\forall t_1\in \cT_1, t_2 \in \cT_2, t_3 \in \cT_3$, $t_1 \prec t_2 \prec t_3$;
    \item $\eta$ only accepts input and produces output on set $\cT_1 \cup \cT_3$, while $\pi$ only accepts input and produces output on set $\cT_2$.
    \end{enumerate}
\end{theorem}
The proof of \Cref{thm:cond-gs} is in \Cref{prf:cond-gs}. 
The proof idea is by contradiction. If the adversary can break the global security of the composed protocol, then the adversary can also break the stand-alone security of the WCF protocol, which contradicts the fact and thus completes the proof.
\Cref{thm:cond-gs} establishes the sufficient condition that the WCF protocols can be composed while maintaining the overall stand-alone security.

For more complicated functionalities, it is usually easy to construct asymmetric imperfect protocols in which one party can cheat more than the other party. 
Unbalanced WCF protocols make it possible for two parties to fairly choose one of these asymmetric imperfect protocols at random.
This results in a less asymmetric imperfect protocol. 
Therefore, unbalanced WCF protocols play a key role in constructing optimal protocols of more complicated functionalities such as SCF and BC~\cite{chailloux2009optimal, chailloux2011optimal}. 
An unbalanced $\mathrm{WCF}(z,\epsilon')$ protocol can be constructed from balanced $\mathrm{WCF}(\frac{1}{2},\epsilon)$ protocols where $z = 0.b_1 b_2 \dots b_n$ is in binary representation, as is shown in \Cref{fig:unb-wcf}. 

\begin{figure}[htbp!]
    \centering
    \resizebox{!}{6cm}{%
    \input{pic/unb-wcf.tex}
    }
    \resizebox{!}{6cm}{%
    \input{pic/unb-order}
    }
    \caption{(a) This figure shows Alice part of the unbalanced WCF protocol $\eta_A = \tau_A (\pi_{A,1} \|\dots \| \pi_{A,n})$ constructed from balanced protocols $\pi$.
    It runs $n$ instances of $\pi_A$ and processes their result in $\tau_A$.
    The program of $\tau_A$ is presented in the dashed box. 
    Bob's part works similarly. 
    (b) This figure shows the form of the Hasse diagram of $\eta_A$, where $\cT_i$ is the posets for instance $\pi_{A,i}$. The instances of $\pi_A$ run in a time-sequential way.
    }
    \label{fig:unb-wcf}
\end{figure}

As an example, we apply \Cref{thm:cond-gs} to this construction to prove the stand-alone security of the unbalanced WCF protocol in \Cref{coro:unb}. The idea is to prove the stand-alone security of $\tau_{i} (\pi_1\|...\|\pi_i)$ constructed from $\tau_{i-1}(\pi_1\|...\|\pi_{i-1})$ and $\pi_i$ using \Cref{thm:cond-gs} recursively. The detailed proof can be found in \Cref{prf:unb}. 

\begin{corollary} \Label{coro:unb}
    Let $z\in[0,\frac{1}{2}]$. The unbalanced WCF protocol shown in \Cref{fig:unb-wcf} is a $\mathrm{WCF}(z,\epsilon')$ protocol with $\epsilon' = 2\epsilon + o(\epsilon)$.
\end{corollary}

\section{Discussion}
Our work demonstrates that the standalone security of protocols does not guarantee the security under certain types of composition. In particular, although quantum WCF protocol can be secure in the standalone scenario, it fails to be composably secure in the abstract cryptography framework. Moreover, when we construct new protocols with quantum WCF protocol, it remains secure in sequential composition, but exhibits security loopholes in concurrent composition. 
Similar security loopholes also arise in other two-party cryptographic tasks, such as bit commitment~\cite{unruh2011concurrent, kaniewski2013secure}.
For these tasks, though we have hoped for simple and good abstractions for actual protocols in the abstract framework, our work suggests that such abstraction does not exist, which makes security analysis of such protocols notably challenging.

\section{Acknowledgement}
We thank Christopher Portmann and Vilasini Venkatesh for useful discussions. 
This research was initiated while AB was hosted at CQT.
This research is supported by the National Research Foundation, Singapore and A*STAR under its CQT Bridging Grant and its Quantum Engineering Programme (NRF2021-QEP2-01-P06). AB acknowledges the partial support provided by Commonwealth Cyber Initiative (CCI-SWVA) under the 2023 Cyber Innovation Scholars Program. 


\appendix
\section{Causal box model}\label{app:causal-box}
In this appendix, we introduce necessary concepts related to the causal box model. Part of this section is from \cite{portmann2017causal}.
\subsection{Causal box}
Let $\cT$ be a partially ordered set (poset). That is, $\cT$ is a set equipped with a binary relation $\preceq$ that is reflexive, antisymmetric and transitive. 
For $x,y \in \cT$, $x\prec y$ if $x \preceq y$ and $x \neq y$.
Denote $\cT^{\preceq t} = \{x\in \cT \mid x \preceq t \}$.

\begin{definition}[Cuts]
A subset $\cC$ of poset $\cT$ is a cut if it takes the form
\begin{align}
    \cC = \bigcup_{t\in \cC} \cT^{\preceq t}.
\end{align}
A cut is bounded if $\exists t \in \cT, \cC \subseteq \cT^{\preceq t}$.
The set of all cuts of $\cT$ is denoted by $\mf{C}(\cT)$.
The set of all bounded cuts of $\cT$ is denoted by $\bar{\fC}(\cT)$.
\end{definition}

\begin{lemma}[An equivalent condition for cut] \label{lem:cut-crt}
    A subset $\cC$ of a poset $\cT$ is a cut if and only if $t \in \cC \Rightarrow \forall x \preceq t, x \in \cC$.
\end{lemma}
\begin{proof}
    If $\cC$ is a cut, for $t\in \cC$ and $x\preceq t$, we have $x \in \cT^{\preceq t}$, thus $x \in \cC$.
    If $\forall t \in \cC, \forall x \preceq t, x \in \cC$, then $\forall x \in \cT^{\preceq t}, x \in \cT$. We have $\forall t \in \cC, \cT^{\preceq t} \subseteq \cC$, thus $\cC = \bigcup_{t \in \cC} \cT^{\preceq t}$.
\end{proof}

\begin{definition}[Causality function] \label{def:causal-fun}
    A function $\chi: \mf{C}(\cT) \to \mf{C}(\cT)$ is a causality function if: 
       \begin{align}
  &\forall \cC,\cD \in \fC(\cT), \quad  \chi(\cC \cup \cD) = \chi(\cC) \cup \chi(\cD)\,, \label{eq:causality.homomorphism} \\
  &\forall \cC,\cD \in \fC(\cT), \quad  \cC \subseteq \cD \implies \chi(\cC) \subseteq \chi(\cD)\,, \label{eq:causality.monotone} \\
  &\forall \cC \in \bar{\fC}(\cT) \setminus \{\emptyset\}, \quad  \chi(\cC) \subsetneq \cC\,, \label{eq:causality.decreasing} \\
  &\forall \cC \in \bar{\fC}(\cT), \forall t \in \cC, \exists n \in \mathbb{N}, \quad t \notin \chi^n(\cC)\,, \label{eq:causality.finite}
\end{align}
\end{definition}

Using cuts and causality function, we can formally define the causal box.
\begin{definition}[Causal box]
    A causal box is a system with input wire $X$ and output wire $Y$, defined as a set of CPTP maps:
    \begin{align}
        \Phi \coloneqq \{\Phi^{\cC}: \fT (\cF_X^{\cT}) \to \fT (\cF_Y^{\cC})\}_{\cC \in \bar{\fC}(\cT)},
    \end{align}
    where $\fT (\cH)$ denotes the set of all trace class operators on the Hilbert space $\cH$.
    These maps must be mutually consistent
    \begin{align}
        \forall \cC, \cD \in \bar{\fC}(\cT), \cC \subseteq \cD: \Phi^{\cC} = \tr_{\cD \setminus \cC} \circ \Phi^{\cD},
    \end{align}
    and respect causality, i.e., there exists a causality function $\chi: \fC(\cT) \to \fC(\cT)$ such that 
    \begin{align}
        \Phi^{\cC} = \Phi^{\cC} \circ \tr_{\cT \setminus \chi (\cC)}.
    \end{align}
\end{definition}
Alternatively, the causal box $\Phi$ can be described in Choi-Jamio{\l}kowski (CJ) representation.
To be precise, the Choi operator $R_{\Phi}^{\cC}$ for the CPTP map $\Phi^{\cC}$ is a sesquilinear semi-definite form on $\cF_{Y}^{\cC} \otimes \cF_{X}^{\chi(\cC)}$ such that
\begin{align}
    R_{\Phi}^{\cC} (\psi_Y \otimes \psi_X, \varphi_Y \otimes \varphi_X ) = 
    \bra{\psi_Y} \Phi^{\cC} \left( \ket{\bar{\psi}_X} \bra{\bar{\varphi}_X} \right) \ket{\varphi_Y}.
\end{align}
where $\ket{\overline{\psi}} = \sum_i \ket{i} \braket{\psi}{i}$ for a fixed basis $\{\ket{i}\}_i$.

\subsection{Connection of causal boxes}
Causal boxes can be connected by parallel composition and loop.
\begin{definition}[Parallel composition of causal boxes]
    Let $\Phi = \{\Phi^{\cC}\}_{\cC \in \bar{\fC}(\cT)}$ and $\Psi = \{ \Psi^{\cC} \}_{\cC \in \bar{\fC}(\cT)}$ be two causal boxes. The parallel composition of them, denoted by $\Phi \| \Psi$ is defined as
    \begin{align}
        \Gamma \coloneqq \left\{\Phi^{\cC} \otimes \Psi^{\cC} \right\}.
    \end{align}
\end{definition}

The above definition of parallel composition implicitly requires the causal boxes $\Phi$ and $\Psi$ share the same partially ordered set $\cT$. 
However, a general causal box may not contain the information of the causal structure of some other causal box. Then, a predefined common poset $\cT$ may not exist.
In this case, to make the composition possible, we propose a more general version of parallel composition. Before that, we introduce the union of posets that might have intersection.

A poset $\cT$ can be treated as a tuple of a set and a partial order relation $\cT= (\fS,\preceq)$. 
We define the union of two posets as follows. 

\begin{definition}[Union of partial orders] \label{def:upo}
    Let $\cT_1 = ( \fS_1, \preceq_1)$ and $\cT_2 = (\fS_2, \preceq_2)$ be two compatible posets.
    Let the relation on the set $\fS =\fS_1 \cup \fS_2$ be $\preceq = \preceq_1 \sqcup \preceq_2$. Then
    $x \preceq y $ if and only if there exists a sequence $\{a_i\}_{i=0}^{n} \subseteq \fS_1 \sqcup \fS_2$ such that $a_1=x$, $a_n=y$ and $\forall i \in [n]$  either $a_{i-1} \preceq_{1} a_i$ or $a_{i-1} \preceq_2 a_i$.  
    The union of posets, denoted by $\cT_1 \sqcup \cT_2$, is then defined as 
    \begin{align}
        \cT_1 \sqcup \cT_2 \coloneqq (\fS_1 \cup \fS_2, \preceq_1 \sqcup \preceq_2) =  (\fS, \preceq)  
    \end{align}
\end{definition}

In the above definition, $\cT_1$ and $\cT_2$ are said to be compatible if $\preceq_1 \sqcup \preceq_2$ is still a partial order.

After introducing the union of posets, we are ready to define the parallel composition of two causal boxes with compatible posets.
\begin{definition}[Parallel composition of causal boxes with arbitrary compatible ordered sets] \label{def:pc2}
    $\Phi = \{\Phi^{\cC}\}_{\cC \in \bar{\mf{C}}(\cT_1)}$ and $\Psi = \{\Psi^{\cC}\}_{\cC \in \bar{\mf{C}} (\cT_2)}$ are causal boxes defined on compatible partially ordered set $\cT_1$ and $\cT_2$ respectively, which can have \emph{arbitrary intersection}. $\Phi$ has input wire $X_1$ and output wire $Y_1$. $\Psi$ has input wire $X_2$ and output wire $Y_2$. 
    The parallel composition of $\Phi$ and $\Psi$, denoted by $\Gamma = \Phi \| \Psi$, is defined on the new poset $\cT = \cT_1 \sqcup \cT_2$ as 
    \begin{align}
        \Gamma &\coloneqq \{\Gamma^{\cC} \}_{\cC \in \bar{\mf{C}}(\cT)},  \\
        \Gamma^{\cC} & \coloneqq \Phi^{\cC_1} \otimes   \Psi^{\cC_2} \circ \tr_{\cF^{\cT \setminus \cT_1}_{X_1} \otimes \cF^{\cT \setminus \cT_2}_{X_2}}, \label{eq:comp-emb2}
    \end{align}
    where $\cC_1 = \cT_1 \cap \cC, \cC_2 = \cT_2 \cap \cC$. 
\end{definition}

\Cref{lemma:union-cuts},~\ref{lemma:cas} and~\ref{lemma:asso} suggest that \Cref{def:pc2} is well-defined.

\begin{lemma} \label{lemma:union-cuts}
    For compatible posets $\cT_1$ and $\cT_2$, if $\cC$ is a cut in $\cT_1 \sqcup \cT_2$, then $\cC_1 = \cC \cap \cT_1$ and $\cC_2 = \cC \cap \cT_2$ are cuts in $\cT_1$ and $\cT_2$ respectively.
\end{lemma}
\begin{proof}
    Notice that
    \begin{align}
        \cC_1 = \cC \cap \cT_1 = \bigcup_{t \in \cC} (\cT^{\preceq t} \cap \cT_1).
    \end{align}
    Assume $t' \in \cT^{\preceq t} \cap \cT_1$,
    then by \Cref{def:upo}, $\forall x \in \cT_1, x \preceq_1 t' \Rightarrow x \preceq t \Rightarrow x \in \cT^{\preceq t} \cap \cT_1$. 
    According to \Cref{lem:cut-crt}, $\cT^{\preceq t} \cap \cT_1$ is a cut in $\cT_1$. 
    As the union of cuts, $\cC_1$ is also a cut. The same argument follows for $\cC_2$.
\end{proof}

\begin{lemma}[Existence of causality function] \label{lemma:cas}
    Let $\Gamma = \Phi \| \Psi$ with the poset $\cT = \cT_1 \sqcup \cT_2$ be the parallel composition of causal boxes $\Phi$ and $\Psi$ with compatible posets $\cT_1$ and $\cT_2$. There exists a causality function $\chi$ on $\fC(\cT)$ such that:
    \begin{align}
        \Gamma^{\cC} = \Gamma^{\cC} \circ \tr_{\cF_{X_1 \oplus X_2}^{\chi(\cC)}}
    \end{align}
\end{lemma}
\begin{proof}
    Assume $\chi_1$ and $\chi_2$ are causality functions for $\Phi$ and $\Psi$, respectively.
    Construct the function $\chi: \fC(\cT) \to \fC(\cT)$ as $\chi(\cC) = \chi_1(\cC \cap \cT_1) \cup \chi_2 (\cC \cap \cT_2)$, then it is easy to verify the conditions in \Cref{def:causal-fun}.
\end{proof}

\begin{lemma}[Associativity] \label{lemma:asso}
    For boxes $\Phi,\Psi,\Theta$ with compatible posets $\cT_{\Phi}, \cT_{\Psi}, \cT_{\Theta}$, we have
    \begin{align}
        (\Phi \| \Psi) \| \Theta = \Phi \| (\Psi \| \Theta).
    \end{align}
\end{lemma}
\begin{proof}
    This lemma follows directly from \Cref{def:pc2}.
\end{proof}

\begin{definition}[Loop]
    Let 
    $\Phi = \{ \Phi^{\cC} : 
    \fT (\cF_{AB}^{\cT}) \to \fT(\cF_{CD}^{\cC}) \}_{\cC \in \bar{\fC}(\cT)}$ 
    be a causal box such that $\dim \mathcal{H}_B = \dim \mathcal{H}_C$ is its CJ representation.
    Let $\{\ket{k_{C}}\}$ and $\{ \ket{l_C}\}$ be any orthonormal bases of $\cF_{C}^{\cC}$ and 
    $\{\ket{k_B}\}$ and $\{\ket{l_B}\}$ be corresponding bases of $\cF_{B}^{\cC}$.
    The new box by looping from wire $B$ to wire $C$, denoted by $\Psi = \Phi^{C \hookrightarrow B}$, is a set of maps $\{\Psi^{\cC} : \fT(\cF_{A}^{\cT}) \to \fT( \cF_{D}^{\cC})\}$ with CJ representation
    \begin{align}
        &R_{\Psi}^{\cC} ( \psi_D \otimes \psi_A , \varphi_D \otimes \varphi_A)\\
        &= \sum_{k,l} R_{\Phi}^{\cC}(k_{C} \otimes \psi_D \otimes \psi_A \otimes \bar{k}_B,
        l_C \otimes \varphi_D \otimes \varphi_A \otimes \bar{l}_B),
    \end{align}
    where the conjugate $\ket{\overline{\psi}} = \sum_i \ket{i_B} \braket{\psi}{i_B}$ is taken with respect to the base $\ket{i_B}$ of $\cF_B^{\cT}$ on which the CJ representation is defined. 
\end{definition}

Note that, to define a loop, it is necessary to specify two bases for each of the connected wire.  With different bases for the wires, the result of the loop can be different.

\subsection{Delay box} \label{app:insecure}

\begin{definition}[Delay function] \label{def:df}
    Let $\cT$ be a poset and $\cG = \{t \in \cT \mid \exists t' \in \cT, t \prec t'\}$ is a subset of $\cT$. A function  $f: \cG \to \cT$ is a delay function on $\cT$ if for any $t \in \cG$, $ t \prec f(t)$. 
\end{definition}

\begin{proposition}
    If $\cT$ is a finite set, a delay function $f$ on $\cT$ induces a causality function 
    \begin{align}
        \chi_f(\cC) = \bigcup_{t:f(t)\in \cC} \cT^{\le t}.
    \end{align}
\end{proposition}
It is easy to verify $\chi_f$ is a causality function by definition.

\begin{definition}[Delay box] \label{def:db}
    A delay box derived from a delay function $f:\cG \to \cT$ with input wire $X$ and output wire $X'$ is denoted by $\Theta_f = \{\Theta_f^{\cC} : \fT (\cF_X^{\chi_f(\cC)}) \to \fT (\cF_X^{\cC}) \}_{\cC \in \bar{\fC}(\cT)}$. $\cH_X$ and $\cH_{X'}$ are of the same dimension. $\Theta_f^{\cC}$ is defined as the following isometry:
    \begin{align}
        U_f^{\cC} = \bigoplus_{n=0}^{\infty} \left( \sum_{i\in I} \ket{i}_{X'}\bra{i}_X \otimes \sum_{ t \in \chi_f (\cC) } \ket{f(t)}\bra{t} \right)^{\otimes n},
    \end{align}
    for some bases $\{\ket{i}_X\}_{i\in I}$ and $\{\ket{i}_{X'}\}_{i\in I}$ of spaces $\cH_X$ and $\cH_{X'}$.
\end{definition}

$\theta_1$ is a delay box from input wire $C'$ to output wire $C''$ and is derived from $f_1$ defined in \Cref{eq:f1}.
Then, $U_{f_1}^{\cC}$ is an isometry between $\cF_{C'}^{\chi_{f_1}(\cC)}$ and $\cF_{C''}^{\cC \cap \cT_{\pi'}}$.

\section{Proofs}
\subsection{Proof of \Cref{prop:insecure}}
\label{prf:insecure}

\begin{proposition*}
The box $(\pi_1\| \pi_2') \theta$ has the same output distribution as $\pi_1 \pi_2$.
\end{proposition*}

\Cref{prop:insecure} implies that two parties running $\pi_1$ and $\pi_2'$ with an adversary running $\theta$ is equivalent to two parties running $\pi_1$ and $\pi_2'$ between themselves. 
To achieve this, the adversary performs the man-in-the-middle attack, i.e., uses delay boxes to pass messages from $\pi_1$ to $\pi_2'$ and pass responses from $\pi_2'$ to $\pi_1$ subsequently.
To be precise, notice that the boxes $\pi_1$ and $\pi_2'$ work on two independent posets, i.e., the top part and bottom part in \Cref{fig:proof-b}, respectively. 
The delay boxes $\theta_1$ and $\theta_2$ are the bridges that connect the two corresponding wire spaces.
Then the key point is to find proper wire space isomorphisms for the wire connections such that $\pi_2'$ together with the delay boxes is a perfect emulation of $\pi_2$.

\begin{figure}[h]
    \centering
     \subcaptionbox[protocol mu]{\label{fig:proof-a}}[\textwidth][c]{
    \input{pic/proof-insecure}}
    \subcaptionbox[order of mu]{\label{fig:proof-b}}[\textwidth][c]{
    \input{pic/proof-order.tex}}
    \caption{
    (a) This figure shows the connected box $(\pi_1 \| \pi_2') \theta$. The wire with a blue notation indicates that the messages on this wire only stay in sets $\cT_{\pi}^{\text{even}}$ and $\cT_{\pi'}^{\text{even}}$, while the red notation corresponds to the sets $\cT_{\pi}^{\text{odd}}$ and $\cT_{\pi'}^{\text{odd}}$.  
    (b) This figure explains the partition of the totally ordered set $\cT_{(\pi_1 \| \pi_2') \theta}$ with a non-standard Hasse diagram. 
    Specifically,
    The blue elements at the top form the set  $\cT_{\pi}^{\text{even}}$;
    The blue elements at the bottom form the set $\cT_{\pi'}^{\text{even}}$;
    The red elements at the top form the set $\cT_{\pi}^{\text{odd}}$;
    The red elements at the bottom form the set $\cT_{\pi'}^{\text{odd}}$. 
    } 
    \label{fig:proof}
\end{figure}

Define the sets
\begin{align}
    \cT_{\pi}^{\text{even}} &\coloneqq \{t_i \in \cT_{\pi} \mid i \text{ is even}\},& \cT_{\pi}^{\text{odd}} & \coloneqq \{t_i \in \cT_{\pi} \mid i \text{ is odd}\}, \\
    \cT_{\pi'}^{\text{even}} & \coloneqq \{t_i \in \cT_{\pi'} \mid i \text{ is even}\},& \cT_{\pi'}^{\text{odd}} & \coloneqq \{t_i \in \cT_{\pi'} \mid i \text{ is odd}\}.
\end{align}
Because the protocol $\pi$ is restricted to the comb model, the wires $B,C,C',C'',E$ are only effective on the set $\cT_{\pi}^{\text{even}}$ and $\cT_{\pi'}^{\text{even}}$.
Similarly, the wires $A,D,D',D'',F$ are only effective on the set $\cT_{\pi}^{\text{odd}}$ and $\cT_{\pi'}^{\text{odd}}$.
By saying wire $X$ is only effective on $\cT'$, we mean the state on this wire is always in the space $\cF_{X}^{\cT'}$.

For convenience, where the cut $\cC$ is obvious, denote for $X=B,C,C',C'',E$
\begin{align}
    \hat{X} &= \cF_{X}^{\cC \cap \cT_{\pi}^{\text{even}}}, \\
    \tilde{X} &= \cF_{X}^{\cC \cap \cT_{\pi'}^{\text{even}}},
\end{align}
where $\cT_{\pi}$ is the poset for $\pi=(\pi_1,\pi_2)$ and $\cT_{\pi'}$ is the partial order for $\pi' = (\pi_2', \pi_1')$.
Denote the wire space of $X = A,D,D',D'',F$ by
\begin{align}
    \hat{X} &= \cF_{X}^{\cC \cap \cT_{\pi}^{\text{odd}}}, \\
    \tilde{X} &= \cF_{X}^{\cC \cap \cT_{\pi'}^{\text{odd}}}.
\end{align}
For any wire $X$, we call the isomorphism between $\hat{X}$ and $\tilde{X}$, obtained by substituting $t_i$ with $t_i'$, the natural isomorphism.

For $\cC \in \fC(\cT_{\pi})$, $\{\ket{i}_{\htA}\}$ is an orthonormal base of space $\cF_{A}^{\chi_{\pi}(\cC)}$, The Choi-Jamio{\l}kowski (CJ) representation of the box $\pi_1$ is a sesquilinear form $R_{\pi_1}^{\cC}$ satisfying 
\begin{align} 
     R_{\pi_1}^{\cC} 
     ( \psi_{\htC} \otimes \psi_{\htE} \otimes \psi_{\htA}, 
    \varphi_{\htC} \otimes \varphi_{\htE} \otimes \varphi_{\htA}) =  
    \bra{\psi_{\htC}} \bra{\psi_{\htE}} 
    \Pi_1^{\cC} 
    \left( \ket{ \bar{\psi} _{\htA}} \bra{ \bar{\varphi} _{\htA}} \right) 
   \ket{\varphi_{\htE}} \ket{\varphi_{\htC}}, 
   \label{eq:cjpi}
\end{align}
where $\Pi_1^{\cC}$ is the CPTP map for box $\pi_1$ on cut $\cC$ and
the conjugate of $\bar{\psi}_{\htA},\bar{\varphi}_{\htA}$ is taken on basis $\{\ket{i}_{\htA}\}$.
We only care about the effect of $R_{\pi_1}^{\cC}$ on space $\htA,\htC,\htE$ because of the conditions (\textbf{R1}), (\textbf{R2}) required by the comb model.

For $\cC \in \fC(\cT_{(\pi_1\|\pi_2') \theta})$, assume that the wire $B,C$ are connected with respect to the bases 
$\{ \ket{ k_{\htC} } \}, \{ \ket{ l_{ \htC } } \}$ 
of $\htC$ and
the bases 
$\{ \ket{ k_{\htB} } \}, \{ \ket{ l_{ \htB } } \}$ of $\htB$.
The wire $D,A$ are connected with respect to 
the orthonormal bases 
$\{ \ket{ p_{\htD} }\}$, $\{ \ket{ q_{ \htD }} \}$ of
$\htD$ and 
the orthonormal bases 
$\{ \ket{ p_{\htA} }\}$, $\{ \ket{ q_{ \htA }} \}$ of
$\htA$.
These bases are predefined by the protocol $\pi$ itself.
The CJ representation of $\pi_1 \pi_2$ is 
\begin{equation} 
\begin{aligned}
    & R_{\pi_1 \pi_2} ^{\cC}
    ( \psi_{\htE} \otimes \psi_{\htF}, 
    \varphi_{\htE} \otimes \varphi_{\htF} )  \\
    &=  \sum_{k,l,p,q}  
    R_{\pi_1}^{\cC} 
    ( k_{\htC} \otimes \psi_{\htE} \otimes \bar{p}_{A}, 
    l_{\htC} \otimes \varphi_{\htE} \otimes \bar{q}_{\htA} ) 
    R_{\pi_2}^{\cC}
    ( p_{\htD} \otimes \psi_{\htF} \otimes \bar{k}_{\htB},  
    q_{\htD} \otimes \varphi_{\htF} \otimes \bar{l}_{\htB}). \label{eq:cpp}
\end{aligned}
\end{equation}

For $\cC' \in \fC(\cT_{\theta_{1}}) $, $\{\ket{i}_{\htC'}\}$ is a base of space $\cF_{C'}^{\chi_{\theta_1}(\cC')}$, and 
$\{\ket{i}_{\tdC''}\}$ is a base of space $\cF_{C''}^{f_1(\chi_{\theta_1} (\cC'))}$ 
such that $U_{f_1} \ket{i}_{\htC} = \ket{i}_{\tdC''}$ 
\footnote{This is possible because $\cF_{C'}^{\chi_{\theta_1}(\cC')} \simeq \cF_{C''}^{f_1(\chi_{\theta_1} (\cC'))}$.}.
The CJ representation of the delay box $\theta_1$ is 
\begin{align} 
    &R_{\theta_1}^{\cC'} (\psi_{\htC''} \otimes \psi_{\tdC''} \otimes \alpha_{\htC'}, 
    \varphi_{\htC''} \otimes \varphi_{\tdC''} \otimes \beta_{\htC'} ) \notag \\
    &= \bra{\psi_{\tdC''}} 
    \Theta_{1}^{\cC'} \left( \ket{\bar{\alpha}_{\htC'}} \bra{\bar{\beta}_{\htC'}} \right) 
    \ket{\varphi_{\tdC''}}  
    \cdot
    \inner{ \psi_{\htC''} }{ \Omega_{\htC''} } 
    \inner{ \Omega_{\htC''} }{ \varphi_{\htC''} } \notag \\
     &=  \inner{ \psi_{\tdC''} }{ \bar{\alpha}_{\tdC''} }
    \inner{ \bar{\beta}_{\tdC''} }{ \varphi_{\tdC''} }  
    \cdot
    \inner{ \psi_{\htC''} }{ \Omega_{\htC''} } 
    \inner{ \Omega_{\htC''} }{ \varphi_{\htC''} } , \label{eq:de3}
\end{align}
where $\bar{\alpha}_{\tdC''} \in \cF_{C''}^{\cC'}$ and $\bar{\alpha}_{\htC'} \in \cF_{C'}^{\chi(\cC')}$ have the same coefficients in the bases of $\{ \ket{i}_{\tdC''} \}$
and $\{\ket{i}_{\htC'}\}$.
Since the state in space $\htC''$ is always a vacuum state, we omit the arguments $\psi_{\htC''}, \varphi_{\htC''}$ of $R_{\theta_1^{\cC'}}$ and the second term of \Cref{eq:de3} in the remaining part of the paper.

For all $\cC \in \bar{\fC}(\cT_{\pi_1 \theta_{1}})$, wire $C$ and $C'$ are connected with respect to the bases 
$\{\ket{k_{\htC'}}\}$ and $\{\ket{l_{\htC'}}\}$.
Then we have the isomorphisms
$\ket{k_{\htC'}} \cong \ket{k_{\htC}}, \ket{l_{\htC'}} \cong \ket{l_{\htC}}$
\footnote{Since any pair of bases $\{\ket{k_{\hat{C}'}}\}$, $\{ \ket{ k_{\hat{C}} }\}$ defines a specific isomorphism between $\hat{C}'$ and $\hat{C}$, we use the notation $\ket{k_{\hat{C}'}} \cong  \ket{ k_{\hat{C}} }$ to denote this isomorphism.}.
the CJ representation of the composition box $\bar{\pi}_1=\pi_1 \theta_1$ is,

\begin{align}
    & R_{\bar{\pi}_1}^{\cC} (\gamma_{\tdC''} \otimes \psi_{\htE} \otimes \psi_{\htA}, \eta_{\tdC''} \otimes \varphi_{\htE} \otimes \varphi_{\htA}) \notag\\
    &=
    \sum_{k,l} 
    R_{\pi_1}^{\cC_1} 
     ( k_{\htC} \otimes \psi_{\htE} \otimes \psi_{\htA}, 
     l_{\htC} \otimes \varphi_{\htE} \otimes \varphi_{\htA}) 
    R_{\theta_1}^{\cC_2} 
    \left(\gamma_{\tdC''} \otimes \bar{k}_{\htC'}, 
    \eta_{\tdC''} \otimes \bar{l}_{\htC'} \right)  \label{eq:ARR1} \\
    &= \sum_{k,l} 
    R_{\pi_1}^{\cC_1} 
     \left( k_{\htC} \otimes \psi_{\htE} \otimes \psi_{\htA},
     l_{\htC} \otimes \varphi_{\htE} \otimes \varphi_{\htA} \right) 
    \inner{ \gamma_{\tdC''} }{ k_{\tdC''} }
    \inner{ l_{\tdC''} }{ \eta_{\tdC''} } \label{eq:ARR2} \\
    &= R_{\pi_1}^{\cC_1} 
    \left( \sum_{k} \inner{ k_{\tdC''} }{ \gamma_{\tdC''} } k_{\htC} 
    \otimes \psi_{\htE} \otimes \psi_{\htA}, 
    \sum_{l} \inner{ l_{\tdC''} }{ \eta_{\tdC''} } l_{\htC} 
    \otimes \varphi_{\htE} \otimes \varphi_{\htA} \right) \label{eq:ARR3} \\
    &= R_{\pi_1}^{\cC_1} 
    \left( \gamma_{\htC} 
    \otimes \psi_{\htE} \otimes \psi_{\htA}, 
     \eta_{\htC} 
    \otimes \varphi_{\htE} \otimes \varphi_{\htA} \right) ,
\end{align}
where 
$\ket{\gamma_{\htC}} = 
\left(\sum_k \ket{k_{\htC}} \bra{k_{\tdC''}} \right) \ket{\gamma_{\tdC''}}$ 
and 
$\ket{\eta_{\htC}} = 
\left( \sum_l \ket{l_{\htC}} \bra{l_{\tdC''}} \right) \ket{\eta_{\tdC''}}$.
\Cref{eq:ARR1} follows from \Cref{def:pc2}, \Cref{eq:ARR2} follows from \Cref{eq:de3}, and \Cref{eq:ARR3} follows from the sesquilinearity.

Similarly, for all $\cC \in \bar{\fC}(\cT_{\pi_2' \theta_2})$, wire $D$ and $D'$ are connected with respect to bases 
$\{\ket{p_{\tdD'}}\}$ and $\{\ket{q_{\tdD'}}\}$. 
The CJ representation of the box $\bar{\pi}'_2 = \pi'_2 \theta_2$ is 
\begin{align}
    & R_{\bar{\pi}'_2}^{\cC} 
    \left( 
    \lambda_{\htD''} \otimes \psi_{\tdF} \otimes \psi_{\tdB},
    \omega_{\htD''} \otimes \varphi_{\tdF} \otimes \varphi_{\tdB}
    \right) \\
    &= R_{\pi'_2}^{\cC_1} 
    \left( 
    \lambda_{\tdD} 
    \otimes \psi_{\htF} \otimes \psi_{\htB}, 
    \omega_{\tdD} 
    \otimes \varphi_{\htF} \otimes \varphi_{\htB} \right),
\end{align}
where 
$\ket{\lambda_{\tdD}} = 
\left( \sum_p \ket{p_{\tdD}}  \bra{p_{\htD''}} \right) \ket{\lambda_{\htD''} }$ and 
$\ket{\omega_{\tdD}} = 
\left( \sum_q \ket{q_{\tdD}} \bra{q_{\htD''}} \right) \ket{\omega_{\htD''}}$.

Notice that, for all $\cC \in \bar{\fC}(\cT_{(\pi_1 \| \pi'_2)})$, we have the isomorphism chain between the wires:
\begin{equation}
\begin{aligned}
    &\ket{k_{\tdB}} \cong \ket{k_{\htB}} \cong 
    \ket{k_{\htC}} \cong \ket{k_{\htC'}} \cong \ket{k_{\tdC''}}, \\
    &\ket{l_{\tdB}} \cong \ket{l_{\htB}} \cong \ket{l_{\htC}} \cong \ket{l_{\htC'}} \cong \ket{l_{\tdC''}}, \\
    &\ket{p_{\htA}} \cong \ket{p_{\tdA}} \cong 
    \ket{p_{\tdD}} \cong \ket{p_{\tdD'}} \cong \ket{p_{\htD''}},\\
    &\ket{q_{\htA}} \cong \ket{q_{\tdA}} \cong 
    \ket{q_{\tdD}} \cong \ket{q_{\tdD'}} \cong \ket{q_{\htD''}}.
\end{aligned} \label{eq:base-co}
\end{equation}
The isomorphisms $\cong$ in the first column are natural isomorphisms.
The isomorphisms in the second column are predefined by the connection of $\pi_1 \pi_2$.
The isomorphisms in the third column can be chosen freely; that is, they are neither predefined nor natural.
The isomorphisms in the fourth column follows from the definition of delay boxes $\theta_1$ and $\theta_2$.
Therefore, the isomorphisms in the third column induce two pairs of bases for $\tilde{B}, \tilde{C}''$ and two pairs of bases for $\hat{A}, \hat{D}''$, which are precisely the bases we will use for $B ,C''$ and $A, D''$ connections.

Then, for $ \cC = \cT_{(\pi_1\|\pi_2')\theta} $, we choose the bases $\{\ket{k}_{\tdC''}\}, \{\ket{l}_{\tdC''}\}, \{ \ket{k}_{\htB}\}, \{ \ket{l}_{\htB} \}$ for the connection of $C'', B$ and the bases $\{\ket{p_{\htD''}}\}, \{\ket{q_{\htD''}}\}, \{\ket{p_{\tdA}}\}, \{\ket{q_{\tdA}}\}$ for the connection of $D'', A$.
The CJ representation of the final box $\bar{\pi}_1 \bar{\pi}'_2$ is
\begin{align}
    & R_{\bar{\pi}_1 \bar{\pi}'_2}^{\cT_{(\pi_1\|\pi_2) \theta}} 
    ( \psi_{\htE} \otimes \psi_{\tdF}, \varphi_{\htE} \otimes \varphi_{\tdF}) \\
    &= \sum_{k,l,p,q}
    R_{\bar{\pi}_1}^{\cT_{\bar{\pi}_1}}
    ( k_{\tdC''} \otimes \psi_{E} \otimes \bar{p}_{A}, l_{\tdC''} \otimes \varphi_{\htE} \otimes \bar{q}_{\htA} ) 
    R_{\bar{\pi}'_2}^{\cT_{\bar{\pi}'_2}}
    ( p_{\htD''} \otimes \psi_{\tdF} \otimes \bar{k}_{\tdB}, 
    q_{\htD''} \otimes \varphi_{\tdF} \otimes \bar{l}_{\tdB}) \\
    &= \sum_{k,l,p,q}
    R_{\pi_1}^{\cT_{\pi}}
    ( k_{\htC} \otimes \psi_{E} \otimes \bar{p}_{A}, l_{\htC} \otimes \varphi_{\htE} \otimes \bar{q}_{\htA} ) 
    R_{\pi'_2}^{\cT_{\pi'}}
    ( p_{\tdD} \otimes \psi_{\tdF} \otimes \bar{k}_{\tdB}, 
    q_{\tdD} \otimes \varphi_{\tdF} \otimes \bar{l}_{\tdB}). \label{eq:cppd}
\end{align}
where $\cC_1' = \cC \cap \cT_{\pi}$ and $\cC_2' = \cC \cap \cT_{\pi'}$.
$R_{\pi_2}$ and $R_{\pi'_2}$ should have the same effect on their respective wire spaces.
Rigorously speaking, 
for natural isomorphisms
\begin{equation}
\begin{aligned}
    &\ket{p_{\tdD}} \cong \ket{p_{\htD}}, \ket{q_{\tdD}} \cong \ket{q_{\htD}}, \\
    &\ket{k_{\tdB}} \cong \ket{k_{\htB}}, \ket{l_{\tdB}} \cong \ket{l_{\htB}}, \\
    &\ket{\psi_{\tdF}} \cong \ket{\psi_{\hat{F}}} ,\ket{\varphi_{\tdF}} \cong \ket{\varphi_{\htF}},
\end{aligned}
\end{equation}
we have
\begin{align}
    R_{\pi'_2}^{\cT_{\pi'}}
    ( p_{\tdD} \otimes \psi_{\tdF} \otimes \bar{k}_{\tdB}, 
    q_{\tdD} \otimes \varphi_{\tdF} \otimes \bar{l}_{\tdB})
    =  R_{\pi_2}^{\cT_{\pi}}
    ( p_{\htD} \otimes \psi_{\htF} \otimes \bar{k}_{\htB},  
    q_{\htD} \otimes \varphi_{\htF} \otimes \bar{l}_{\htB}). \label{eq:SAM}
\end{align}
Substituting \Cref{eq:SAM} into \Cref{eq:cppd}, we have
\begin{align}
    R_{\bar{\pi}_1 \bar{\pi}'_2} ^{\cT_{(\pi_1\|\pi_2) \theta}} 
    ( \psi_{\htE} \otimes \psi_{\tdF}, 
    \varphi_{\htE} \otimes \varphi_{\tdF}) = 
    R_{\pi_1 \pi_2} ^{\cT_{\pi}}
    ( \psi_{\htE} \otimes \psi_{\htF}, 
    \varphi_{\htE} \otimes \varphi_{\htF} ),
\end{align}
which concludes the proof.

\subsection{Proof of \Cref{thm:cond-gs}}
\label{prf:cond-gs}
\begin{theorem*}
     Let $\pi$ be a $\mathrm{WCF}(z,\epsilon)$ protocol and $\eta$ be any protocol with only classical output at the outer interface. $\xi = \pi \| \eta$ is their parallel composition with the ordered set $\cT$. Then $\pi$ is globally secure in the protocol if there exists a partition $\cT_1,\cT_2,\cT_3$ of $\cT$ such that
    \begin{enumerate}[start=1, label={(\bfseries C\arabic*)}]
    \item $\forall t_1\in \cT_1, t_2 \in \cT_2, t_3 \in \cT_3$, $t_1 \prec t_2 \prec t_3$;
    \item $\eta$ only accepts input and produces output on set $\cT_1 \cup \cT_3$, while $\pi$ only accepts input and produces output on set $\cT_2$. 
    \end{enumerate}
\end{theorem*}
In this proof, we focus on the first condition in \Cref{def:global-sec}, then the second condition follows similarly.

Let $\alpha$ be the adversarial box connected to the parallel composition $\pi_A \| \eta_A$, as shown in \Cref{fig:secure-basic}. For convenience, $c_A,K_A$ in \Cref{def:global-sec} are relabeled as $E,K$, respectively.
The idea of this proof is by contradiction, explained as follows.
Our target is to bound the conditional probability 
$\Pr[E=1|K=k]$.
Due to condition (\textbf{C1}) and (\textbf{C2}), the interaction of $\alpha$ divides into three phases (see \Cref{fig:secure-equiv}). 
Observe that the only way that $K$ can influence $E$ is through the information $Q_1$ passed from the first phase $\alpha_1$ to the second phase $\alpha_2$. 
If the adversarial box $\alpha$ can bias $\Pr[E=1|K=k]$ beyond $\frac{1}{2} + \epsilon$,
then the information on $Q_1$ conditioned on $K=k$, together with the box $\alpha_2$, becomes a successful attacking strategy for $\pi_A$ alone.

We first simplify the adversarial box and prove this theorem by contradiction. 
Let $\cT$ be the ordered set of $\pi_A \| \eta_A$.
Without loss of generality, we can assume the final box $(\pi_A \| \eta_A) \alpha$ is defined on the ordered set $\cT$, because any extra element $t$ introduced by $\alpha$ has no effect on $\pi_A \| \eta_A$.

\begin{figure}[ht]
    \centering
    \begin{minipage}{.5\textwidth}
        \centering
        \subcaptionbox{\label{fig:secure-basic}}{\input{pic/secure-basic}}
    \end{minipage}%
    \begin{minipage}{.5\textwidth}
        \centering
        \subcaptionbox{\label{fig:secure-equiv}}{\input{pic/secure-equiv}}
    \end{minipage}
    \begin{minipage}{\textwidth}
        \centering
        \subcaptionbox{\label{fig:secure-order}}{\input{pic/secure-order}}
    \end{minipage}
    
    \caption{
    (a) The basic connection of box $(\pi_A\|\eta_A) \alpha$. 
    (b) An equivalent form of the box. The wires $X,Y,V,W$ are split into $(X_1,X_2) , (Y_1,Y_2), (V_1,V_2), (W_1,W_2)$, respectively. 
    (c) Causal order of box $(\pi_A\|\eta_A) \alpha$.
    The sets $\cT_1,\cT_2,\cT_3$ as a whole are strictly ordered.}
\end{figure}

According to the condition (\textbf{C}1) and (\textbf{C}2), the wire $(V_1,W_1)$, $(B,D)$ and $(V_2,W_2)$ are effective only on $\cT_1$, $\cT_2$ and $\cT_3$, respectively.
Let $\{ \Phi_{\alpha}^{\cC}\}_{\cC \in \bar{\fC}(\cT)}$ be the set of CPTP maps corresponding to $\alpha$.
Since we are only interested in the final output, we set $\cC =\cT$.
By~\cite[Lemma 9]{portmann2017causal}, the map $ \Phi_{\alpha}^{\cT}$ can be decomposed into three maps $\Phi_{\alpha_1}^{\cT}:V_1 \mapsto W_1Q_1$, $ \Phi_{\alpha_2}^{\cT}:BQ_1 \mapsto DQ_2$ and $\Phi_{\alpha_3}^{\cT}:V_2Q_2 \mapsto W_2$ that satisfy
\begin{align}
    \Phi_{\alpha}^{\cT} = (\id_{W_1} \otimes \id_{D} \otimes \Phi_{\alpha_3}^{\cT}) \circ (\id_{W_1} \otimes \Phi_{\alpha_2}^{\cT} \otimes \id_{V_2}) \circ (\Phi_{\alpha_1}^{\cT} \otimes \id_{B} \otimes \id_{V_2}).
\end{align}
Here $\Phi_{\alpha_1}^{\cT}$, $\Phi_{\alpha_2}^{\cT}$ and $\Phi_{\alpha_3}^{\cT}$ can be seen as the map of some boxes $\alpha_1$, $\alpha_2$ and $\alpha_3$, respectively.
Similarly, the box $\eta_A$ is split into $\eta_{A1}$ and $\eta_{A2}$. 
The decomposition is illustrated in \Cref{fig:secure-equiv}.

Let $\Phi_{\eta_A \alpha_1}^{\cT}$ and $\Phi_{\pi_A \alpha_2}^{\cT}$ be the map of connected boxes $\eta_A \alpha_1$ and $\pi_A \alpha_2$.
Denote $\bar{Q} = \cF_{Q}^{\cT}$, then the output of $\Phi_{\eta_A \alpha_1}^{\cT}$ is a classical-quantum state $\rho_{\bar{Q}K} = \sum_k P_K(k)\proj{k} \otimes \rho_{\bar{Q}}^k$.
$ \Phi_{\pi_A \alpha_2}^{\cT}$ is a map from $\fT(\bar{Q})$ to the classical bit $E$, which can be represented as a binary POVM $\{M_0,M_1\}$.
Then the conditional probability is $P_{E|K}(e|k) = \tr(M_e \rho_{\bar{Q}}^k)$.

Now we are ready to prove the theorem by contradiction.
Assume $\exists k \in \mathbb{N}, \Pr[E=1 | K =k] \ge z + \epsilon$, then $ \tr(M_e \rho_{\bar{Q}}^k) \ge z + \epsilon$.
Consider a box $\gamma$ that outputs state $\rho_{\bar{Q}}^k$ and has no input. The box $\pi_A \alpha_2 \gamma$ will output $E=1$ with probability larger than $z + \epsilon$, which contradicts the stand-alone security. 
Therefore, the first condition in \Cref{def:global-sec} is satisfied.

\subsection{Proof of \Cref{coro:unb}}
\Label{prf:unb}
\begin{corollary*}
      Let $z\in[0,\frac{1}{2}]$. The unbalanced WCF protocol shown in \Cref{fig:unb-wcf} is a $\mathrm{WCF}(z,\epsilon')$ protocol with $\epsilon' = 2\epsilon + o(\epsilon)$.
\end{corollary*}

The intuition for the proof of \Cref{coro:unb} is as follows.
The probability of one party ultimately winning is a continuous function of the probabilities of winning in each round.
Thus, a small bias in each round necessarily leads to a small overall bias.

Denote the fractional part of $z$ as a binary string $b_1^n$ (with finite precision), that is, $z = 0.b_1b_2\dots b_n$. 
Define the sets $I \coloneqq \{i \in [n] \mid b_i = 0 \}$ and $I' \coloneqq \{i \in [n] \mid b_i = 1\}$.  
The elements of $I$ and $I'$ form the sequences $(a_1, \dots, a_{|I|})$ and $(a'_1, \dots, a'_{|I'|})$, respectively, arranged in ascending order.
The $i$th WCF protocol $\pi_i$ as well as the other WCF protocols $\pi_1\|...\|\pi_{i-1}$ and $\pi_{i+1}\|...\|\pi_n$ output classical bits and follow strict temporal orders. 
Therefore, for $i=2,3,\dots, n$ and any bit sequence $c_1^{i-1}$, we have global security of $\pi_i$ according to \Cref{thm:cond-gs}:
\begin{align}
    \Pr \left[ C_{A,i} = 1 | C_{A,1}^{i-1} = c_1^{i-1} \right] &\leq \frac{1}{2}+\epsilon, \label{eq:gl} \\
    \Pr \left[ C_{B,i} = 0 | C_{B,1}^{i-1} = c_1^{i-1} \right] &\leq \frac{1}{2}+\epsilon, \label{eq:gl-b}
\end{align}
where $C_{B(A),1}^i = C_{B(A),1} C_{B(A),2} \dots C_{B(A),i}$.

We first argue that a cheating party needs to maximize the winning probability of each game $\pi_i$ in order to maximize the final winning probability. 
At each step $\pi_i$, if $b_i=0$, Bob wins the final game with certainty if $C_{A,i} =1$ and wins the final game with a probability smaller than $1$ if $C_{A,i}=0$; if $b_i=1$, Bob wins with non-zero probability if $C_{A,i}=0$ while loses with certainty if $C_{A,i}=1$. 
Therefore, a cheating Bob strives to win at each round of $\pi_i$ in order to finally win with the highest probability. 
The similar argument holds for cheating Alice. 
In conclusion, 
Bob's winning probability is no larger than the case when
\begin{align}
    \forall i \in [n], \Pr \left[ C_{A,i} = 1 | C_{A,1}^{i-1} = b_1^{i-1} \right] = \frac{1}{2}+\epsilon, \label{eq:Bcon}
\end{align}
and Alice's winning probability is no larger than the case when
\begin{align}
    \forall i \in [n], \Pr \left[ C_{B,i} = 0 | C_{B,1}^{i-1} = b_1^{i-1} \right] = \frac{1}{2}+\epsilon. \label{eq:Acon}
\end{align}
We will use the above probabilities to compute the upper bound on Alice and Bob's optimal winning probability. 

According to the program of $\tau_A$ in \Cref{fig:unb-wcf}, Bob can win the final game only if the loop breaks at the $a$th round where $b_a=0$, or equivalently, $a \in I$.
Then, his winning probability can be derived as 
\begin{align}
    \Pr[\text{Bob wins}] &= \sum_{a \in I} \Pr[\text{Loop breaks at $a$th round and $C_{A,a}=1$}]\\
    &= \sum_{i =1}^{|I|} \Pr \left[C_{A,1}^{a_i} = (b_1^{a_i-1},1) \right]\\
    &= \sum_{i=1}^{|I|} \Pr \left[ C_{A,a_i} = 1 | C_{A,1}^{a_i-1} = b_1^{a_i-1} \right] \cdot \prod_{j=1}^{a_i-1} \Pr \left[ C_{A,j} = b_j|C_{A,1}^{j-1} = b_1^{j-1} \right] \\
    & = \sum_{i=1}^{|I|} 
    \Pr \left[ C_{A,a_i} = 1 | C_{A,1}^{a_i-1} = b_1^{a_i-1} \right] \cdot
    \prod_{j \in \{ {a}_{k}\}_{k=1}^{i-1}} 
    \Pr \left[ C_{A,j} = 0 | C_{A,1}^{j-1} = b_1^{j-1} \right] \cdot \notag \\
    & \hphantom{=} 
    \prod_{j' \in [a_i-1] \backslash \{ {a}_{k}\}_{k=1}^{i-1}}  
    \Pr \left[ C_{A,j'} = 0 | C_{A,1}^{j'-1} = b_1^{j'-1} \right] \\
    & \le \sum_{i=1}^{|I|} 
    \left(\frac{1}{2} + \epsilon \right)^{a_i-i+1} \left( \frac{1}{2} - \epsilon \right)^{i-1}, \label{eq:coB}
\end{align}
where $C_{A,1}^0 = b_1^0$ is defined as a certain event and Inequality~\eqref{eq:coB} follows from Inequality~\eqref{eq:gl} and~\eqref{eq:Bcon}.
Similarly, for cheating Alice, we obtain
\begin{align}
    \Pr[\text{Alice wins}] &= \sum_{i =1}^{|I'|} \Pr \left[C_{B,1}^{a'_i} = (b_1^{a'_i-1},0) \right]\\
    &= \sum_{i=1}^{|I'|} \Pr \left[ C_{B,a'_i} = 0 | C_{B,i}^{a'_i-1} = b_1^{a'_i-1} \right] \prod_{j=1}^{a'_i-1} \Pr \left[ C_{B,j} = b_j|C_{B,1}^{j-1} = b_1^{j-1} \right] \label{eqn:alice}\\
    & \le \sum_{i=1}^{|I'|} \left(\frac{1}{2} + \epsilon \right)^{a'_i-i+1} \left( \frac{1}{2} - \epsilon \right)^{i-1} , \label{eq:coA}
\end{align}
Then the bias achieved by Bob is
\begin{align}
    \epsilon'_B & = |\Pr[\text{Bob wins}] - (1 - z) | \\
    &\le \sum_{i=1}^{|I|}   \left(\frac{1}{2} + \epsilon \right)^{a_i-i+1} \left( \frac{1}{2} - \epsilon \right)^{i-1} - \frac{1}{2^{a_i}} \\
    &\le  \sum_{i=1}^{|I|} \frac{2(a_i-2i+2)}{2^{a_i}} \epsilon  + o(\epsilon).
\end{align}
We have omitted the absolute value because the winning probability of a cheating Bob is always larger than that of an honest Bob, $\Pr[\text{Bob wins}] \geq 1-z$. Note that
\begin{align}
    \sum_{n=i}^{\infty} \frac{n}{2^n} =  \frac{i+1}{2^{i-1}}.
\end{align}
When $|I|=1$, 
\begin{align}
    \epsilon'_B \leq \frac{2a_1}{2^{a_1}} \epsilon + o(\epsilon) \leq \epsilon + o(\epsilon).
\end{align}
When $|I|\geq 2$,
\begin{align}
    \epsilon'_B \leq \frac{2a_1}{2^{a_1}} \epsilon + \frac{1}{2} \sum_{i=2}^{|I|} \frac{(a_i-2)}{2^{a_i-2}} \epsilon +o(\epsilon) \leq  \frac{2a_1}{2^{a_1}} \epsilon +\frac{1}{2} \sum_{i=2}^{\infty} \frac{i-2}{2^{i-2}}\epsilon + o(\epsilon) \leq 2\epsilon + o(\epsilon).
\end{align}
In conclusion, for cheating Bob, 
\begin{align}
    \epsilon'_B = |\Pr[\text{Bob wins}] - (1-z) | \le 2\epsilon +o(\epsilon).
\end{align}
Similarly, for cheating Alice, 
\begin{align}
    \epsilon'_A = |\Pr[\text{Alice wins}] - z | 
    \le 2\epsilon +o(\epsilon).
\end{align}
The final bias is 
\begin{align}
    \epsilon' = \max\{\epsilon'_A, \epsilon'_B\} \le 2 \epsilon + o(\epsilon).
\end{align}

\end{document}

%% file: pic/comb.tex
\tikzset{every picture/.style={line width=0.75pt}} 

\begin{tikzpicture}[x=0.75pt,y=0.75pt,yscale=-1,xscale=1]

\draw   (114,42) -- (148,42) -- (148,76.87) -- (114,76.87) -- cycle ;
\draw    (139,78) -- (165.93,104.93) ;
\draw [shift={(167.34,106.34)}, rotate = 225] [color={rgb, 255:red, 0; green, 0; blue, 0 }  ][line width=0.75]    (10.93,-3.29) .. controls (6.95,-1.4) and (3.31,-0.3) .. (0,0) .. controls (3.31,0.3) and (6.95,1.4) .. (10.93,3.29)   ;
\draw    (148,58) -- (210,58) ;
\draw [shift={(212,58)}, rotate = 180] [color={rgb, 255:red, 0; green, 0; blue, 0 }  ][line width=0.75]    (10.93,-3.29) .. controls (6.95,-1.4) and (3.31,-0.3) .. (0,0) .. controls (3.31,0.3) and (6.95,1.4) .. (10.93,3.29)   ;
\draw    (96,106) -- (122.59,79.41) ;
\draw [shift={(124,78)}, rotate = 135] [color={rgb, 255:red, 0; green, 0; blue, 0 }  ][line width=0.75]    (10.93,-3.29) .. controls (6.95,-1.4) and (3.31,-0.3) .. (0,0) .. controls (3.31,0.3) and (6.95,1.4) .. (10.93,3.29)   ;
\draw   (214,42) -- (248,42) -- (248,76.87) -- (214,76.87) -- cycle ;
\draw    (239,78) -- (265.59,104.59) ;
\draw [shift={(267,106)}, rotate = 225] [color={rgb, 255:red, 0; green, 0; blue, 0 }  ][line width=0.75]    (10.93,-3.29) .. controls (6.95,-1.4) and (3.31,-0.3) .. (0,0) .. controls (3.31,0.3) and (6.95,1.4) .. (10.93,3.29)   ;
\draw    (248,58) -- (283,58) ;
\draw [shift={(285,58)}, rotate = 180] [color={rgb, 255:red, 0; green, 0; blue, 0 }  ][line width=0.75]    (10.93,-3.29) .. controls (6.95,-1.4) and (3.31,-0.3) .. (0,0) .. controls (3.31,0.3) and (6.95,1.4) .. (10.93,3.29)   ;
\draw    (199,107) -- (226.59,79.41) ;
\draw [shift={(228,78)}, rotate = 135] [color={rgb, 255:red, 0; green, 0; blue, 0 }  ][line width=0.75]    (10.93,-3.29) .. controls (6.95,-1.4) and (3.31,-0.3) .. (0,0) .. controls (3.31,0.3) and (6.95,1.4) .. (10.93,3.29)   ;
\draw   (386,41) -- (420,41) -- (420,75.87) -- (386,75.87) -- cycle ;
\draw    (411,77) -- (437.59,103.59) ;
\draw [shift={(439,105)}, rotate = 225] [color={rgb, 255:red, 0; green, 0; blue, 0 }  ][line width=0.75]    (10.93,-3.29) .. controls (6.95,-1.4) and (3.31,-0.3) .. (0,0) .. controls (3.31,0.3) and (6.95,1.4) .. (10.93,3.29)   ;
\draw    (369.66,103.34) -- (394.59,78.41) ;
\draw [shift={(396,77)}, rotate = 135] [color={rgb, 255:red, 0; green, 0; blue, 0 }  ][line width=0.75]    (10.93,-3.29) .. controls (6.95,-1.4) and (3.31,-0.3) .. (0,0) .. controls (3.31,0.3) and (6.95,1.4) .. (10.93,3.29)   ;
\draw    (350,57) -- (383,57) ;
\draw [shift={(385,57)}, rotate = 180] [color={rgb, 255:red, 0; green, 0; blue, 0 }  ][line width=0.75]    (10.93,-3.29) .. controls (6.95,-1.4) and (3.31,-0.3) .. (0,0) .. controls (3.31,0.3) and (6.95,1.4) .. (10.93,3.29)   ;

\draw (131,59.44) node    {$\mathcal{E}_{1}$};
\draw (231,59.44) node    {$\mathcal{E}_{2}$};
\draw (403,58.44) node    {$\mathcal{E}_{n}$};
\draw (322.26,55.5) node   [align=left] {...};

\end{tikzpicture}

%% file: pic/CFConverters.tex
\tikzset{every picture/.style={line width=0.75pt}} 

\begin{tikzpicture}[x=0.75pt,y=0.75pt,yscale=-1,xscale=1]

\draw   (367.17,92.08) -- (444.42,92.08) -- (444.42,169.33) -- (367.17,169.33) -- cycle ;
\draw   (122.5,91.75) -- (199.75,91.75) -- (199.75,169) -- (122.5,169) -- cycle ;
\draw    (489.42,130.83) -- (445.42,130.83) ;
\draw [shift={(491.42,130.83)}, rotate = 180] [color={rgb, 255:red, 0; green, 0; blue, 0 }  ][line width=0.75]    (10.93,-3.29) .. controls (6.95,-1.4) and (3.31,-0.3) .. (0,0) .. controls (3.31,0.3) and (6.95,1.4) .. (10.93,3.29)   ;
\draw    (79.4,129.5) -- (121.25,129.5) ;
\draw [shift={(77.4,129.5)}, rotate = 0] [color={rgb, 255:red, 0; green, 0; blue, 0 }  ][line width=0.75]    (10.93,-3.29) .. controls (6.95,-1.4) and (3.31,-0.3) .. (0,0) .. controls (3.31,0.3) and (6.95,1.4) .. (10.93,3.29)   ;
\draw    (243.75,109.5) -- (199.75,109.5) ;
\draw [shift={(245.75,109.5)}, rotate = 180] [color={rgb, 255:red, 0; green, 0; blue, 0 }  ][line width=0.75]    (10.93,-3.29) .. controls (6.95,-1.4) and (3.31,-0.3) .. (0,0) .. controls (3.31,0.3) and (6.95,1.4) .. (10.93,3.29)   ;
\draw    (363.42,110.83) -- (319.42,110.83) ;
\draw [shift={(365.42,110.83)}, rotate = 180] [color={rgb, 255:red, 0; green, 0; blue, 0 }  ][line width=0.75]    (10.93,-3.29) .. controls (6.95,-1.4) and (3.31,-0.3) .. (0,0) .. controls (3.31,0.3) and (6.95,1.4) .. (10.93,3.29)   ;
\draw    (202.4,153.5) -- (244.25,153.5) ;
\draw [shift={(200.4,153.5)}, rotate = 0] [color={rgb, 255:red, 0; green, 0; blue, 0 }  ][line width=0.75]    (10.93,-3.29) .. controls (6.95,-1.4) and (3.31,-0.3) .. (0,0) .. controls (3.31,0.3) and (6.95,1.4) .. (10.93,3.29)   ;
\draw    (325.06,154.83) -- (366.92,154.83) ;
\draw [shift={(323.06,154.83)}, rotate = 0] [color={rgb, 255:red, 0; green, 0; blue, 0 }  ][line width=0.75]    (10.93,-3.29) .. controls (6.95,-1.4) and (3.31,-0.3) .. (0,0) .. controls (3.31,0.3) and (6.95,1.4) .. (10.93,3.29)   ;

\draw (407.11,130.9) node  [font=\large] [align=left] {$\displaystyle \pi _{B}$};
\draw (162.28,130.57) node  [font=\large] [align=left] {$\displaystyle \pi _{A}$};
\draw (479.84,114.64) node   [align=left] {($\displaystyle c_{B} ,t_{b}$)};
\draw (73.18,114.3) node   [align=left] {$\displaystyle ( c_{A} ,t_{a})$};
\draw (237.72,96.94) node  [font=\footnotesize] [align=left] {$\displaystyle \ket{\psi } \in \mathcal{F}_{X}^{\mathcal{T}}$};
\draw (331.05,96.28) node  [font=\footnotesize] [align=left] {$\displaystyle \ket{\psi '} \in \mathcal{F}_{X'}^{\mathcal{T}}$};
\draw (239.05,138.94) node  [font=\footnotesize] [align=left] {$\displaystyle \ket{\phi } \in \mathcal{F}_{Y}^{\mathcal{T}}$};
\draw (327.72,139.61) node  [font=\footnotesize] [align=left] {$\displaystyle \ket{\phi '} \in \mathcal{F}_{Y'}^{\mathcal{T}}$};

\end{tikzpicture}

%% file: pic/WCFres2.tex
\tikzset{every picture/.style={line width=0.75pt}} 

\begin{tikzpicture}[x=0.75pt,y=0.75pt,yscale=-1,xscale=1]

\draw   (87.84,104.32) -- (138.89,104.32) -- (138.89,162.8) -- (87.84,162.8) -- cycle ;
\draw    (87.24,133.67) -- (62.88,133.67) ;
\draw [shift={(60.88,133.67)}, rotate = 360] [color={rgb, 255:red, 0; green, 0; blue, 0 }  ][line width=0.75]    (10.93,-3.29) .. controls (6.95,-1.4) and (3.31,-0.3) .. (0,0) .. controls (3.31,0.3) and (6.95,1.4) .. (10.93,3.29)   ;
\draw    (139.24,133.67) -- (164.88,133.67) ;
\draw [shift={(166.88,133.67)}, rotate = 180] [color={rgb, 255:red, 0; green, 0; blue, 0 }  ][line width=0.75]    (10.93,-3.29) .. controls (6.95,-1.4) and (3.31,-0.3) .. (0,0) .. controls (3.31,0.3) and (6.95,1.4) .. (10.93,3.29)   ;
\draw   (298.17,105.32) -- (349.22,105.32) -- (349.22,163.8) -- (298.17,163.8) -- cycle ;
\draw    (298.24,118) -- (264.82,118) ;
\draw [shift={(262.82,118)}, rotate = 360] [color={rgb, 255:red, 0; green, 0; blue, 0 }  ][line width=0.75]    (10.93,-3.29) .. controls (6.95,-1.4) and (3.31,-0.3) .. (0,0) .. controls (3.31,0.3) and (6.95,1.4) .. (10.93,3.29)   ;
\draw    (348.9,153.67) -- (386.82,153.67) ;
\draw [shift={(388.82,153.67)}, rotate = 180] [color={rgb, 255:red, 0; green, 0; blue, 0 }  ][line width=0.75]    (10.93,-3.29) .. controls (6.95,-1.4) and (3.31,-0.3) .. (0,0) .. controls (3.31,0.3) and (6.95,1.4) .. (10.93,3.29)   ;
\draw    (264.82,153.33) -- (295.88,153.33) ;
\draw [shift={(297.88,153.33)}, rotate = 180] [color={rgb, 255:red, 0; green, 0; blue, 0 }  ][line width=0.75]    (10.93,-3.29) .. controls (6.95,-1.4) and (3.31,-0.3) .. (0,0) .. controls (3.31,0.3) and (6.95,1.4) .. (10.93,3.29)   ;
\draw   (478.17,104.52) -- (529.22,104.52) -- (529.22,163) -- (478.17,163) -- cycle ;
\draw    (477.57,152.67) -- (445.82,152.67) ;
\draw [shift={(443.82,152.67)}, rotate = 360] [color={rgb, 255:red, 0; green, 0; blue, 0 }  ][line width=0.75]    (10.93,-3.29) .. controls (6.95,-1.4) and (3.31,-0.3) .. (0,0) .. controls (3.31,0.3) and (6.95,1.4) .. (10.93,3.29)   ;
\draw    (528.95,113.33) -- (562.82,113.33) ;
\draw [shift={(564.82,113.33)}, rotate = 180] [color={rgb, 255:red, 0; green, 0; blue, 0 }  ][line width=0.75]    (10.93,-3.29) .. controls (6.95,-1.4) and (3.31,-0.3) .. (0,0) .. controls (3.31,0.3) and (6.95,1.4) .. (10.93,3.29)   ;
\draw    (565.82,151) -- (531.88,151) ;
\draw [shift={(529.88,151)}, rotate = 360] [color={rgb, 255:red, 0; green, 0; blue, 0 }  ][line width=0.75]    (10.93,-3.29) .. controls (6.95,-1.4) and (3.31,-0.3) .. (0,0) .. controls (3.31,0.3) and (6.95,1.4) .. (10.93,3.29)   ;
\draw [color={rgb, 255:red, 113; green, 113; blue, 113 }  ,draw opacity=1 ] [dash pattern={on 3.75pt off 3pt on 3.75pt off 3pt}]  (196.17,94.02) -- (196.17,165.36) ;
\draw [color={rgb, 255:red, 113; green, 113; blue, 113 }  ,draw opacity=1 ] [dash pattern={on 3.75pt off 3pt on 3.75pt off 3pt}]  (418.5,94.69) -- (418.5,163.36) ;

\draw (113.36,133.56) node   [align=left] {$\displaystyle S$};
\draw (66.55,121.13) node   [align=left] {$\displaystyle ( c,t_{a})$};
\draw (162.21,121.3) node   [align=left] {$\displaystyle ( c,t_{b})$};
\draw (323.69,134.56) node   [align=left] {$\displaystyle S_{A}$};
\draw (274.98,105.3) node   [align=left] {$\displaystyle ( c',t_{0} ')$};
\draw (374.02,141.3) node   [align=left] {$\displaystyle ( c_{B} ,t_{b})$};
\draw (259.55,140.3) node   [align=left] {$\displaystyle (( b',p') ,t_{1} ')$};
\draw (505.64,135.32) node   [align=left] {$\displaystyle S_{B}$};
\draw (554.65,100.63) node   [align=left] {$\displaystyle ( c'',t_{0} '')$};
\draw (452.19,139.8) node   [align=left] {$\displaystyle ( c_{A} ,t_{a})$};
\draw (572.21,138.63) node   [align=left] {$\displaystyle (( b'',p'') ,t_{2} '')$};
\draw (112.17,174.25) node   [align=left] {(a)};
\draw (324.67,174.25) node   [align=left] {(b)};
\draw (501.67,174.25) node   [align=left] {(c)};

\end{tikzpicture}

%% file: pic/WCFcom2.tex
\tikzset{every picture/.style={line width=0.75pt}} 

\begin{tikzpicture}[x=0.75pt,y=0.75pt,yscale=-1,xscale=1]

\draw   (61.4,76.25) -- (114.28,76.25) -- (114.28,137.24) -- (61.4,137.24) -- cycle ;
\draw   (165.26,76.25) -- (218.14,76.25) -- (218.14,137.24) -- (165.26,137.24) -- cycle ;
\draw   (274.12,76.25) -- (327,76.25) -- (327,137.24) -- (274.12,137.24) -- cycle ;
\draw    (61.57,121.47) -- (37.82,121.47) ;
\draw [shift={(35.82,121.47)}, rotate = 360] [color={rgb, 255:red, 0; green, 0; blue, 0 }  ][line width=0.75]    (10.93,-3.29) .. controls (6.95,-1.4) and (3.31,-0.3) .. (0,0) .. controls (3.31,0.3) and (6.95,1.4) .. (10.93,3.29)   ;
\draw    (327,121.47) -- (349.82,121.47) ;
\draw [shift={(351.82,121.47)}, rotate = 180] [color={rgb, 255:red, 0; green, 0; blue, 0 }  ][line width=0.75]    (10.93,-3.29) .. controls (6.95,-1.4) and (3.31,-0.3) .. (0,0) .. controls (3.31,0.3) and (6.95,1.4) .. (10.93,3.29)   ;
\draw    (114.82,87.23) -- (163.09,87.23) ;
\draw [shift={(165.09,87.23)}, rotate = 180] [color={rgb, 255:red, 0; green, 0; blue, 0 }  ][line width=0.75]    (10.93,-3.29) .. controls (6.95,-1.4) and (3.31,-0.3) .. (0,0) .. controls (3.31,0.3) and (6.95,1.4) .. (10.93,3.29)   ;
\draw    (273.82,86.81) -- (220.14,86.81) ;
\draw [shift={(218.14,86.81)}, rotate = 360] [color={rgb, 255:red, 0; green, 0; blue, 0 }  ][line width=0.75]    (10.93,-3.29) .. controls (6.95,-1.4) and (3.31,-0.3) .. (0,0) .. controls (3.31,0.3) and (6.95,1.4) .. (10.93,3.29)   ;
\draw    (116.82,119.84) -- (165.09,119.84) ;
\draw [shift={(114.82,119.84)}, rotate = 0] [color={rgb, 255:red, 0; green, 0; blue, 0 }  ][line width=0.75]    (10.93,-3.29) .. controls (6.95,-1.4) and (3.31,-0.3) .. (0,0) .. controls (3.31,0.3) and (6.95,1.4) .. (10.93,3.29)   ;
\draw    (218.46,119.06) -- (271.82,119.06) ;
\draw [shift={(273.82,119.06)}, rotate = 180] [color={rgb, 255:red, 0; green, 0; blue, 0 }  ][line width=0.75]    (10.93,-3.29) .. controls (6.95,-1.4) and (3.31,-0.3) .. (0,0) .. controls (3.31,0.3) and (6.95,1.4) .. (10.93,3.29)   ;
\draw   (435.56,75.24) -- (484.38,75.24) -- (484.38,134.06) -- (435.56,134.06) -- cycle ;
\draw    (434.73,104.67) -- (410.38,104.67) ;
\draw [shift={(408.38,104.67)}, rotate = 360] [color={rgb, 255:red, 0; green, 0; blue, 0 }  ][line width=0.75]    (10.93,-3.29) .. controls (6.95,-1.4) and (3.31,-0.3) .. (0,0) .. controls (3.31,0.3) and (6.95,1.4) .. (10.93,3.29)   ;
\draw    (484.73,103.67) -- (510.38,103.67) ;
\draw [shift={(512.38,103.67)}, rotate = 180] [color={rgb, 255:red, 0; green, 0; blue, 0 }  ][line width=0.75]    (10.93,-3.29) .. controls (6.95,-1.4) and (3.31,-0.3) .. (0,0) .. controls (3.31,0.3) and (6.95,1.4) .. (10.93,3.29)   ;

\draw (88.82,107.21) node   [align=left] {$\displaystyle S_{B}$};
\draw (301.83,107.21) node   [align=left] {$\displaystyle S_{A}$};
\draw (193.26,106.4) node   [align=left] {$\displaystyle \sigma $};
\draw (458.86,103.56) node   [align=left] {$\displaystyle S$};
\draw (413.04,91.63) node   [align=left] {$\displaystyle c$};
\draw (505.71,91.3) node   [align=left] {$\displaystyle c$};
\draw (48.04,110.63) node   [align=left] {$\displaystyle c_{A}$};
\draw (137.48,77.47) node   [align=left] {$\displaystyle c''$};
\draw (141.04,108.3) node   [align=left] {$\displaystyle ( b'',p'')$};
\draw (246.14,77.47) node   [align=left] {$\displaystyle c'$};
\draw (339.04,110.63) node   [align=left] {$\displaystyle c_{B}$};
\draw (355,96.4) node [anchor=north west][inner sep=0.75pt]    {$\approx _{3\delta }$};
\draw (246.04,108.3) node   [align=left] {$\displaystyle ( b',p')$};

\end{tikzpicture}

%% file: pic/insecure.tex
\tikzset{every picture/.style={line width=0.75pt}} 

\begin{tikzpicture}[x=0.75pt,y=0.75pt,yscale=-1,xscale=1]

\draw   (171.08,70.42) -- (221.08,70.42) -- (221.08,120.42) -- (171.08,120.42) -- cycle ;
\draw    (119.87,95.83) -- (171.58,95.83) ;
\draw [shift={(117.87,95.83)}, rotate = 0] [color={rgb, 255:red, 0; green, 0; blue, 0 }  ][line width=0.75]    (10.93,-3.29) .. controls (6.95,-1.4) and (3.31,-0.3) .. (0,0) .. controls (3.31,0.3) and (6.95,1.4) .. (10.93,3.29)   ;
\draw    (252.33,85.5) -- (220.75,85.5) ;
\draw [shift={(254.33,85.5)}, rotate = 180] [color={rgb, 255:red, 0; green, 0; blue, 0 }  ][line width=0.75]    (10.93,-3.29) .. controls (6.95,-1.4) and (3.31,-0.3) .. (0,0) .. controls (3.31,0.3) and (6.95,1.4) .. (10.93,3.29)   ;
\draw    (223.06,106.83) -- (252.04,106.83) ;
\draw [shift={(221.06,106.83)}, rotate = 0] [color={rgb, 255:red, 0; green, 0; blue, 0 }  ][line width=0.75]    (10.93,-3.29) .. controls (6.95,-1.4) and (3.31,-0.3) .. (0,0) .. controls (3.31,0.3) and (6.95,1.4) .. (10.93,3.29)   ;
\draw   (171.08,149.75) -- (221.08,149.75) -- (221.08,199.76) -- (171.08,199.76) -- cycle ;
\draw    (119.87,175.17) -- (171.58,175.17) ;
\draw [shift={(117.87,175.17)}, rotate = 0] [color={rgb, 255:red, 0; green, 0; blue, 0 }  ][line width=0.75]    (10.93,-3.29) .. controls (6.95,-1.4) and (3.31,-0.3) .. (0,0) .. controls (3.31,0.3) and (6.95,1.4) .. (10.93,3.29)   ;
\draw    (252.33,186.17) -- (220.75,186.17) ;
\draw [shift={(254.33,186.17)}, rotate = 180] [color={rgb, 255:red, 0; green, 0; blue, 0 }  ][line width=0.75]    (10.93,-3.29) .. controls (6.95,-1.4) and (3.31,-0.3) .. (0,0) .. controls (3.31,0.3) and (6.95,1.4) .. (10.93,3.29)   ;
\draw    (223.73,165.5) -- (252.71,165.5) ;
\draw [shift={(221.73,165.5)}, rotate = 0] [color={rgb, 255:red, 0; green, 0; blue, 0 }  ][line width=0.75]    (10.93,-3.29) .. controls (6.95,-1.4) and (3.31,-0.3) .. (0,0) .. controls (3.31,0.3) and (6.95,1.4) .. (10.93,3.29)   ;
\draw   (370.63,149.75) -- (320.63,149.75) -- (320.63,199.76) -- (370.63,199.76) -- cycle ;
\draw    (411.87,175.17) -- (370.13,175.17) ;
\draw [shift={(413.87,175.17)}, rotate = 180] [color={rgb, 255:red, 0; green, 0; blue, 0 }  ][line width=0.75]    (10.93,-3.29) .. controls (6.95,-1.4) and (3.31,-0.3) .. (0,0) .. controls (3.31,0.3) and (6.95,1.4) .. (10.93,3.29)   ;
\draw    (289.38,164.83) -- (320.96,164.83) ;
\draw [shift={(287.38,164.83)}, rotate = 0] [color={rgb, 255:red, 0; green, 0; blue, 0 }  ][line width=0.75]    (10.93,-3.29) .. controls (6.95,-1.4) and (3.31,-0.3) .. (0,0) .. controls (3.31,0.3) and (6.95,1.4) .. (10.93,3.29)   ;
\draw    (318.65,186.17) -- (289.67,186.17) ;
\draw [shift={(320.65,186.17)}, rotate = 180] [color={rgb, 255:red, 0; green, 0; blue, 0 }  ][line width=0.75]    (10.93,-3.29) .. controls (6.95,-1.4) and (3.31,-0.3) .. (0,0) .. controls (3.31,0.3) and (6.95,1.4) .. (10.93,3.29)   ;
\draw   (370.63,71.08) -- (320.63,71.08) -- (320.63,121.09) -- (370.63,121.09) -- cycle ;
\draw    (411.2,96.5) -- (370.13,96.5) ;
\draw [shift={(413.2,96.5)}, rotate = 180] [color={rgb, 255:red, 0; green, 0; blue, 0 }  ][line width=0.75]    (10.93,-3.29) .. controls (6.95,-1.4) and (3.31,-0.3) .. (0,0) .. controls (3.31,0.3) and (6.95,1.4) .. (10.93,3.29)   ;
\draw    (289.38,106.83) -- (320.96,106.83) ;
\draw [shift={(287.38,106.83)}, rotate = 0] [color={rgb, 255:red, 0; green, 0; blue, 0 }  ][line width=0.75]    (10.93,-3.29) .. controls (6.95,-1.4) and (3.31,-0.3) .. (0,0) .. controls (3.31,0.3) and (6.95,1.4) .. (10.93,3.29)   ;
\draw    (317.98,85.5) -- (289,85.5) ;
\draw [shift={(319.98,85.5)}, rotate = 180] [color={rgb, 255:red, 0; green, 0; blue, 0 }  ][line width=0.75]    (10.93,-3.29) .. controls (6.95,-1.4) and (3.31,-0.3) .. (0,0) .. controls (3.31,0.3) and (6.95,1.4) .. (10.93,3.29)   ;
\draw   (73.33,70.82) -- (117.87,70.82) -- (117.87,200.16) -- (73.33,200.16) -- cycle ;
\draw   (414,70.82) -- (458.53,70.82) -- (458.53,200.16) -- (414,200.16) -- cycle ;
\draw    (17.2,134.83) -- (73.63,134.83) ;
\draw [shift={(15.2,134.83)}, rotate = 0] [color={rgb, 255:red, 0; green, 0; blue, 0 }  ][line width=0.75]    (10.93,-3.29) .. controls (6.95,-1.4) and (3.31,-0.3) .. (0,0) .. controls (3.31,0.3) and (6.95,1.4) .. (10.93,3.29)   ;
\draw    (515.87,134.83) -- (458.08,134.83) ;
\draw [shift={(517.87,134.83)}, rotate = 180] [color={rgb, 255:red, 0; green, 0; blue, 0 }  ][line width=0.75]    (10.93,-3.29) .. controls (6.95,-1.4) and (3.31,-0.3) .. (0,0) .. controls (3.31,0.3) and (6.95,1.4) .. (10.93,3.29)   ;
\draw  [color={rgb, 255:red, 65; green, 117; blue, 5 }  ,draw opacity=1 ][dash pattern={on 4.5pt off 4.5pt}] (20.53,63.13) -- (237.2,63.13) -- (237.2,208.47) -- (20.53,208.47) -- cycle ;
\draw  [color={rgb, 255:red, 65; green, 117; blue, 5 }  ,draw opacity=1 ][dash pattern={on 4.5pt off 4.5pt}] (301.2,63.13) -- (510.53,63.13) -- (510.53,208.47) -- (301.2,208.47) -- cycle ;

\draw (196.08,95.42) node  [font=\normalsize] [align=left] {$\displaystyle \pi _{1}$};
\draw (196.08,174.75) node  [font=\normalsize] [align=left] {$\displaystyle \pi '_{2}$};
\draw (345.63,174.75) node  [font=\normalsize] [align=left] {$\displaystyle \pi '_{1}$};
\draw (345.63,96.09) node  [font=\normalsize] [align=left] {$\displaystyle \pi _{2}$};
\draw (144.75,84.28) node  [font=\small]  {$( E,t_{n})$};
\draw (392.72,85.03) node  [font=\small]  {$( G,t_{n})$};
\draw (144.63,162.95) node  [font=\small]  {$( F,t_{n} ')$};
\draw (392.02,162.03) node  [font=\small]  {$( H,t'_{n})$};
\draw (97.31,136.86) node    {$\eta _{A}$};
\draw (437.14,136.86) node    {$\eta _{B}$};
\draw (44.55,120.52) node  [font=\small]  {$( c''_{A} ,t_{A})$};
\draw (486.55,119.86) node  [font=\small]  {$( c''_{B} ,t_{B})$};
\draw (130.75,46.86) node    {$\mu _{A}$};
\draw (390.09,47.52) node    {$\mu _{B}$};
\draw (268.09,85.62) node  [font=\small]  {$X$};
\draw (268.75,107.62) node  [font=\small]  {$Y$};
\draw (270.09,166.28) node  [font=\small]  {$X'$};
\draw (270.75,187.62) node  [font=\small]  {$Y'$};

\end{tikzpicture}

%% file: pic/insecure-order.tex
\tikzset{every picture/.style={line width=0.75pt}} 

\begin{tikzpicture}[x=0.75pt,y=0.75pt,yscale=-1,xscale=1]

\draw    (150.33,85.83) -- (126.42,85.83) ;
\draw [shift={(152.33,85.83)}, rotate = 180] [color={rgb, 255:red, 0; green, 0; blue, 0 }  ][line width=0.75]    (10.93,-3.29) .. controls (6.95,-1.4) and (3.31,-0.3) .. (0,0) .. controls (3.31,0.3) and (6.95,1.4) .. (10.93,3.29)   ;
\draw    (181.81,197.6) -- (152.42,211.83) ;
\draw [shift={(183.61,196.73)}, rotate = 154.17] [color={rgb, 255:red, 0; green, 0; blue, 0 }  ][line width=0.75]    (10.93,-3.29) .. controls (6.95,-1.4) and (3.31,-0.3) .. (0,0) .. controls (3.31,0.3) and (6.95,1.4) .. (10.93,3.29)   ;
\draw    (181.81,186.96) -- (152.42,172.73) ;
\draw [shift={(183.61,187.83)}, rotate = 205.83] [color={rgb, 255:red, 0; green, 0; blue, 0 }  ][line width=0.75]    (10.93,-3.29) .. controls (6.95,-1.4) and (3.31,-0.3) .. (0,0) .. controls (3.31,0.3) and (6.95,1.4) .. (10.93,3.29)   ;
\draw    (181.81,269.6) -- (152.42,283.83) ;
\draw [shift={(183.61,268.73)}, rotate = 154.17] [color={rgb, 255:red, 0; green, 0; blue, 0 }  ][line width=0.75]    (10.93,-3.29) .. controls (6.95,-1.4) and (3.31,-0.3) .. (0,0) .. controls (3.31,0.3) and (6.95,1.4) .. (10.93,3.29)   ;
\draw    (181.81,258.96) -- (152.42,244.73) ;
\draw [shift={(183.61,259.83)}, rotate = 205.83] [color={rgb, 255:red, 0; green, 0; blue, 0 }  ][line width=0.75]    (10.93,-3.29) .. controls (6.95,-1.4) and (3.31,-0.3) .. (0,0) .. controls (3.31,0.3) and (6.95,1.4) .. (10.93,3.29)   ;
\draw    (460.97,108.07) -- (436.33,81.53) ;
\draw [shift={(462.33,109.53)}, rotate = 227.12] [color={rgb, 255:red, 0; green, 0; blue, 0 }  ][line width=0.75]    (10.93,-3.29) .. controls (6.95,-1.4) and (3.31,-0.3) .. (0,0) .. controls (3.31,0.3) and (6.95,1.4) .. (10.93,3.29)   ;
\draw    (202.33,85.83) -- (178.42,85.83) ;
\draw [shift={(204.33,85.83)}, rotate = 180] [color={rgb, 255:red, 0; green, 0; blue, 0 }  ][line width=0.75]    (10.93,-3.29) .. controls (6.95,-1.4) and (3.31,-0.3) .. (0,0) .. controls (3.31,0.3) and (6.95,1.4) .. (10.93,3.29)   ;
\draw    (152.33,135.83) -- (128.42,135.83) ;
\draw [shift={(154.33,135.83)}, rotate = 180] [color={rgb, 255:red, 0; green, 0; blue, 0 }  ][line width=0.75]    (10.93,-3.29) .. controls (6.95,-1.4) and (3.31,-0.3) .. (0,0) .. controls (3.31,0.3) and (6.95,1.4) .. (10.93,3.29)   ;
\draw    (204.33,135.83) -- (180.42,135.83) ;
\draw [shift={(206.33,135.83)}, rotate = 180] [color={rgb, 255:red, 0; green, 0; blue, 0 }  ][line width=0.75]    (10.93,-3.29) .. controls (6.95,-1.4) and (3.31,-0.3) .. (0,0) .. controls (3.31,0.3) and (6.95,1.4) .. (10.93,3.29)   ;
\draw    (361.33,79.83) -- (337.42,79.83) ;
\draw [shift={(363.33,79.83)}, rotate = 180] [color={rgb, 255:red, 0; green, 0; blue, 0 }  ][line width=0.75]    (10.93,-3.29) .. controls (6.95,-1.4) and (3.31,-0.3) .. (0,0) .. controls (3.31,0.3) and (6.95,1.4) .. (10.93,3.29)   ;
\draw    (413.33,79.83) -- (389.42,79.83) ;
\draw [shift={(415.33,79.83)}, rotate = 180] [color={rgb, 255:red, 0; green, 0; blue, 0 }  ][line width=0.75]    (10.93,-3.29) .. controls (6.95,-1.4) and (3.31,-0.3) .. (0,0) .. controls (3.31,0.3) and (6.95,1.4) .. (10.93,3.29)   ;
\draw    (360.33,154.83) -- (336.42,154.83) ;
\draw [shift={(362.33,154.83)}, rotate = 180] [color={rgb, 255:red, 0; green, 0; blue, 0 }  ][line width=0.75]    (10.93,-3.29) .. controls (6.95,-1.4) and (3.31,-0.3) .. (0,0) .. controls (3.31,0.3) and (6.95,1.4) .. (10.93,3.29)   ;
\draw    (412.33,154.83) -- (388.42,154.83) ;
\draw [shift={(414.33,154.83)}, rotate = 180] [color={rgb, 255:red, 0; green, 0; blue, 0 }  ][line width=0.75]    (10.93,-3.29) .. controls (6.95,-1.4) and (3.31,-0.3) .. (0,0) .. controls (3.31,0.3) and (6.95,1.4) .. (10.93,3.29)   ;
\draw    (362.33,203.83) -- (338.42,203.83) ;
\draw [shift={(364.33,203.83)}, rotate = 180] [color={rgb, 255:red, 0; green, 0; blue, 0 }  ][line width=0.75]    (10.93,-3.29) .. controls (6.95,-1.4) and (3.31,-0.3) .. (0,0) .. controls (3.31,0.3) and (6.95,1.4) .. (10.93,3.29)   ;
\draw    (414.33,203.83) -- (390.42,203.83) ;
\draw [shift={(416.33,203.83)}, rotate = 180] [color={rgb, 255:red, 0; green, 0; blue, 0 }  ][line width=0.75]    (10.93,-3.29) .. controls (6.95,-1.4) and (3.31,-0.3) .. (0,0) .. controls (3.31,0.3) and (6.95,1.4) .. (10.93,3.29)   ;
\draw    (361.33,278.83) -- (337.42,278.83) ;
\draw [shift={(363.33,278.83)}, rotate = 180] [color={rgb, 255:red, 0; green, 0; blue, 0 }  ][line width=0.75]    (10.93,-3.29) .. controls (6.95,-1.4) and (3.31,-0.3) .. (0,0) .. controls (3.31,0.3) and (6.95,1.4) .. (10.93,3.29)   ;
\draw    (413.33,278.83) -- (389.42,278.83) ;
\draw [shift={(415.33,278.83)}, rotate = 180] [color={rgb, 255:red, 0; green, 0; blue, 0 }  ][line width=0.75]    (10.93,-3.29) .. controls (6.95,-1.4) and (3.31,-0.3) .. (0,0) .. controls (3.31,0.3) and (6.95,1.4) .. (10.93,3.29)   ;
\draw    (459.97,128) -- (435.33,154.53) ;
\draw [shift={(461.33,126.53)}, rotate = 132.88] [color={rgb, 255:red, 0; green, 0; blue, 0 }  ][line width=0.75]    (10.93,-3.29) .. controls (6.95,-1.4) and (3.31,-0.3) .. (0,0) .. controls (3.31,0.3) and (6.95,1.4) .. (10.93,3.29)   ;
\draw    (463.97,232.07) -- (439.33,205.53) ;
\draw [shift={(465.33,233.53)}, rotate = 227.12] [color={rgb, 255:red, 0; green, 0; blue, 0 }  ][line width=0.75]    (10.93,-3.29) .. controls (6.95,-1.4) and (3.31,-0.3) .. (0,0) .. controls (3.31,0.3) and (6.95,1.4) .. (10.93,3.29)   ;
\draw    (462.97,252) -- (438.33,278.53) ;
\draw [shift={(464.33,250.53)}, rotate = 132.88] [color={rgb, 255:red, 0; green, 0; blue, 0 }  ][line width=0.75]    (10.93,-3.29) .. controls (6.95,-1.4) and (3.31,-0.3) .. (0,0) .. controls (3.31,0.3) and (6.95,1.4) .. (10.93,3.29)   ;

\draw (115.89,83.86) node    {$t_{0}$};
\draw (165.59,83.86) node   {...};
\draw (79.23,83.52) node    {$\mathcal{T}_{\pi } :$};
\draw (78.23,131.52) node    {$\mathcal{T}_{\pi } ':$};
\draw (99.23,193.52) node    {$\mathcal{T}_{\eta _{A}} :$};
\draw (140.56,171.19) node    {$t_{n}$};
\draw (139.89,209.52) node    {$t'_{n}$};
\draw (201.56,190.19) node    {$t_{A}$};
\draw (100.23,262.52) node    {$\mathcal{T}_{\eta _{B}} :$};
\draw (140.56,243.19) node    {$t'_{n}$};
\draw (139.89,281.52) node    {$t_{n}$};
\draw (201.56,262.19) node    {$t_{B}$};
\draw (284.89,111.86) node    {$\mathcal{T}_{\mu _{A}} :$};
\draw (476.56,117.19) node    {$t_{A}$};
\draw (289.55,238.86) node    {$\mathcal{T}_{\mu _{B}} :$};
\draw (215.89,84.86) node    {$t_{n}$};
\draw (117.89,134.86) node    {$t'_{0}$};
\draw (167.59,134.86) node   [align=left] {...};
\draw (217.89,134.86) node    {$t'_{n}$};
\draw (326.89,78.86) node    {$t_{0}$};
\draw (376.59,78.86) node   [align=left] {...};
\draw (426.89,78.86) node    {$t_{n}$};
\draw (325.89,153.86) node    {$t'_{0}$};
\draw (375.59,153.86) node   [align=left] {...};
\draw (425.89,153.86) node    {$t'_{n}$};
\draw (327.89,203) node    {$t_{0}$};
\draw (377.59,203) node   [align=left] {...};
\draw (427.89,203) node    {$t_{n}$};
\draw (326.89,277.86) node    {$t'_{0}$};
\draw (376.59,277.86) node   [align=left] {...};
\draw (426.89,277.86) node    {$t'_{n}$};
\draw (479.56,241.19) node    {$t_{B}$};

\end{tikzpicture}

%% file: pic/delaybox.tex
\tikzset{every picture/.style={line width=0.75pt}} 

\begin{tikzpicture}[x=0.75pt,y=0.75pt,yscale=-1,xscale=1]

\draw  [color={rgb, 255:red, 65; green, 117; blue, 5 }  ,draw opacity=1 ][dash pattern={on 4.5pt off 4.5pt}] (333,35.79) -- (429.73,35.79) -- (429.73,180.74) -- (333,180.74) -- cycle ;
\draw   (191.08,42.42) -- (241.08,42.42) -- (241.08,92.42) -- (191.08,92.42) -- cycle ;
\draw    (139.87,67.83) -- (191.58,67.83) ;
\draw [shift={(137.87,67.83)}, rotate = 0] [color={rgb, 255:red, 0; green, 0; blue, 0 }  ][line width=0.75]    (10.93,-3.29) .. controls (6.95,-1.4) and (3.31,-0.3) .. (0,0) .. controls (3.31,0.3) and (6.95,1.4) .. (10.93,3.29)   ;
\draw    (272.33,57.5) -- (240.75,57.5) ;
\draw [shift={(274.33,57.5)}, rotate = 180] [color={rgb, 255:red, 0; green, 0; blue, 0 }  ][line width=0.75]    (10.93,-3.29) .. controls (6.95,-1.4) and (3.31,-0.3) .. (0,0) .. controls (3.31,0.3) and (6.95,1.4) .. (10.93,3.29)   ;
\draw    (243.06,78.83) -- (302.75,78.83) ;
\draw [shift={(241.06,78.83)}, rotate = 0] [color={rgb, 255:red, 0; green, 0; blue, 0 }  ][line width=0.75]    (10.93,-3.29) .. controls (6.95,-1.4) and (3.31,-0.3) .. (0,0) .. controls (3.31,0.3) and (6.95,1.4) .. (10.93,3.29)   ;
\draw   (191.08,121.75) -- (241.08,121.75) -- (241.08,171.76) -- (191.08,171.76) -- cycle ;
\draw    (139.87,147.17) -- (191.58,147.17) ;
\draw [shift={(137.87,147.17)}, rotate = 0] [color={rgb, 255:red, 0; green, 0; blue, 0 }  ][line width=0.75]    (10.93,-3.29) .. controls (6.95,-1.4) and (3.31,-0.3) .. (0,0) .. controls (3.31,0.3) and (6.95,1.4) .. (10.93,3.29)   ;
\draw    (302.33,158.17) -- (242.75,158.17) ;
\draw [shift={(240.75,158.17)}, rotate = 360] [color={rgb, 255:red, 0; green, 0; blue, 0 }  ][line width=0.75]    (10.93,-3.29) .. controls (6.95,-1.4) and (3.31,-0.3) .. (0,0) .. controls (3.31,0.3) and (6.95,1.4) .. (10.93,3.29)   ;
\draw    (241.73,137.5) -- (270.71,137.5) ;
\draw [shift={(272.71,137.5)}, rotate = 180] [color={rgb, 255:red, 0; green, 0; blue, 0 }  ][line width=0.75]    (10.93,-3.29) .. controls (6.95,-1.4) and (3.31,-0.3) .. (0,0) .. controls (3.31,0.3) and (6.95,1.4) .. (10.93,3.29)   ;
\draw   (93.33,42.82) -- (137.87,42.82) -- (137.87,172.16) -- (93.33,172.16) -- cycle ;
\draw    (37.2,106.83) -- (93.63,106.83) ;
\draw [shift={(35.2,106.83)}, rotate = 0] [color={rgb, 255:red, 0; green, 0; blue, 0 }  ][line width=0.75]    (10.93,-3.29) .. controls (6.95,-1.4) and (3.31,-0.3) .. (0,0) .. controls (3.31,0.3) and (6.95,1.4) .. (10.93,3.29)   ;
\draw  [color={rgb, 255:red, 65; green, 117; blue, 5 }  ,draw opacity=1 ][dash pattern={on 4.5pt off 4.5pt}] (40.53,35.13) -- (257.2,35.13) -- (257.2,180.47) -- (40.53,180.47) -- cycle ;
\draw    (406.11,57.5) -- (274.33,57.5) ;
\draw    (304.75,78.83) -- (358.11,78.83) ;
\draw [shift={(302.75,78.83)}, rotate = 0] [color={rgb, 255:red, 0; green, 0; blue, 0 }  ][line width=0.75]    (10.93,-3.29) .. controls (6.95,-1.4) and (3.31,-0.3) .. (0,0) .. controls (3.31,0.3) and (6.95,1.4) .. (10.93,3.29)   ;
\draw    (406.78,158.17) -- (304.33,158.17) ;
\draw [shift={(302.33,158.17)}, rotate = 360] [color={rgb, 255:red, 0; green, 0; blue, 0 }  ][line width=0.75]    (10.93,-3.29) .. controls (6.95,-1.4) and (3.31,-0.3) .. (0,0) .. controls (3.31,0.3) and (6.95,1.4) .. (10.93,3.29)   ;
\draw    (272.71,137.5) -- (360.33,137.5) ;
\draw   (344.72,93.78) -- (374.72,93.78) -- (374.72,123.78) -- (344.72,123.78) -- cycle ;
\draw   (390.72,93.78) -- (420.72,93.78) -- (420.72,123.78) -- (390.72,123.78) -- cycle ;
\draw    (360.33,137.5) -- (360.33,125.67) ;
\draw [shift={(360.33,123.67)}, rotate = 90] [color={rgb, 255:red, 0; green, 0; blue, 0 }  ][line width=0.75]    (10.93,-3.29) .. controls (6.95,-1.4) and (3.31,-0.3) .. (0,0) .. controls (3.31,0.3) and (6.95,1.4) .. (10.93,3.29)   ;
\draw    (406.29,158.17) -- (406.29,123.67) ;
\draw    (406.11,91.67) -- (406.11,57.5) ;
\draw [shift={(406.11,93.67)}, rotate = 270] [color={rgb, 255:red, 0; green, 0; blue, 0 }  ][line width=0.75]    (10.93,-3.29) .. controls (6.95,-1.4) and (3.31,-0.3) .. (0,0) .. controls (3.31,0.3) and (6.95,1.4) .. (10.93,3.29)   ;
\draw    (358.11,92.66) -- (358.11,78.83) ;

\draw (216.08,67.42) node  [font=\normalsize] [align=left] {$\displaystyle \pi _{1}$};
\draw (216.08,146.75) node  [font=\normalsize] [align=left] {$\displaystyle \pi '_{2}$};
\draw (164.75,56.28) node  [font=\small]  {$( E,t_{n})$};
\draw (164.63,134.95) node  [font=\small]  {$( F,t_{n} ')$};
\draw (117.31,108.86) node    {$\eta _{A}$};
\draw (64.55,92.52) node  [font=\small]  {$( c''_{A} ,t_{3})$};
\draw (150.75,18.86) node    {$\mu _{A}$};
\draw (405.72,108.78) node  [font=\small] [align=left] {$\displaystyle \theta _{f_{1}}$};
\draw (359.72,108.78) node  [font=\small] [align=left] {$\displaystyle \theta _{f_{2}}$};
\draw (387.09,21.33) node    {$\theta $};
\draw (268.75,70.28) node  [font=\small]  {$A$};
\draw (270.09,44.95) node  [font=\small]  {$C$};
\draw (268.75,148.95) node  [font=\small]  {$B$};
\draw (270.09,123.62) node  [font=\small]  {$D$};
\draw (313.75,44.95) node  [font=\small]  {$C'$};
\draw (309.09,148.28) node  [font=\small]  {$C''$};
\draw (309.75,122.95) node  [font=\small]  {$D'$};
\draw (311.09,68.28) node  [font=\small]  {$D''$};

\end{tikzpicture}

%% file: pic/delaybox-order.tex
\tikzset{every picture/.style={line width=0.75pt}} 

\begin{tikzpicture}[x=0.75pt,y=0.75pt,yscale=-1,xscale=1]

\draw    (180.75,128.67) -- (180.75,89.83) ;
\draw [shift={(180.75,130.67)}, rotate = 270] [color={rgb, 255:red, 0; green, 0; blue, 0 }  ][line width=0.75]    (10.93,-3.29) .. controls (6.95,-1.4) and (3.31,-0.3) .. (0,0) .. controls (3.31,0.3) and (6.95,1.4) .. (10.93,3.29)   ;
\draw    (524.08,130.67) -- (524.08,91.83) ;
\draw [shift={(524.08,89.83)}, rotate = 90] [color={rgb, 255:red, 0; green, 0; blue, 0 }  ][line width=0.75]    (10.93,-3.29) .. controls (6.95,-1.4) and (3.31,-0.3) .. (0,0) .. controls (3.31,0.3) and (6.95,1.4) .. (10.93,3.29)   ;
\draw    (408.75,130.01) -- (408.75,91.17) ;
\draw [shift={(408.75,89.17)}, rotate = 90] [color={rgb, 255:red, 0; green, 0; blue, 0 }  ][line width=0.75]    (10.93,-3.29) .. controls (6.95,-1.4) and (3.31,-0.3) .. (0,0) .. controls (3.31,0.3) and (6.95,1.4) .. (10.93,3.29)   ;
\draw    (222.75,128.67) -- (222.75,89.83) ;
\draw [shift={(222.75,130.67)}, rotate = 270] [color={rgb, 255:red, 0; green, 0; blue, 0 }  ][line width=0.75]    (10.93,-3.29) .. controls (6.95,-1.4) and (3.31,-0.3) .. (0,0) .. controls (3.31,0.3) and (6.95,1.4) .. (10.93,3.29)   ;
\draw    (451.42,130.01) -- (451.42,91.17) ;
\draw [shift={(451.42,89.17)}, rotate = 90] [color={rgb, 255:red, 0; green, 0; blue, 0 }  ][line width=0.75]    (10.93,-3.29) .. controls (6.95,-1.4) and (3.31,-0.3) .. (0,0) .. controls (3.31,0.3) and (6.95,1.4) .. (10.93,3.29)   ;
\draw    (295.42,129.15) -- (295.42,90.31) ;
\draw [shift={(295.42,131.15)}, rotate = 270] [color={rgb, 255:red, 0; green, 0; blue, 0 }  ][line width=0.75]    (10.93,-3.29) .. controls (6.95,-1.4) and (3.31,-0.3) .. (0,0) .. controls (3.31,0.3) and (6.95,1.4) .. (10.93,3.29)   ;
\draw    (269,204.33) -- (245.08,204.33) ;
\draw [shift={(271,204.33)}, rotate = 180] [color={rgb, 255:red, 0; green, 0; blue, 0 }  ][line width=0.75]    (10.93,-3.29) .. controls (6.95,-1.4) and (3.31,-0.3) .. (0,0) .. controls (3.31,0.3) and (6.95,1.4) .. (10.93,3.29)   ;
\draw    (503,204.33) -- (479.08,204.33) ;
\draw [shift={(505,204.33)}, rotate = 180] [color={rgb, 255:red, 0; green, 0; blue, 0 }  ][line width=0.75]    (10.93,-3.29) .. controls (6.95,-1.4) and (3.31,-0.3) .. (0,0) .. controls (3.31,0.3) and (6.95,1.4) .. (10.93,3.29)   ;
\draw    (367,204.33) -- (343.08,204.33) ;
\draw [shift={(369,204.33)}, rotate = 180] [color={rgb, 255:red, 0; green, 0; blue, 0 }  ][line width=0.75]    (10.93,-3.29) .. controls (6.95,-1.4) and (3.31,-0.3) .. (0,0) .. controls (3.31,0.3) and (6.95,1.4) .. (10.93,3.29)   ;
\draw    (318,204.33) -- (294.08,204.33) ;
\draw [shift={(320,204.33)}, rotate = 180] [color={rgb, 255:red, 0; green, 0; blue, 0 }  ][line width=0.75]    (10.93,-3.29) .. controls (6.95,-1.4) and (3.31,-0.3) .. (0,0) .. controls (3.31,0.3) and (6.95,1.4) .. (10.93,3.29)   ;
\draw    (420,204.33) -- (396.08,204.33) ;
\draw [shift={(422,204.33)}, rotate = 180] [color={rgb, 255:red, 0; green, 0; blue, 0 }  ][line width=0.75]    (10.93,-3.29) .. controls (6.95,-1.4) and (3.31,-0.3) .. (0,0) .. controls (3.31,0.3) and (6.95,1.4) .. (10.93,3.29)   ;

\draw (181.56,72.52) node    {$t_{0}$};
\draw (258.59,70.17) node   [align=left] {...};
\draw (527.56,72.19) node    {$t_{n}$};
\draw (180.56,147.52) node    {$t'_{0}$};
\draw (258.25,145.17) node   [align=left] {...};
\draw (526.56,147.19) node    {$t'_{n}$};
\draw (411.56,71.86) node    {$t_{1}$};
\draw (410.56,146.86) node    {$t'_{1}$};
\draw (223.56,72.52) node    {$t_{2}$};
\draw (222.56,147.52) node    {$t'_{2}$};
\draw (452.23,71.86) node    {$t_{3}$};
\draw (451.23,146.86) node    {$t'_{3}$};
\draw (488.59,71.5) node   [align=left] {...};
\draw (488.25,146.5) node   [align=left] {...};
\draw (298.89,72.67) node    {$t_{n-1}$};
\draw (297.89,147.67) node    {$t'_{n-1}$};
\draw (124,97.73) node [anchor=north west][inner sep=0.75pt]    {$\mathcal{T}_{\theta _{1}} :$};
\draw (354,103.73) node [anchor=north west][inner sep=0.75pt]    {$\mathcal{T}_{\theta _{2}} :$};
\draw (145.67,196.73) node [anchor=north west][inner sep=0.75pt]    {$\mathcal{T}_{( \pi _{1} \| \pi '_{2}) \theta } :$};
\draw (234.56,204.33) node    {$t_{0}$};
\draw (434.92,204.33) node
{...};
\draw (522.56,204.33) node    {$t_{n}$};
\draw (333.23,204.33) node    {$t'_{1}$};
\draw (283.56,204.33) node    {$t'_{0}$};
\draw (382.89,204.33) node    {$t_{1}$};
\draw (461.56,204.33) node    {$t'_{n}$};

\end{tikzpicture}

%% file: pic/unb-wcf.tex
\tikzset{every picture/.style={line width=0.75pt}} 

\begin{tikzpicture}[x=0.75pt,y=0.75pt,yscale=-1,xscale=1]

\draw   (355.08,64.56) -- (405.08,64.56) -- (405.08,114.56) -- (355.08,114.56) -- cycle ;
\draw    (293.87,89.98) -- (354.84,89.98) ;
\draw [shift={(291.87,89.98)}, rotate = 0] [color={rgb, 255:red, 0; green, 0; blue, 0 }  ][line width=0.75]    (10.93,-3.29) .. controls (6.95,-1.4) and (3.31,-0.3) .. (0,0) .. controls (3.31,0.3) and (6.95,1.4) .. (10.93,3.29)   ;
\draw   (355.08,130.89) -- (405.08,130.89) -- (405.08,180.9) -- (355.08,180.9) -- cycle ;
\draw    (293.87,156.31) -- (354.84,156.31) ;
\draw [shift={(291.87,156.31)}, rotate = 0] [color={rgb, 255:red, 0; green, 0; blue, 0 }  ][line width=0.75]    (10.93,-3.29) .. controls (6.95,-1.4) and (3.31,-0.3) .. (0,0) .. controls (3.31,0.3) and (6.95,1.4) .. (10.93,3.29)   ;
\draw   (247.33,64.96) -- (291.87,64.96) -- (291.87,267.44) -- (247.33,267.44) -- cycle ;
\draw    (211.18,165.17) -- (247.63,165.17) ;
\draw [shift={(209.18,165.17)}, rotate = 0] [color={rgb, 255:red, 0; green, 0; blue, 0 }  ][line width=0.75]    (10.93,-3.29) .. controls (6.95,-1.4) and (3.31,-0.3) .. (0,0) .. controls (3.31,0.3) and (6.95,1.4) .. (10.93,3.29)   ;
\draw   (355.08,217.89) -- (405.08,217.89) -- (405.08,267.9) -- (355.08,267.9) -- cycle ;
\draw    (293.87,243.31) -- (354.84,243.31) ;
\draw [shift={(291.87,243.31)}, rotate = 0] [color={rgb, 255:red, 0; green, 0; blue, 0 }  ][line width=0.75]    (10.93,-3.29) .. controls (6.95,-1.4) and (3.31,-0.3) .. (0,0) .. controls (3.31,0.3) and (6.95,1.4) .. (10.93,3.29)   ;
\draw    (406.18,90.17) -- (435.18,90.17) ;
\draw    (405.18,157.17) -- (434.18,157.17) ;
\draw    (405.18,243.17) -- (434.18,243.17) ;

\draw (380.08,89.56) node  [font=\large] [align=left] {$\displaystyle \pi _{A,1}$};
\draw (380.08,155.89) node  [font=\large] [align=left] {$\displaystyle \pi _{A,2}$};
\draw (269.6,166.2) node  [font=\large]  {$\tau _{A}$};
\draw (323.08,76.23) node  [font=\large] [align=left] {$\displaystyle ( c_{A,1} ,t_{1})$};
\draw (323.41,141.89) node  [font=\large] [align=left] {$\displaystyle ( c_{A,2} ,t_{2})$};
\draw (320.92,42.16) node  [font=\normalsize] [align=left] {Alice};
\draw (229.08,148.69) node  [font=\large] [align=left] {$\displaystyle c_{A}$};
\draw (380.08,242.89) node  [font=\large] [align=left] {$\displaystyle \pi _{A,n}$};
\draw (376.6,201.2) node    {$\cdots $};
\draw (323.41,229.89) node  [font=\large] [align=left] {$\displaystyle ( c_{A,n} ,t_{n})$};
\draw (242.16,288.5) node   [align=left] {(a)};
\draw  [color={rgb, 255:red, 0; green, 0; blue, 0 }  ,draw opacity=1 ][dash pattern={on 4.5pt off 4.5pt}]  (65.1,80.81) .. controls (65.1,71.42) and (72.71,63.81) .. (82.1,63.81) -- (181.1,63.81) .. controls (190.49,63.81) and (198.1,71.42) .. (198.1,80.81) -- (198.1,256.81) .. controls (198.1,266.2) and (190.49,273.81) .. (181.1,273.81) -- (82.1,273.81) .. controls (72.71,273.81) and (65.1,266.2) .. (65.1,256.81) -- cycle  ;
\draw (131.6,168.81) node   [align=left] {\begin{minipage}[lt]{95pt}\setlength\topsep{0pt}
\begin{center}
Program of $\displaystyle \tau _{A}$
\end{center}
$\displaystyle i=1$ {\fontfamily{pcr}\selectfont ;}\\
{\fontfamily{pcr}\selectfont \textbf{While} } $\displaystyle i< n${\fontfamily{pcr}\selectfont}\\
{\fontfamily{pcr}\selectfont \hspace*{1em}\textbf{If} } $\displaystyle c_{A,i} \neq b_{i}$ {\fontfamily{pcr}\selectfont}\\
{\fontfamily{pcr}\selectfont \hspace*{2em}\textbf{Break;} }\\
{\fontfamily{pcr}\selectfont \hspace*{1em}\textbf{Endif}}\\
{\fontfamily{pcr}\selectfont \hspace*{1em}}$\displaystyle i=i+1${\fontfamily{pcr}\selectfont;}\\
\textbf{{\fontfamily{pcr}\selectfont Endwhile}}\\
{\fontfamily{pcr}\selectfont \textbf{Output }} $\displaystyle c_{A} =c_{A,i}$
\end{minipage}};

\end{tikzpicture}

%% file: pic/unb-order.tex
\tikzset{every picture/.style={line width=0.75pt}} 

\begin{tikzpicture}[x=0.75pt,y=0.75pt,yscale=-1,xscale=1]

\draw  [color={rgb, 255:red, 65; green, 117; blue, 5 }  ,draw opacity=1 ][dash pattern={on 4.5pt off 4.5pt}] (284.08,266.18) .. controls (284.08,252.73) and (305.15,241.83) .. (331.13,241.83) .. controls (357.12,241.83) and (378.18,252.73) .. (378.18,266.18) .. controls (378.18,279.63) and (357.12,290.53) .. (331.13,290.53) .. controls (305.15,290.53) and (284.08,279.63) .. (284.08,266.18) -- cycle ;

\draw  [color={rgb, 255:red, 65; green, 117; blue, 5 }  ,draw opacity=1 ][dash pattern={on 4.5pt off 4.5pt}] (284.08,191.51) .. controls (284.08,178.06) and (305.15,167.16) .. (331.13,167.16) .. controls (357.12,167.16) and (378.18,178.06) .. (378.18,191.51) .. controls (378.18,204.96) and (357.12,215.86) .. (331.13,215.86) .. controls (305.15,215.86) and (284.08,204.96) .. (284.08,191.51) -- cycle ;
\draw  [color={rgb, 255:red, 65; green, 117; blue, 5 }  ,draw opacity=1 ][dash pattern={on 4.5pt off 4.5pt}] (284.08,75.32) .. controls (284.08,61.87) and (305.15,50.97) .. (331.13,50.97) .. controls (357.12,50.97) and (378.18,61.87) .. (378.18,75.32) .. controls (378.18,88.77) and (357.12,99.67) .. (331.13,99.67) .. controls (305.15,99.67) and (284.08,88.77) .. (284.08,75.32) -- cycle ;

\draw    (332.46,222.67) -- (332.46,237.54) ;
\draw [shift={(332.46,220.67)}, rotate = 90] [color={rgb, 255:red, 0; green, 0; blue, 0 }  ][line width=0.75]    (10.93,-3.29) .. controls (6.95,-1.4) and (3.31,-0.3) .. (0,0) .. controls (3.31,0.3) and (6.95,1.4) .. (10.93,3.29)   ;
\draw    (331.13,107.86) -- (331.13,122.72) ;
\draw [shift={(331.13,105.86)}, rotate = 90] [color={rgb, 255:red, 0; green, 0; blue, 0 }  ][line width=0.75]    (10.93,-3.29) .. controls (6.95,-1.4) and (3.31,-0.3) .. (0,0) .. controls (3.31,0.3) and (6.95,1.4) .. (10.93,3.29)   ;
\draw    (331.48,146.68) -- (331.48,161.55) ;
\draw [shift={(331.48,144.68)}, rotate = 90] [color={rgb, 255:red, 0; green, 0; blue, 0 }  ][line width=0.75]    (10.93,-3.29) .. controls (6.95,-1.4) and (3.31,-0.3) .. (0,0) .. controls (3.31,0.3) and (6.95,1.4) .. (10.93,3.29)   ;

\draw (391.87,184.42) node [anchor=north west][inner sep=0.75pt]    {$\mathcal{T}_{2}$};
\draw (390.9,259.97) node [anchor=north west][inner sep=0.75pt]    {$\mathcal{T}_{1}$};
\draw (391.86,64.69) node [anchor=north west][inner sep=0.75pt]    {$\mathcal{T}_{n}$};
\draw (331.13,75.32) node   [align=left] {...};
\draw (331.13,266.18) node   [align=left] {...};
\draw (331.13,191.51) node   [align=left] {...};
\draw (330.16,130.54) node   [align=left] {...};
\draw (333.95,305.66) node   [align=left] {(b)};

\end{tikzpicture}

%% file: pic/proof-insecure.tex
\tikzset{every picture/.style={line width=0.75pt}} 

\begin{tikzpicture}[x=0.75pt,y=0.75pt,yscale=-1,xscale=1]

\draw   (127.88,56.82) -- (177.88,56.82) -- (177.88,106.82) -- (127.88,106.82) -- cycle ;
\draw    (76.67,82.23) -- (128.38,82.23) ;
\draw [shift={(74.67,82.23)}, rotate = 0] [color={rgb, 255:red, 0; green, 0; blue, 0 }  ][line width=0.75]    (10.93,-3.29) .. controls (6.95,-1.4) and (3.31,-0.3) .. (0,0) .. controls (3.31,0.3) and (6.95,1.4) .. (10.93,3.29)   ;
\draw    (209.13,71.9) -- (177.55,71.9) ;
\draw [shift={(211.13,71.9)}, rotate = 180] [color={rgb, 255:red, 0; green, 0; blue, 0 }  ][line width=0.75]    (10.93,-3.29) .. controls (6.95,-1.4) and (3.31,-0.3) .. (0,0) .. controls (3.31,0.3) and (6.95,1.4) .. (10.93,3.29)   ;
\draw    (179.86,93.23) -- (239.55,93.23) ;
\draw [shift={(177.86,93.23)}, rotate = 0] [color={rgb, 255:red, 0; green, 0; blue, 0 }  ][line width=0.75]    (10.93,-3.29) .. controls (6.95,-1.4) and (3.31,-0.3) .. (0,0) .. controls (3.31,0.3) and (6.95,1.4) .. (10.93,3.29)   ;
\draw   (127.88,136.15) -- (177.88,136.15) -- (177.88,186.16) -- (127.88,186.16) -- cycle ;
\draw    (76.67,161.57) -- (128.38,161.57) ;
\draw [shift={(74.67,161.57)}, rotate = 0] [color={rgb, 255:red, 0; green, 0; blue, 0 }  ][line width=0.75]    (10.93,-3.29) .. controls (6.95,-1.4) and (3.31,-0.3) .. (0,0) .. controls (3.31,0.3) and (6.95,1.4) .. (10.93,3.29)   ;
\draw    (239.13,172.57) -- (179.55,172.57) ;
\draw [shift={(177.55,172.57)}, rotate = 360] [color={rgb, 255:red, 0; green, 0; blue, 0 }  ][line width=0.75]    (10.93,-3.29) .. controls (6.95,-1.4) and (3.31,-0.3) .. (0,0) .. controls (3.31,0.3) and (6.95,1.4) .. (10.93,3.29)   ;
\draw    (178.53,151.9) -- (207.51,151.9) ;
\draw [shift={(209.51,151.9)}, rotate = 180] [color={rgb, 255:red, 0; green, 0; blue, 0 }  ][line width=0.75]    (10.93,-3.29) .. controls (6.95,-1.4) and (3.31,-0.3) .. (0,0) .. controls (3.31,0.3) and (6.95,1.4) .. (10.93,3.29)   ;
\draw    (342.91,71.9) -- (211.13,71.9) ;
\draw    (241.55,93.23) -- (294.91,93.23) ;
\draw [shift={(239.55,93.23)}, rotate = 0] [color={rgb, 255:red, 0; green, 0; blue, 0 }  ][line width=0.75]    (10.93,-3.29) .. controls (6.95,-1.4) and (3.31,-0.3) .. (0,0) .. controls (3.31,0.3) and (6.95,1.4) .. (10.93,3.29)   ;
\draw    (343.58,172.57) -- (241.13,172.57) ;
\draw [shift={(239.13,172.57)}, rotate = 360] [color={rgb, 255:red, 0; green, 0; blue, 0 }  ][line width=0.75]    (10.93,-3.29) .. controls (6.95,-1.4) and (3.31,-0.3) .. (0,0) .. controls (3.31,0.3) and (6.95,1.4) .. (10.93,3.29)   ;
\draw    (209.51,151.9) -- (297.13,151.9) ;
\draw   (281.52,108.18) -- (311.52,108.18) -- (311.52,138.18) -- (281.52,138.18) -- cycle ;
\draw   (327.52,108.18) -- (357.52,108.18) -- (357.52,138.18) -- (327.52,138.18) -- cycle ;
\draw    (297.13,151.9) -- (297.13,140.07) ;
\draw [shift={(297.13,138.07)}, rotate = 90] [color={rgb, 255:red, 0; green, 0; blue, 0 }  ][line width=0.75]    (10.93,-3.29) .. controls (6.95,-1.4) and (3.31,-0.3) .. (0,0) .. controls (3.31,0.3) and (6.95,1.4) .. (10.93,3.29)   ;
\draw    (343.09,172.57) -- (343.09,138.07) ;
\draw    (342.91,106.07) -- (342.91,71.9) ;
\draw [shift={(342.91,108.07)}, rotate = 270] [color={rgb, 255:red, 0; green, 0; blue, 0 }  ][line width=0.75]    (10.93,-3.29) .. controls (6.95,-1.4) and (3.31,-0.3) .. (0,0) .. controls (3.31,0.3) and (6.95,1.4) .. (10.93,3.29)   ;
\draw    (294.91,107.06) -- (294.91,93.23) ;
\draw  [color={rgb, 255:red, 65; green, 117; blue, 5 }  ,draw opacity=1 ][dash pattern={on 4.5pt off 4.5pt}] (274.6,64.29) -- (371.33,64.29) -- (371.33,181.89) -- (274.6,181.89) -- cycle ;

\draw (152.88,81.82) node  [font=\normalsize] [align=left] {$\displaystyle \pi _{1}$};
\draw (152.88,161.15) node  [font=\normalsize] [align=left] {$\displaystyle \pi '_{2}$};
\draw (101.55,70.68) node  [font=\small]  {$(\textcolor[rgb]{0,0,1}{E ,t_{n}})$};
\draw (101.43,149.35) node  [font=\small]  {$(\textcolor[rgb]{1,0,0}{F ,t_{n} '})$};
\draw (342.52,123.18) node  [font=\small] [align=left] {$\displaystyle \theta _{1}$};
\draw (296.52,123.18) node  [font=\small] [align=left] {$\displaystyle \theta _{2}$};
\draw (205.55,82.68) node  [font=\small,color={rgb, 255:red, 255; green, 6; blue, 2 }  ,opacity=1 ]  {$A$};
\draw (206.89,59.35) node  [font=\small,color={rgb, 255:red, 0; green, 22; blue, 255 }  ,opacity=1 ]  {$C$};
\draw (205.55,163.35) node  [font=\small,color={rgb, 255:red, 0; green, 22; blue, 255 }  ,opacity=1 ]  {$B$};
\draw (206.89,138.02) node  [font=\small,color={rgb, 255:red, 255; green, 6; blue, 2 }  ,opacity=1 ]  {$D$};
\draw (356.95,90.55) node  [font=\small,color={rgb, 255:red, 0; green, 22; blue, 255 }  ,opacity=1 ]  {$C'$};
\draw (358.69,148.28) node  [font=\small,color={rgb, 255:red, 0; green, 22; blue, 255 }  ,opacity=1 ]  {$C''$};
\draw (309.75,147.95) node  [font=\small,color={rgb, 255:red, 255; green, 6; blue, 2 }  ,opacity=1 ]  {$D'$};
\draw (307.09,91.48) node  [font=\small,color={rgb, 255:red, 255; green, 6; blue, 2 }  ,opacity=1 ]  {$D''$};
\draw (323.89,53.82) node    {$\theta $};

\end{tikzpicture}

%% file: pic/proof-order.tex
\tikzset{every picture/.style={line width=0.75pt}} 

\begin{tikzpicture}[x=0.75pt,y=0.75pt,yscale=-1,xscale=1]

\draw    (162.83,117.99) -- (162.83,79.15) ;
\draw [shift={(162.83,119.99)}, rotate = 270] [color={rgb, 255:red, 0; green, 0; blue, 0 }  ][line width=0.75]    (10.93,-3.29) .. controls (6.95,-1.4) and (3.31,-0.3) .. (0,0) .. controls (3.31,0.3) and (6.95,1.4) .. (10.93,3.29)   ;
\draw    (245.57,117.99) -- (245.57,79.15) ;
\draw [shift={(245.57,119.99)}, rotate = 270] [color={rgb, 255:red, 0; green, 0; blue, 0 }  ][line width=0.75]    (10.93,-3.29) .. controls (6.95,-1.4) and (3.31,-0.3) .. (0,0) .. controls (3.31,0.3) and (6.95,1.4) .. (10.93,3.29)   ;
\draw    (345.26,119.99) -- (345.26,81.15) ;
\draw [shift={(345.26,79.15)}, rotate = 90] [color={rgb, 255:red, 0; green, 0; blue, 0 }  ][line width=0.75]    (10.93,-3.29) .. controls (6.95,-1.4) and (3.31,-0.3) .. (0,0) .. controls (3.31,0.3) and (6.95,1.4) .. (10.93,3.29)   ;
\draw    (204.2,119.99) -- (204.2,81.15) ;
\draw [shift={(204.2,79.15)}, rotate = 90] [color={rgb, 255:red, 0; green, 0; blue, 0 }  ][line width=0.75]    (10.93,-3.29) .. controls (6.95,-1.4) and (3.31,-0.3) .. (0,0) .. controls (3.31,0.3) and (6.95,1.4) .. (10.93,3.29)   ;
\draw    (286.93,119.99) -- (286.93,81.15) ;
\draw [shift={(286.93,79.15)}, rotate = 90] [color={rgb, 255:red, 0; green, 0; blue, 0 }  ][line width=0.75]    (10.93,-3.29) .. controls (6.95,-1.4) and (3.31,-0.3) .. (0,0) .. controls (3.31,0.3) and (6.95,1.4) .. (10.93,3.29)   ;
\draw    (191.33,62.35) -- (175.63,62.35) ;
\draw [shift={(193.33,62.35)}, rotate = 180] [color={rgb, 255:red, 0; green, 0; blue, 0 }  ][line width=0.75]    (10.93,-3.29) .. controls (6.95,-1.4) and (3.31,-0.3) .. (0,0) .. controls (3.31,0.3) and (6.95,1.4) .. (10.93,3.29)   ;
\draw    (190.53,136.75) -- (174.83,136.75) ;
\draw [shift={(192.53,136.75)}, rotate = 180] [color={rgb, 255:red, 0; green, 0; blue, 0 }  ][line width=0.75]    (10.93,-3.29) .. controls (6.95,-1.4) and (3.31,-0.3) .. (0,0) .. controls (3.31,0.3) and (6.95,1.4) .. (10.93,3.29)   ;
\draw    (233.53,62.55) -- (217.83,62.55) ;
\draw [shift={(235.53,62.55)}, rotate = 180] [color={rgb, 255:red, 0; green, 0; blue, 0 }  ][line width=0.75]    (10.93,-3.29) .. controls (6.95,-1.4) and (3.31,-0.3) .. (0,0) .. controls (3.31,0.3) and (6.95,1.4) .. (10.93,3.29)   ;
\draw    (232.73,136.95) -- (217.03,136.95) ;
\draw [shift={(234.73,136.95)}, rotate = 180] [color={rgb, 255:red, 0; green, 0; blue, 0 }  ][line width=0.75]    (10.93,-3.29) .. controls (6.95,-1.4) and (3.31,-0.3) .. (0,0) .. controls (3.31,0.3) and (6.95,1.4) .. (10.93,3.29)   ;
\draw    (275.93,61.75) -- (260.23,61.75) ;
\draw [shift={(277.93,61.75)}, rotate = 180] [color={rgb, 255:red, 0; green, 0; blue, 0 }  ][line width=0.75]    (10.93,-3.29) .. controls (6.95,-1.4) and (3.31,-0.3) .. (0,0) .. controls (3.31,0.3) and (6.95,1.4) .. (10.93,3.29)   ;
\draw    (275.13,136.15) -- (259.43,136.15) ;
\draw [shift={(277.13,136.15)}, rotate = 180] [color={rgb, 255:red, 0; green, 0; blue, 0 }  ][line width=0.75]    (10.93,-3.29) .. controls (6.95,-1.4) and (3.31,-0.3) .. (0,0) .. controls (3.31,0.3) and (6.95,1.4) .. (10.93,3.29)   ;

\draw (164.12,61) node  [color={rgb, 255:red, 0; green, 22; blue, 255 }  ,opacity=1 ]  {$t_{0}$};
\draw (316.08,61) node   [align=left] {...};
\draw (162.52,136) node  [color={rgb, 255:red, 0; green, 22; blue, 255 }  ,opacity=1 ]  {$t'_{0}$};
\draw (314.48,136) node   [align=left] {...};
\draw (246.6,61) node  [color={rgb, 255:red, 0; green, 22; blue, 255 }  ,opacity=1 ]  {$t_{2}$};
\draw (245,136) node  [color={rgb, 255:red, 0; green, 22; blue, 255 }  ,opacity=1 ]  {$t'_{2}$};
\draw (350.83,61) node  [color={rgb, 255:red, 255; green, 14; blue, 2 }  ,opacity=1 ]  {$t_{n}$};
\draw (349.23,136) node  [color={rgb, 255:red, 255; green, 14; blue, 2 }  ,opacity=1 ]  {$t'_{n}$};
\draw (205.36,61) node  [color={rgb, 255:red, 255; green, 14; blue, 2 }  ,opacity=1 ]  {$t_{1}$};
\draw (203.76,136) node  [color={rgb, 255:red, 255; green, 14; blue, 2 }  ,opacity=1 ]  {$t'_{1}$};
\draw (287.84,61) node  [color={rgb, 255:red, 255; green, 14; blue, 2 }  ,opacity=1 ]  {$t_{3}$};
\draw (286.24,136) node  [color={rgb, 255:red, 255; green, 14; blue, 2 }  ,opacity=1 ]  {$t'_{3}$};
\draw (76.07,86.39) node [anchor=north west][inner sep=0.75pt]    {$\mathcal{T}_{( \pi _{1} \| \pi '_{2}) \theta } :$};

\end{tikzpicture}

%% file: pic/secure-basic.tex
\tikzset{every picture/.style={line width=0.75pt}} 

\begin{tikzpicture}[x=0.75pt,y=0.75pt,yscale=-1,xscale=1,baseline=(current bounding box.center)]

\draw  [color={rgb, 255:red, 0; green, 0; blue, 0 }  ,draw opacity=1 ] (213,39.96) -- (256.18,39.96) -- (256.18,166.96) -- (213,166.96) -- cycle ;
\draw   (97.08,38.31) -- (147.08,38.31) -- (147.08,88.32) -- (97.08,88.32) -- cycle ;
\draw    (66.18,63.73) -- (97.58,63.73) ;
\draw [shift={(64.18,63.73)}, rotate = 0] [color={rgb, 255:red, 0; green, 0; blue, 0 }  ][line width=0.75]    (10.93,-3.29) .. controls (6.95,-1.4) and (3.31,-0.3) .. (0,0) .. controls (3.31,0.3) and (6.95,1.4) .. (10.93,3.29)   ;
\draw    (171.18,53.39) -- (146.75,53.39) ;
\draw [shift={(173.18,53.39)}, rotate = 180] [color={rgb, 255:red, 0; green, 0; blue, 0 }  ][line width=0.75]    (10.93,-3.29) .. controls (6.95,-1.4) and (3.31,-0.3) .. (0,0) .. controls (3.31,0.3) and (6.95,1.4) .. (10.93,3.29)   ;
\draw    (149.06,74.73) -- (182.75,74.73) ;
\draw [shift={(147.06,74.73)}, rotate = 0] [color={rgb, 255:red, 0; green, 0; blue, 0 }  ][line width=0.75]    (10.93,-3.29) .. controls (6.95,-1.4) and (3.31,-0.3) .. (0,0) .. controls (3.31,0.3) and (6.95,1.4) .. (10.93,3.29)   ;
\draw   (97.08,117.64) -- (147.08,117.64) -- (147.08,167.65) -- (97.08,167.65) -- cycle ;
\draw    (67.18,143.06) -- (97.58,143.06) ;
\draw [shift={(65.18,143.06)}, rotate = 0] [color={rgb, 255:red, 0; green, 0; blue, 0 }  ][line width=0.75]    (10.93,-3.29) .. controls (6.95,-1.4) and (3.31,-0.3) .. (0,0) .. controls (3.31,0.3) and (6.95,1.4) .. (10.93,3.29)   ;
\draw    (179.18,154.06) -- (148.75,154.06) ;
\draw [shift={(146.75,154.06)}, rotate = 360] [color={rgb, 255:red, 0; green, 0; blue, 0 }  ][line width=0.75]    (10.93,-3.29) .. controls (6.95,-1.4) and (3.31,-0.3) .. (0,0) .. controls (3.31,0.3) and (6.95,1.4) .. (10.93,3.29)   ;
\draw    (147.73,133.39) -- (176.71,133.39) ;
\draw [shift={(178.71,133.39)}, rotate = 180] [color={rgb, 255:red, 0; green, 0; blue, 0 }  ][line width=0.75]    (10.93,-3.29) .. controls (6.95,-1.4) and (3.31,-0.3) .. (0,0) .. controls (3.31,0.3) and (6.95,1.4) .. (10.93,3.29)   ;
\draw    (210.18,53.39) -- (173.18,53.39) ;
\draw [shift={(212.18,53.39)}, rotate = 180] [color={rgb, 255:red, 0; green, 0; blue, 0 }  ][line width=0.75]    (10.93,-3.29) .. controls (6.95,-1.4) and (3.31,-0.3) .. (0,0) .. controls (3.31,0.3) and (6.95,1.4) .. (10.93,3.29)   ;
\draw    (184.75,74.73) -- (212.18,74.73) ;
\draw [shift={(182.75,74.73)}, rotate = 0] [color={rgb, 255:red, 0; green, 0; blue, 0 }  ][line width=0.75]    (10.93,-3.29) .. controls (6.95,-1.4) and (3.31,-0.3) .. (0,0) .. controls (3.31,0.3) and (6.95,1.4) .. (10.93,3.29)   ;
\draw    (213.18,154.06) -- (181.18,154.06) ;
\draw [shift={(179.18,154.06)}, rotate = 360] [color={rgb, 255:red, 0; green, 0; blue, 0 }  ][line width=0.75]    (10.93,-3.29) .. controls (6.95,-1.4) and (3.31,-0.3) .. (0,0) .. controls (3.31,0.3) and (6.95,1.4) .. (10.93,3.29)   ;
\draw    (178.18,133.39) -- (210.18,133.39) ;
\draw [shift={(212.18,133.39)}, rotate = 180] [color={rgb, 255:red, 0; green, 0; blue, 0 }  ][line width=0.75]    (10.93,-3.29) .. controls (6.95,-1.4) and (3.31,-0.3) .. (0,0) .. controls (3.31,0.3) and (6.95,1.4) .. (10.93,3.29)   ;

\draw (122.08,63.31) node  [font=\normalsize] [align=left] {$\displaystyle \pi _{A}$};
\draw (122.08,142.65) node  [font=\normalsize] [align=left] {$\displaystyle \eta _{A}$};
\draw (79.75,52.18) node  [font=\small]  {$E$};
\draw (81.63,131.84) node  [font=\small]  {$K,K'$};
\draw (233.59,105.16) node    {$\alpha $};
\draw (162.75,66.18) node  [font=\small]  {$A$};
\draw (164.09,40.84) node  [font=\small]  {$C$};
\draw (162.75,144.84) node  [font=\small]  {$Y$};
\draw (164.09,119.51) node  [font=\small]  {$X$};
\draw (195.75,40.84) node  [font=\small]  {$B$};
\draw (194.09,144.18) node  [font=\small]  {$W$};
\draw (194.75,119.84) node  [font=\small]  {$V$};
\draw (196.09,64.18) node  [font=\small]  {$D$};

\end{tikzpicture}

%% file: pic/secure-equiv.tex
\tikzset{every picture/.style={line width=0.75pt}} 

\begin{tikzpicture}[x=0.75pt,y=0.75pt,yscale=-1,xscale=1,baseline=(current bounding box.center)]

\draw  [color={rgb, 255:red, 0; green, 0; blue, 0 }  ,draw opacity=1 ] (236,206.69) -- (268.18,206.69) -- (268.18,253.69) -- (236,253.69) -- cycle ;
\draw   (120.08,128.31) -- (170.08,128.31) -- (170.08,178.32) -- (120.08,178.32) -- cycle ;
\draw    (89.18,153.73) -- (120.58,153.73) ;
\draw [shift={(87.18,153.73)}, rotate = 0] [color={rgb, 255:red, 0; green, 0; blue, 0 }  ][line width=0.75]    (10.93,-3.29) .. controls (6.95,-1.4) and (3.31,-0.3) .. (0,0) .. controls (3.31,0.3) and (6.95,1.4) .. (10.93,3.29)   ;
\draw    (194.18,143.39) -- (169.75,143.39) ;
\draw [shift={(196.18,143.39)}, rotate = 180] [color={rgb, 255:red, 0; green, 0; blue, 0 }  ][line width=0.75]    (10.93,-3.29) .. controls (6.95,-1.4) and (3.31,-0.3) .. (0,0) .. controls (3.31,0.3) and (6.95,1.4) .. (10.93,3.29)   ;
\draw    (172.06,164.73) -- (205.75,164.73) ;
\draw [shift={(170.06,164.73)}, rotate = 0] [color={rgb, 255:red, 0; green, 0; blue, 0 }  ][line width=0.75]    (10.93,-3.29) .. controls (6.95,-1.4) and (3.31,-0.3) .. (0,0) .. controls (3.31,0.3) and (6.95,1.4) .. (10.93,3.29)   ;
\draw   (120.08,207.64) -- (170.08,207.64) -- (170.08,257.65) -- (120.08,257.65) -- cycle ;
\draw    (90.18,246.06) -- (120.58,246.06) ;
\draw [shift={(88.18,246.06)}, rotate = 0] [color={rgb, 255:red, 0; green, 0; blue, 0 }  ][line width=0.75]    (10.93,-3.29) .. controls (6.95,-1.4) and (3.31,-0.3) .. (0,0) .. controls (3.31,0.3) and (6.95,1.4) .. (10.93,3.29)   ;
\draw    (202.18,244.06) -- (171.75,244.06) ;
\draw [shift={(169.75,244.06)}, rotate = 360] [color={rgb, 255:red, 0; green, 0; blue, 0 }  ][line width=0.75]    (10.93,-3.29) .. controls (6.95,-1.4) and (3.31,-0.3) .. (0,0) .. controls (3.31,0.3) and (6.95,1.4) .. (10.93,3.29)   ;
\draw    (170.73,223.39) -- (199.71,223.39) ;
\draw [shift={(201.71,223.39)}, rotate = 180] [color={rgb, 255:red, 0; green, 0; blue, 0 }  ][line width=0.75]    (10.93,-3.29) .. controls (6.95,-1.4) and (3.31,-0.3) .. (0,0) .. controls (3.31,0.3) and (6.95,1.4) .. (10.93,3.29)   ;
\draw    (233.18,143.39) -- (196.18,143.39) ;
\draw [shift={(235.18,143.39)}, rotate = 180] [color={rgb, 255:red, 0; green, 0; blue, 0 }  ][line width=0.75]    (10.93,-3.29) .. controls (6.95,-1.4) and (3.31,-0.3) .. (0,0) .. controls (3.31,0.3) and (6.95,1.4) .. (10.93,3.29)   ;
\draw    (207.75,164.73) -- (235.18,164.73) ;
\draw [shift={(205.75,164.73)}, rotate = 0] [color={rgb, 255:red, 0; green, 0; blue, 0 }  ][line width=0.75]    (10.93,-3.29) .. controls (6.95,-1.4) and (3.31,-0.3) .. (0,0) .. controls (3.31,0.3) and (6.95,1.4) .. (10.93,3.29)   ;
\draw    (234.76,244.06) -- (204.14,244.06) ;
\draw [shift={(202.14,244.06)}, rotate = 360] [color={rgb, 255:red, 0; green, 0; blue, 0 }  ][line width=0.75]    (10.93,-3.29) .. controls (6.95,-1.4) and (3.31,-0.3) .. (0,0) .. controls (3.31,0.3) and (6.95,1.4) .. (10.93,3.29)   ;
\draw    (201.18,223.39) -- (233.72,223.39) ;
\draw [shift={(235.72,223.39)}, rotate = 180] [color={rgb, 255:red, 0; green, 0; blue, 0 }  ][line width=0.75]    (10.93,-3.29) .. controls (6.95,-1.4) and (3.31,-0.3) .. (0,0) .. controls (3.31,0.3) and (6.95,1.4) .. (10.93,3.29)   ;
\draw    (269.18,232.39) -- (287.18,232.39) ;
\draw    (287.18,232.39) -- (287.18,162.5) ;
\draw  [color={rgb, 255:red, 0; green, 0; blue, 0 }  ,draw opacity=1 ] (236,129.69) -- (268.18,129.69) -- (268.18,176.69) -- (236,176.69) -- cycle ;
\draw    (287.18,162.5) -- (270.18,162.5) ;
\draw [shift={(268.18,162.5)}, rotate = 360] [color={rgb, 255:red, 0; green, 0; blue, 0 }  ][line width=0.75]    (10.93,-3.29) .. controls (6.95,-1.4) and (3.31,-0.3) .. (0,0) .. controls (3.31,0.3) and (6.95,1.4) .. (10.93,3.29)   ;
\draw  [color={rgb, 255:red, 0; green, 0; blue, 0 }  ,draw opacity=1 ] (235,48.48) -- (267.18,48.48) -- (267.18,95.48) -- (235,95.48) -- cycle ;
\draw   (119.08,49.43) -- (169.08,49.43) -- (169.08,99.43) -- (119.08,99.43) -- cycle ;
\draw    (89.18,63.85) -- (119.58,63.85) ;
\draw [shift={(87.18,63.85)}, rotate = 0] [color={rgb, 255:red, 0; green, 0; blue, 0 }  ][line width=0.75]    (10.93,-3.29) .. controls (6.95,-1.4) and (3.31,-0.3) .. (0,0) .. controls (3.31,0.3) and (6.95,1.4) .. (10.93,3.29)   ;
\draw    (201.18,85.85) -- (170.75,85.85) ;
\draw [shift={(168.75,85.85)}, rotate = 360] [color={rgb, 255:red, 0; green, 0; blue, 0 }  ][line width=0.75]    (10.93,-3.29) .. controls (6.95,-1.4) and (3.31,-0.3) .. (0,0) .. controls (3.31,0.3) and (6.95,1.4) .. (10.93,3.29)   ;
\draw    (169.73,65.18) -- (198.71,65.18) ;
\draw [shift={(200.71,65.18)}, rotate = 180] [color={rgb, 255:red, 0; green, 0; blue, 0 }  ][line width=0.75]    (10.93,-3.29) .. controls (6.95,-1.4) and (3.31,-0.3) .. (0,0) .. controls (3.31,0.3) and (6.95,1.4) .. (10.93,3.29)   ;
\draw    (233.76,85.85) -- (203.14,85.85) ;
\draw [shift={(201.14,85.85)}, rotate = 360] [color={rgb, 255:red, 0; green, 0; blue, 0 }  ][line width=0.75]    (10.93,-3.29) .. controls (6.95,-1.4) and (3.31,-0.3) .. (0,0) .. controls (3.31,0.3) and (6.95,1.4) .. (10.93,3.29)   ;
\draw    (200.18,65.18) -- (232.72,65.18) ;
\draw [shift={(234.72,65.18)}, rotate = 180] [color={rgb, 255:red, 0; green, 0; blue, 0 }  ][line width=0.75]    (10.93,-3.29) .. controls (6.95,-1.4) and (3.31,-0.3) .. (0,0) .. controls (3.31,0.3) and (6.95,1.4) .. (10.93,3.29)   ;
\draw    (270.18,148.39) -- (288.18,148.39) ;
\draw    (288.18,148.39) -- (288.18,70.69) ;
\draw    (288.18,70.69) -- (271.18,70.69) ;
\draw [shift={(269.18,70.69)}, rotate = 360] [color={rgb, 255:red, 0; green, 0; blue, 0 }  ][line width=0.75]    (10.93,-3.29) .. controls (6.95,-1.4) and (3.31,-0.3) .. (0,0) .. controls (3.31,0.3) and (6.95,1.4) .. (10.93,3.29)   ;
\draw    (119.18,221.56) -- (63.18,221.56) ;
\draw    (63.18,221.56) -- (63.18,87.69) ;
\draw    (63.18,87.69) -- (116.18,87.69) ;
\draw [shift={(118.18,87.69)}, rotate = 180] [color={rgb, 255:red, 0; green, 0; blue, 0 }  ][line width=0.75]    (10.93,-3.29) .. controls (6.95,-1.4) and (3.31,-0.3) .. (0,0) .. controls (3.31,0.3) and (6.95,1.4) .. (10.93,3.29)   ;

\draw (145.08,153.31) node  [font=\normalsize] [align=left] {$\displaystyle \pi _{A}$};
\draw (145.08,232.65) node  [font=\normalsize] [align=left] {$\displaystyle \eta _{A1}$};
\draw (102.75,142.18) node  [font=\small]  {$E$};
\draw (104.63,234.84) node  [font=\small]  {$K$};
\draw (254.09,230.19) node    {$\alpha _{1}$};
\draw (185.75,156.18) node  [font=\small]  {$A$};
\draw (187.09,130.84) node  [font=\small]  {$C$};
\draw (185.75,234.84) node  [font=\small]  {$Y_{1}$};
\draw (187.09,209.51) node  [font=\small]  {$X_{1}$};
\draw (218.75,130.84) node  [font=\small]  {$B$};
\draw (217.09,234.18) node  [font=\small]  {$W_{1}$};
\draw (217.75,209.84) node  [font=\small]  {$V_{1}$};
\draw (219.09,154.18) node  [font=\small]  {$D$};
\draw (252.09,153.19) node    {$\alpha _{2}$};
\draw (297.75,197.5) node  [font=\small]  {$Q_1$};
\draw (297.75,110) node  [font=\small]  {$Q_2$};
\draw (144.08,74.43) node  [font=\normalsize] [align=left] {$\displaystyle \eta _{A2}$};
\draw (103.63,52.63) node  [font=\small]  {$K '$};
\draw (253.09,71.98) node    {$\alpha_{3}$};
\draw (184.75,76.63) node  [font=\small]  {$Y_{2}$};
\draw (186.09,51.3) node  [font=\small]  {$X_{2}$};
\draw (216.09,75.96) node  [font=\small]  {$W_{2}$};
\draw (216.75,51.63) node  [font=\small]  {$V_{2}$};

\end{tikzpicture}

%% file: pic/secure-order.tex
\tikzset{every picture/.style={line width=0.75pt}} 

\begin{tikzpicture}[x=0.75pt,y=0.75pt,yscale=-1,xscale=1]

\draw    (291,103.5) -- (267.08,103.5) ;
\draw [shift={(293,103.5)}, rotate = 180] [color={rgb, 255:red, 0; green, 0; blue, 0 }  ][line width=0.75]    (10.93,-3.29) .. controls (6.95,-1.4) and (3.31,-0.3) .. (0,0) .. controls (3.31,0.3) and (6.95,1.4) .. (10.93,3.29)   ;
\draw  [color={rgb, 255:red, 65; green, 117; blue, 5 }  ,draw opacity=1 ][dash pattern={on 4.5pt off 4.5pt}] (158,104.56) .. controls (158,89.34) and (179.59,77) .. (206.23,77) .. controls (232.86,77) and (254.45,89.34) .. (254.45,104.56) .. controls (254.45,119.78) and (232.86,132.12) .. (206.23,132.12) .. controls (179.59,132.12) and (158,119.78) .. (158,104.56) -- cycle ;
\draw  [color={rgb, 255:red, 65; green, 117; blue, 5 }  ,draw opacity=1 ][dash pattern={on 4.5pt off 4.5pt}] (307,105.56) .. controls (307,90.34) and (328.59,78) .. (355.23,78) .. controls (381.86,78) and (403.45,90.34) .. (403.45,105.56) .. controls (403.45,120.78) and (381.86,133.12) .. (355.23,133.12) .. controls (328.59,133.12) and (307,120.78) .. (307,105.56) -- cycle ;
\draw    (436,103.5) -- (412.08,103.5) ;
\draw [shift={(438,103.5)}, rotate = 180] [color={rgb, 255:red, 0; green, 0; blue, 0 }  ][line width=0.75]    (10.93,-3.29) .. controls (6.95,-1.4) and (3.31,-0.3) .. (0,0) .. controls (3.31,0.3) and (6.95,1.4) .. (10.93,3.29)   ;
\draw  [color={rgb, 255:red, 65; green, 117; blue, 5 }  ,draw opacity=1 ][dash pattern={on 4.5pt off 4.5pt}] (452,105.56) .. controls (452,90.34) and (473.59,78) .. (500.23,78) .. controls (526.86,78) and (548.45,90.34) .. (548.45,105.56) .. controls (548.45,120.78) and (526.86,133.12) .. (500.23,133.12) .. controls (473.59,133.12) and (452,120.78) .. (452,105.56) -- cycle ;

\draw (206.23,104.56) node   [align=left] {...};
\draw (345.67,53.4) node [anchor=north west][inner sep=0.75pt]    {$\mathcal{T}_{2}$};
\draw (194.67,53.4) node [anchor=north west][inner sep=0.75pt]    {$\mathcal{T}_{1}$};
\draw (355.23,105.56) node   [align=left] {...};
\draw (120.67,95.4) node [anchor=north west][inner sep=0.75pt]    {$\mathcal{T} :$};
\draw (490.67,53.4) node [anchor=north west][inner sep=0.75pt]    {$\mathcal{T}_{3}$};
\draw (500.23,105.56) node   [align=left] {...};

\end{tikzpicture}